\documentclass[10pt,usenames,dvipsnames]{article} 
\usepackage[letterpaper, left=1.2truein, right=1.2truein, top = 1truein, bottom = 1truein]{geometry}

\usepackage{enumerate}
\usepackage{amsmath,amssymb}
\usepackage{natbib}
\usepackage{caption}
\usepackage{enumitem}
\usepackage{placeins}
\usepackage[usenames]{color}
\usepackage{bm}
\usepackage{ying}
\usepackage{multirow}
\usepackage{rotating}
\usepackage{fullpage}
\usepackage{authblk}
\usepackage{parskip}
\usepackage{float}
\usepackage{subcaption}
\usepackage{tikz,tabularx}
\usetikzlibrary{patterns}
\usetikzlibrary{arrows}
\usetikzlibrary{tikzmark, positioning, fit, shapes.misc}

\usetikzlibrary{decorations.pathreplacing, calc}
\tikzset{brace/.style={decorate, decoration={brace}},
  brace mirrored/.style={decorate, decoration={brace,mirror}},
}

\newcolumntype{g}{>{\columncolor{red}}c}

\usepackage{graphicx,amssymb}
\usepackage{xcolor}

\usepackage[colorlinks,
            linkcolor=red,
            anchorcolor=blue,
            citecolor=blue
            ]{hyperref}
\usepackage{algorithm}
\usepackage{algorithmic}

\let\widehat\hat

\def \calib {\textrm{calib}}

\def \test {\textrm{test}}

\theoremstyle{plain}

\allowdisplaybreaks
\usepackage{mathtools}
\mathtoolsset{showonlyrefs}

\usepackage{hyperref}
\def\@#1\@{\begin{align}#1\end{align}}
\def\$#1\${\begin{align*}#1\end{align*}}

\usepackage{color-edits}
\addauthor[Ying]{ying}{magenta}

\usepackage{xcolor}

\title{Online Selective Conformal Prediction with Asymmetric Rules: \\A Permutation Test Approach}
\author{Mingyi Zheng ~and~ Ying Jin\thanks{Department of Statistics and Data Science, University of Pennsylvania. Email: \url{yjinstat@wharton.upenn.edu}. 
Reproduction code for experimental results in the paper can be found in \url{https://github.com/mingyi811/PEMI}.}}
\date{}

\begin{document}

\maketitle

\begin{abstract}     
Selective conformal prediction aims to construct prediction sets with valid coverage for a test unit conditional on it being selected by a data-driven mechanism. While existing methods in the offline setting handle any selection mechanism that is permutation invariant to the labeled data, their extension to the online setting---where data arrives sequentially and later decisions depend on earlier ones---is challenged by the fact that the selection mechanism is naturally asymmetric. As such, existing methods only address a limited collection of selection mechanisms.

In this paper, we propose PErmutation-based Mondrian Conformal Inference (PEMI), a general permutation-based framework for selective conformal prediction with arbitrary asymmetric selection rules. Motivated by full and Mondrian conformal prediction, PEMI identifies all permutations of the observed data (or a Monte-Carlo subset thereof) that lead to the same selection event, and calibrates a prediction set using conformity scores over this selection-preserving reference set. Under standard exchangeability conditions, our prediction sets achieve finite-sample exact selection-conditional coverage for any asymmetric selection mechanism and any prediction model. 
PEMI naturally incorporates additional offline labeled data, extends to selection mechanisms with multiple test samples, and achieves FCR control with fine-grained selection taxonomies. 
We further work out several efficient instantiations for commonly-used online selection rules, including covariate-based rules, conformal p/e-values-based procedures, and selection based on earlier outcomes. 
Finally, we demonstrate the efficacy of our methods across various selection rules on a real drug discovery dataset and investigate their performance via  simulations. 
\end{abstract}
% \keywords{}
% \fontsize{10}{12}\selectfont

%!tex root = main.tex

\section{Introduction}
\label{sec:intro}

Conformal prediction is a general framework for quantifying the uncertainty of black-box prediction models. 
Given a set of labeled (calibration) data, it builds a prediction set  that covers the unknown label of a new (test) sample with a prescribed probability, assuming the test sample is exchangeable with the labeled data~\citep{vovk2005algorithmic}. In many scenarios, however, practitioners may be interested in only a subset of test points, and the decision  to issue a prediction set is itself data-driven. 
For example, a scientist may examine a prediction set for the unknown activity of a drug candidate when its predicted activity is high~\citep{svensson2018conformal}; 
a clinician may seek uncertainty quantification for a medical diagnosis when multiple deep learning models disagree~\citep{loftus2022uncertainty}. 
Such selective issuance breaks the validity of standard conformal prediction, which no longer guarantees coverage for the selected instances~\citep{jin2024confidence}. 

This motivates the problem of \emph{selective conformal prediction}: given any data-driven process that decides whether to query a prediction set, how to construct a prediction set  with valid coverage for the queried sample?  
Formally, assume access to labeled (calibration) data $\{Z_i\}_{i=1}^{t-1}$, where $Z_i=(X_i,Y_i)$ with features $X_i\in \cX$ and labels $Y_i\in \cY$, and a new test unit with observed features $X_t\in \cX$ and  unknown label $Y_t\in \cY$.  
A selection mechanism is a fixed mapping $\cS_t\colon (\cX\times\cY)^{t-1}\times \cX  \to \{0,1\}$ using all the available data to determine whether to issue a prediction set  via the indicator $  {S}_t:=\cS_t(Z_1,\dots,Z_{t-1} ,X_t)$. 
Given a confidence level $\alpha\in (0,1)$, the goal is to construct a prediction set $\hat{\cC}_{\alpha,t}\subseteq \cY$ obeying
\@\label{eq:scc}
\PP ( Y_{t}\in \hat{\cC}_{\alpha,t} \given {S}_t= 1 \,) \geq 1-\alpha, 
\@
where the (conditional) probability is over the joint distribution of all the data points. Following the literature~\citep{jin2024confidence}, we refer to~\eqref{eq:scc} as selection-conditional coverage (SCC).
The only assumption throughout this paper is exchangeability:
\begin{assumption}
    The sequence $(Z_1,\dots,Z_t)$ is exchangeable for every $t\in \NN^+$. 
\end{assumption}

For clarity, we first focus on the online setting with a single test point arriving at time $t$, which has attracted growing interest in selective conformal prediction~\citep{bao2024cap,sale2025online,humbert2025online}. 
We later show that our framework extends to (i) incorporating additional offline data (Section~\ref{subsec:offline_data}) (ii) addressing asymmetric selection rules in the standard offline setting with multiple test samples (Section~\ref{subsec:generalize_multiple_test}), and (iii) achieving control of the false coverage rate (FCR) for certain classes of selection rules (Section~\ref{subsec:generalize_FCR}), all strictly generalizing existing results in the literature.

The selective conformal prediction problem was first studied in the offline setting where the selection mechanism picks a subset of multiple test points~\citep{bao2024selective,jin2024confidence,gazin2025selecting}. 
Prior work recognizes the key challenge that, given selective issuance (i.e., conditional on ${S}_t=1$ as in~\eqref{eq:scc}), the exchangeability between labeled and unlabeled data no longer holds, thus breaking the validity of standard conformal prediction. To address this, existing methods typically construct a ``reference set'' of labeled data that remain exchangeable with the test point conditional on selection, and calibrate the prediction set with this subset only. 
This often relies on certain ``swapping'' technique, which is formally developed in the Joint MOndrian Conformal Inference (JOMI) method of~\cite{jin2024confidence}. Specifically, the reference set is defined as the subset of labeled data which, when \emph{swapped} with a selected test point, would lead to the same selection event. JOMI achieves~\eqref{eq:scc} for any selection rule that is permutation invariant to $\{Z_i\}_{i=1}^{t-1}$.

In the online setting, a fundamental complication is that the selection rule $\cS_t$ is inherently \emph{asymmetric}: the order of the (labeled) calibration data affects the selection trajectory and thus the decision to select at time $t$. 
This breaks the only symmetry assumption needed in the earlier offline methods. To address this challenge, existing solutions   that adapt such a ``swapping'' strategy to online, asymmetric selection rules need to place strong constraints on both the selection rule and how the subset of ``exchangeable'' calibration data is formed~\citep{bao2024cap,sale2025online}. 
As a result, they (i) cover only a narrow class of ``decision-driven'' rules, namely, the decision at time $t$ depends only on past decisions $\{S_i\}_{i=1}^{t-1}$ and current $X_t$, and (ii) can produce vacuous or overly conservative prediction sets.

\subsection{Our approach: selective inference over permutations}

In this paper, we propose \underline{PE}rmutation-based \underline{M}ondrian Conformal \underline{I}nference (PEMI), a general framework for achieving SCC in selective conformal prediction with arbitrary asymmetric selection rules. 
Our key idea is to perform selective inference over \emph{permutations} of the data, rather than over individual calibration data points through pairwise swapping. 
Specifically, for 
a collection $\Pi$ of permutations of the observed data $(Z_1,\dots,Z_t)$, we identify the subset of permutations, called a reference set, that lead to the same selection event at time step $t$.  
We then calibrate the prediction set $\hat{C}_{\alpha,t}$ in the usual way of conformal prediction: a hypothesized label value $y\in \cY$ is included in $\hat{C}_{\alpha,t}$ if and only if the resulting conformity score stays within the normal range defined by the scores under the \emph{selected} permutations. See Figure~\ref{fig:intro} for a visualization.

\begin{figure}
    \centering
    \includegraphics[width=0.8\linewidth]{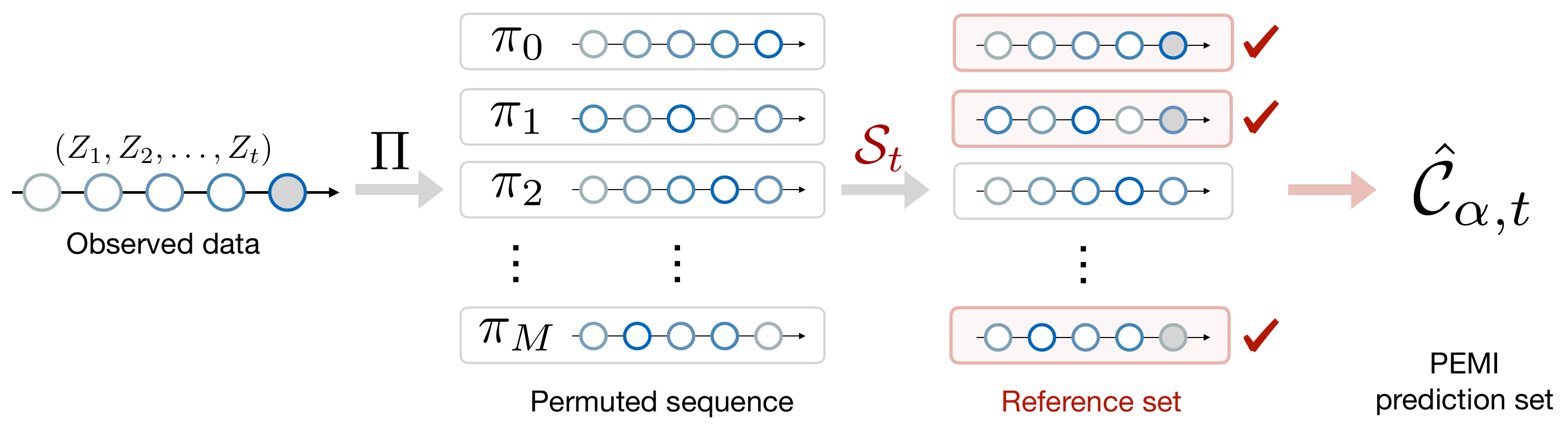}
    \caption{An illustration of the proposed PEMI method. Starting from a set of permutations $\Pi$ of the observed data, we identify the subset which leads to the same selection event after permutation. The final prediction set is then calibrated based on the subset of permutations.}
    \label{fig:intro}
\end{figure}

In Section~\ref{subsec:theory}, we show that two choices of $\Pi$ lead to finite-sample SCC  of PEMI prediction sets: one is when $\Pi$ is  the entire set of permutations over $\{1,\dots,t\}$, and the other is when $\Pi$ is a Monte-Carlo sample, i.e., a random subset from all permutations. 
In both cases, the underlying rationale of PEMI connects to the conditional inference approach in selective inference~\citep{Lee_2016,taylor2016postselectioninferencel1penalizedlikelihood,markovic2017unifying,reid2017post,tibshirani2018uniform}, where inference is based on the distribution of a test statistic conditional on a selection event.  
Here, the underlying random object is the \emph{random} permutation that leads to the observed data. 
Formally, let a ``bag'' of unordered data $[z_1,\dots,z_t]$ denote the unordered set of realized data values, where each $z_i\in \cX\times\cY$ is a fixed value. Under exchangeability, the permutation $\hat\pi$ such that $(Z_1,\dots,Z_t) = (z_{\hat\pi(1)},\dots,z_{\hat\pi(t)})$ follows a uniform distribution over the set of all permutations. 
When $\Pi$ is the set of all permutations, we leverage the fact that $\hat\pi$ follows a uniform distribution over the reference set conditional on the selection event. 
When $\Pi$ is a randomly sampled subset,  the coverage guarantee follows from the exchangeability of $\hat\pi$ with other randomly sampled permutations.

Our key technique can thus be viewed as generalizing the swapping idea in JOMI~\citep{jin2024confidence} to the level of permutations. 
In particular, the JOMI method seeks all the labeled data $Z_i=(X_i,Y_i)$ that, when posited as a test point (i.e., swapped with $Z_t$), would lead to the same selection event. 
Indeed, swapping $(Z_i,Z_t)$ is a special permutation of $(Z_1,\dots,Z_t)$, and the swapping is all that matters when the selection rule is symmetric (Remark~\ref{rm:connection_jomi}). 
However, when the order of labeled data matters, we must reason at the level of permutations to restore exchangeability.
Such an idea builds upon the perspective of full conformal prediction (FCP) as a permutation test; see Section~\ref{sec: warm up} for a detailed discussion.  

\subsection{Preview of results}

The PEMI framework is detailed in Section~\ref{sec:general_method}, where   we develop our main methods and theory for online selection rules concerning one single unlabeled data. Then, we discuss natural extensions to  (i) incorporate offline data, (ii) address asymmetric selection rules for multiple unlabeled data, and (iii) achieve false coverage rate control for certain selection rules. All of these results strictly generalize the existing literature.

Our general framework necessitates computing the reference set of permutations with every imputed label $y\in \cY$ for the $t$-th data point, which can be computationally intractable for continuously-valued responses. 
In Section~\ref{sec:instantiation}, we work out computationally efficient instantiations of PEMI for a variety of online or asymmetric selection rules, including:
\vspace{0.25em}
\begin{enumerate}[label=(\roman*)]
    \item Covariate-based rules, where $\cS_t$ only involves the covariates $\{X_i\}_{i=1}^t$. This covers decision-driven selection rules studied in earlier work~\citep{bao2024cap,sale2025online}, as well as selection based on weighted quantile/average of covariate-based scores, and any black-box optimization programs.
    \item Online multiple testing procedures based on conformal p/e-values  such as~\cite{xu2023onlinemultipletestingevalues}. 
    \item Selection based on weighted quantiles of earlier labels. 
\end{enumerate}
\vspace{0.25em}

In Section~\ref{sec:real}, we demonstrate the validity and efficiency of PEMI for a wide range of practical online selection rules via a drug discovery dataset, where PEMI yields finite-size prediction sets with a handful of labeled data, and achieves valid selection-conditional coverage across all settings. 
In Section~\ref{sec:simu}, we further investigate the behavior of PEMI through simulations and analyze the impact of several factors in the selection rule and conformity score on its performance when compared with vanilla conformal prediction.

%!tex root = main.tex

\subsection{Related work}

This work adds to the literature of selective conformal prediction~\citep{bao2024selective,jin2024confidence,gazin2025selecting,bao2024cap,sale2025online}. 
In the offline setting, the closest to our work is~\citet{jin2024confidence} whose method achieves selection-conditional coverage yet only works for symmetric selection rules.  
Our method strictly generalizes JOMI to asymmetric selection rules, extending the construction of reference sets from the data points to the permutations, and exploiting exchangeability  to achieve selection-conditional coverage. We also show that PEMI reduces to JOMI under symmetric selection rules (Remark~\ref{rm:connection_jomi}).

Many existing methods on selective conformal prediction in the online setting, including  \citet{bao2024cap} and \citet{sale2025online}, rely on swap-based mechanisms to find a ``reference set'' conceptually similar to~\cite{jin2024confidence}. 
Specifically, \citet{bao2024cap} proposes the CAP procedure to construct prediction sets based on an adaptively selected subset of labeled data, but their method only applies to the so-called decision-driven (i.e., $\cS_t$ only involves $\{S_i\}_{i<t}$ and $X_t$) and \emph{symmetric} covariate-dependent selection rules. 
% In contrast, our  framework accommodates arbitrary selection rules that can involve any information observed. 
The recent EXPRESS method~\citep{sale2025online} improves CAP to achieve exact SCC, yet it only applies to the same family of decision-driven selection rules. 
In addition, to ensure validity, EXPRESS imposes stringent requirements when selecting the reference set (of data points), which may result in vacuous prediction sets of infinite length. In contrast, we observe empirically that PEMI reduces the frequency of vacuous prediction sets, thereby improving practicality while expanding the scope of selection rules addressed. 
Finally, an alternative approach is to  extend the adaptive conformal inference (ACI) ideas~\citep{gibbs2021adaptive}  to selected samples, which updates the conformity score threshold based on past performance at selection time points to provide asymptotic SCC and FDR control, or a stronger instantaneous error rate~\citep{bao2024cap,humbert2025online}. While these methods provide asymptotic SCC and FCR control, we target finite-sample exact selection-conditional coverage and leverage the exchangeability structure in the data. 
Our experiments (Section~\ref{sec:real}) show the performance of these methods may be sensitive to hyperparamter choice and subjective to approximation errors, and our method stands out as a robust solution under exchangeability.

Besides SCC, another goal in selective conformal prediction is to control the false coverage rate (FCR)~\citep{weinstein2019onlinecontrolfalsecoverage,xu2023onlinemultipletestingevalues,humbert2025online}. 
The FCR control is addressed in the offline setting for symmetric selection rules in~\cite{jin2024confidence} through a more fine-grained selection taxonomy. 
In the online setting, \cite{bao2024cap} and \cite{humbert2025online} establishes certain asymptotic FCR control, while \cite{sale2025online} shows EXPRESS achieves FCR control for decision-driven selection rules. 
As we show later, our framework can be generalized to achieve FCR control under similar conditions.
Compared with FCR-controlling methods in classical selective inference~\citep{benjamini2005false,weinstein2019onlinecontrolfalsecoverage,xu2024post}, as we address the predictive inference problem, the setup and the methodology are significantly different.

Our work is also closely related to the literature on conditional  inference in classical post-selection inference (POSI)~\citep{Lee_2016,taylor2016postselectioninferencel1penalizedlikelihood,markovic2017unifying,reid2017post,tibshirani2018uniform}. In particular, many existing methods focus on constructing confidence intervals for selected parameters conditional on selection events. Our method is concerned with constructing prediction sets rather than confidence sets for model parameters, and achieves selection-conditional coverage by exploiting the exchangeability structure inherent in the data, instead of distribution of test statistics. Consequently, both the setting and the methodology in our work are fundamentally different from those in POSI.

%!tex root = main.tex
\vspace{-0.5em}
\section{General method}
\label{sec:general_method}

\subsection{Warm-up: full conformal prediction as permutation test}
\label{sec: warm up}

We begin by reviewing the full conformal prediction (FCP) procedure~\citep{vovk2005algorithmic} and interpreting it as a permutation test. This perspective is not entirely new, but helps to warm up our approach~\citep{angelopoulos2025theoreticalfoundationsconformalprediction}. It also offers an extension of FCP to asymmetric conformity scores.

To align with our online notations, we let $n=t-1$ and write $\mathcal D_n = \{(X_i, Y_i)\}_{i=1}^n$ as the labeled data at time $t$. The augmented data is denoted as $\mathcal D_{n+1}^y = (Z_1,\dots,Z_n, Z_{n+1}^y)$ where $Z_i=(X_i,Y_i)$ and $Z_{n+1}^y=(X_{n+1},y)$ for a hypothesized value $y\in \cY$. 
A conformity score is a function $\cV$ that maps a data point $(x,y)$ and a dataset $\cD$ to $\cV((x,y);\cD)\in \RR$. FCP assumes $\cV$ is \emph{symmetric}, that is,  $\cV((x,y);\cD) = \cV((x,y);\cD_\pi)$ for any permutation $\pi$ of the elements in $\cD$ (viewing $\cD$ as an ordered sequence).

For every $i \in [n]$, we define $V_i^y = \mathcal{V}\big((X_i, Y_i); \mathcal D_{n+1}^y\}\big)$, and  $V_{n+1}^y = \mathcal{V}\big((X_{n+1}, y); \mathcal D_{n+1}^y\big).$ Then, the FCP set is given by inverting a conformal p-value:
\begin{equation*}
  \hat{\mathcal C}_{\alpha, n+1}=\{y \in \mathcal{Y}:p_{\text{FCP}}^y > \alpha\},\quad \text{where}\quad   p_{\text{FCP}}^y:=\frac{1+\sum_{i=1}^n \ind\{V_i^y \geq V_{n+1}^y\}}{n+1}.
\end{equation*} 

Intuitively, FCP can be interpreted as imputing a hypothesized value $y$ for $Y_{n+1}$, and assessing whether $(X_{n+1}, y)$ conforms to the observed data $\{(X_i, Y_i)\}_{i=1}^n$. This perspective is closely aligned with the core idea behind the permutation test~\citep{angelopoulos2025theoreticalfoundationsconformalprediction}, which we state next.

Given a test statistic $T\colon \cZ^{n+1}\to \RR$ and (ordered) data $(Z_1,\dots,Z_{n+1})$, the permutation test p-value is  
\begin{equation*}
    p_{\text{perm}}=\frac{\sum_{\pi \in \Pi_{n+1}}\ind\{T(Z_{\pi(1)}, \dots, Z_{\pi(n+1)})\geq T(Z_1, \dots, Z_{n+1}) \}}{(n+1)!},
\end{equation*}
where $\Pi_{n+1}$ is the set of all permutations on $\{1, \dots, n+1\}$. 
Now, consider the test statistic $T(z_1,\dots,z_{n+1})=\cV(z_{n+1};(z_1,\dots,z_{n+1}))$ for any $(z_1,\dots,z_{n+1})\in \cZ^{n+1}$. Due to the symmetry  of $\cV$, for any $\pi\in\Pi$,
\begin{equation*}
    T(Z_{\pi(1)}, \dots, Z_{\pi(n+1)})=\mathcal{V}\big(Z_{\pi(n+1)}; (Z_{\pi(1)}, \dots, Z_{\pi(n+1)}) \big)=\mathcal{V}\big(Z_{\pi(n+1)}; \mathcal D_{n+1}\}\big)=V_{\pi(n+1)}^{Y_{n+1}}.
\end{equation*}
The permutation p-value using this specific statistic $T$ then reduces to 
\$
    p_{\text{perm}}=\frac{\sum_{\pi \in \Pi_{n+1}}\ind\{V_{\pi(n+1)}^{Y_{n+1}}\geq V_{n+1}^{Y_{n+1}} \}}{(n+1)!} &=\frac{\sum_{i=1}^{n+1}n! \cdot \ind\{V_i^{Y_{n+1}} \geq V_{n+1}^{Y_{n+1}}\}}{(n+1)!} \\ 
    &=\frac{1+\sum_{i=1}^n \ind\{V_i^{Y_{n+1}} \geq V_{n+1}^{Y_{n+1}}\}}{n+1}=p_{\text{FCP}}^{Y_{n+1}}.
\$
In other words, the permutation p-value under $T(\cdot)$ is equivalent to the FCP p-value with the ground-truth $Y_{n+1}$ imputed. Thus,   $p_{\text{FCP}}^y$ can be viewed as a permutation test for the exchangeability among $(Z_1,\dots,Z_n,(Z_{n+1},y))$ to determine whether each $y\in\cY$ should be included in the prediction set.

\vspace{-0.25em}
\paragraph{Permutation FCP with asymmetric score.} The same permutation‑test construction remains sensible for non-symmetric conformity scores.
Formally, consider any conformity score $\cV((x,y);\cD)$ that can be sensitive to data ordering in $\cD$.
For any hypothesized value $y\in \cY$, we denote $\cD_{n+1}^y = (Z_1,\dots,Z_n,Z_{n+1}^y)$ and its $i$-th element as $Z_i^y$, as well as $V_i^y = \cV(Z_i^y;\cD_{n+1}^y)$. 
For any permutation $\pi\in\Pi_{n+1}$, we denote $V_{\pi,n+1}^{y}:= \cV(Z_{\pi(n+1)}^y;\cD_{\pi}^y)$, where $\cD_\pi^y = (Z_{\pi(1)}^y,\dots,Z_{\pi(n+1)}^y)$. Then, the above arguments imply 
\@\label{eq:fcp_perm_asym}
\hat\cC_{\alpha,n+1}^{\text{perm}} := \{y\in \cY\colon p_{\text{perm}}^y >\alpha\},\quad \text{where} \quad 
p_{\text{perm}}^y = \frac{\sum_{\pi \in \Pi_{n+1}}\ind\{V_{\pi,n+1}^{y}\geq V_{n+1}^{y} \}}{(n+1)!}
\@
yields finite-sample coverage: $\PP(Y_{n+1}\in \hat\cC_{\alpha,n+1}^{\text{perm}})\geq 1-\alpha$ under exchangeability. Unlike the symmetric case, however, $p_{\text{perm}}^y$ does not necessarily simplify to a rank‑based expression, so this formulation is best viewed as a conceptual device.
This observation motivates PEMI: we ``lift up'' the unit of comparison from individual data to permutations, but now restricted to those permutations that preserve the selection decision.

\subsection{PEMI: General procedure}
\label{sec: procedure}

We are now ready to introduce the PEMI framework, which can be viewed as a permutation test with a restricted reference class. From now on, we return to the notations with time point $t\in \NN^+$.   At time 
$t$, we consider a collection $\Pi_t^\star $ of permutations over $\{1,\dots,t\}$. For a candidate label $y\in \cY$, we perform a permutation-test p-value over a subset of  permutations in $\Pi_t^\star$ that would lead to the same selection decision, and invert it to obtain the prediction set.
% We begin by the version using the set of all permutations over $\{1,\dots,t\}$ which resembles the permutation test in Section~\ref{subsubsec:full_pemi}, followed by a computationally feasible version with randomly sampled permutations in~\ref{sec: random perm}. 

\vspace{-0.5em}
\paragraph{Choice of $\Pi_t^\star$.} 
Following the permutation test perspective, a natural idea is to set $\Pi_t^\star = \Pi_t$, the set of all permutations over $\{1,\dots,t\}$. %We formalize this ``full reference set'' version in Appendix XX. 
For computational feasibility, we focus here on a randomized version in which $\Pi_t^\star$ is an i.i.d.~sample of $M\in \NN^+$ permutations from $\textnormal{Unif}(\Pi_t)$, denoted as 
\$
\Pi_t^{(M)}=\{\pi^{(1)}, \ldots, \pi^{(M)}\},
\$  
and $M\in \NN^+$ is any pre-fixed positive integer. 
Importantly, PEMI's guarantee holds for any fixed $M$.

Throughout Sections~\ref{sec: procedure} and~\ref{subsec:theory}, the selection rule is any mapping $\cS_t\colon \cZ^{t-1}\times \cX\to \{0,1\}$, and the prediction set relies on a user-specified conformity score $\cV_t\colon \cZ^{t}\to \RR$, both sensitive to the ordering of the data in the arguments. For example, one may simply set $\cV_t(Z_1,\dots,Z_t) = |Y_t-\hat\mu(X_t)|$ for a pre-trained prediction model $\hat\mu\colon \cX\to \cY$ and $\cY=\RR$. 
Let $Z_t^y = (X_t, y)$ denote the augmented data for any $y\in \cY$. 
% In the computationally efficient instantiations of PEMI in Section~\ref{sec:instantiation}, we shall focus on conformity scores that is a function of the $t$-th data point only. 

% \subsubsection{PEMI with full permutation reference set}
% \label{subsubsec:full_pemi}

% For generality, here we consider any conformity score $\cV_t \colon \cZ^t \to \RR$ as any function of $t$ (ordered) data points.
% 
% Following the permutation test perspective, we let $\Pi_t^\star$ be a set of permutations over $\{1,\dots,t\}$, which will be specified soon. 
% 
The first step of PEMI is to identify those $\pi \in \Pi_t^\star$ that preserve the selection decisions, {i.e.}, permuting the data leads to the same selection event under $\cS_t$. 
Formally, for any $\pi \in \Pi^\star$, we define the permuted data 
\@\label{eq:permuted_data_t}
    \mathcal{D}_{\pi,t}^y = \left( Z_{\pi(1)}^y, Z_{\pi(2)}^y, \ldots, Z_{\pi(t-1)}^y, X_{\pi(t)} \right),
\@
where for $j \in [t-1]$, we define the imputed data point after permutation 
\begin{equation*}
    Z_{\pi(j)}^y =
\begin{cases}
Z_t^y, &  \pi(j) = t, \\
Z_{\pi(j)}, & \pi(j) \neq t.
\end{cases}
\end{equation*}
Here we slightly override the earlier notation for the augmented dataset, so that now $\cD_{\pi,t}^y$ records the information used to decide the selective issuance at time point $t$ (after permutation).

For any permutation $\pi\in\Pi^\star_t$, we define the selection decision for the permuted data as $S_t^y(\pi)=\cS_t(\cD_{\pi,t}^y)$.  
The reference set of permutations is then set as
\@\label{eq:full_ref_set}
\widehat{R}_t(y; \Pi_t^\star) = \{\pi \in \Pi_t^\star :  S_t^y(\pi) = 1\} \cup \{\pi_0\},
\@
where  
$\pi_0$ denotes the identity map.   
In addition, recall the conformity score for permuted data $V_t^y(\pi)=\mathcal V_t(Z_{\pi(1)}^y, Z_{\pi(2)}^y, \ldots, Z_{\pi(t)}^y)$. 
Finally, the PEMI prediction set is constructed as 
\begin{equation}\label{eq: deter interval}
    \widehat{\mathcal{C}}_{\alpha,t}(\Pi_t^\star)=\big\{y \in \mathcal{Y}:p_t (y; \Pi_t^\star) > \alpha \big\},\quad \text{where}\quad p_t (y; \Pi_t^\star)=\frac{\sum_{\pi \in \widehat{R}_t(y; \Pi_t^\star)}\mathds{1}\{V_t^y(\pi_0) \leq V_t^y(\pi)\}}{\vert \widehat{R}_t(y; \Pi_t^\star) \vert}.
\end{equation}
The p-value in~\eqref{eq: deter interval} can be viewed as a permutation test~\eqref{eq:fcp_perm_asym} with the selected reference set of permutations. As we shall see in our theory in Section~\ref{subsec:theory}, with a proper choice of $\Pi^\star_t$, this adjustment restores the validity of the permutation test conditional on the selection event, thereby leading to selection-conditional validity: this holds  even though the conformity score $\cV_t$ and the selection rule $\cS_t$ are both sensitive to data ordering.

To achieve exact coverage, one may compute the randomized PEMI prediction set. Given i.i.d.~random variables $U_1, \dots, U_t\sim$ Unif$([0,1])$ to break the ties, we set  $\widehat{\mathcal{C}}_{\alpha,t}^{\text{rand}}(\Pi_t^\star)=\{y \in \mathcal{Y}:p_t^\text{rand}(y; \Pi_t^\star) > \alpha\}$, with 
\begin{equation}\label{eq: rand randomized pvalue}
     p_t^\text{rand}(y; \Pi_t^\star) =\frac{\sum_{\pi \in \hat{R}_t(y; \Pi_t^\star)}\mathds{1}\{V_t^y(\pi_0) < V_t^y(\pi)\}  +U_t \cdot \sum_{\pi \in \hat{R}_t(y; \Pi_t^\star)}\mathds{1}\{V_t^y(\pi_0)= V_t^y(\pi)\}  }{|\hat{R}_t(y; \Pi_t^\star)|}.
\end{equation}

\begin{remark}\label{rm:connection_jomi}
Our method, especially the ``full reference set'' version with $\Pi_t^\star = \Pi_t$, is a strict generalization of JOMI~\citep{jin2024confidence}. To see this, we fix a time point $t$, and let $\cS(\cdot)$ be any symmetric selection rule to decide the selection of the $t$-th data point, which  is permutation invariant to the first $t-1$ data points, and a conformity score $V\colon \cZ\to \RR$ that only involves the test point. We again denote the permuted data as in~\eqref{eq:permuted_data_t}.
The JOMI framework identifies for every $y\in \cY$ a subset $\hat{R}_{\text{JOMI}}(y)\subseteq\{1,\dots,t-1\}$ which consists of data points $i\in [t-1]$ such that, when swapping $Z_i$ with $(X_{n+1},y)$, would lead to the same selection event. The prediction set is then defined by $\hat{\cC}_{\alpha,t}^{\text{JOMI}}=\{y\colon V(X_{n+1},y)\leq \text{Quantile}(1-\alpha;\{V(X_i,Y_i)\}_{i\in \hat{R}_{\text{JOMI}}(y)}\cup\{+\infty\})\}$. Indeed, with the symmetric selection function $\cS_t=\cS$, one can show that the PEMI reference set in~\eqref{eq:full_ref_set} is  given by  $\hat{R}_t(y)=\cup_{i\in \hat{R}_{\text{JOMI}}(y) \cup\{t\}}\{\pi\in\Pi\colon \pi(i)=t\}$. With the conformity score $\cV_t(z_1,\dots,z_{t}):=V(z_t)$, the PEMI prediction set thus coincides with the JOMI prediction set. In its full generality, the JOMI framework addresses the offline setting where a subset of multiple test points are selected; we discuss an extension of PEMI to multiple test point in Section~\ref{subsec:generalize_multiple_test}, which again generalizes JOMI in a similar way. 
\end{remark}

\subsection{Theoretical guarantees}
\label{subsec:theory}

Theorem~\ref{thm: general scc} establishes the finite-sample selection-conditional validity of the prediction sets above, without any conditions on the selection mechanism. 

\begin{theorem}\label{thm: general scc}
Suppose the data $\{Z_i\}_{i=1}^t$ are exchangeable. For any selection rule $\cS_t\colon \cZ^{t-1}\times \cX\to \{0,1\}$ and any conformity score $\cV_t\colon \cZ^{t}\to \RR$, and for $\Pi_t^\star \in\{\Pi_t,\Pi_t^{(M)}\}$ stated above, we have 
\@\label{eq:validity_dtm}
\PP\big( Y_t\in \hat{\cC}_{\alpha,t}(\Pi_t^\star) \biggiven S_t=1\big) \geq 1-\alpha,\quad \forall ~ t\geq 1. 
\@
In addition, 
\@\label{eq:validity_rand}
\PP\big( Y_t\in \hat{\cC}_{\alpha,t}^{\textnormal{rand}} (\Pi_t^\star) \biggiven S_t=1\big) = 1-\alpha,\quad \forall ~ t\geq 1. 
\@ 

\end{theorem}

Notably, in the case $\Pi_t^\star = \Pi_t^{(M)}$, the coverage guarantee  holds exactly for any finite number $M\in \NN^+$ instead of relying on asymptotic Monte-Carlo approximation. A large value of $M$ increases the resolution of the p-values yet requires more extensive computation. %In our experiments, we fix $M=1000$. 

We defer the detailed proof to Appendix~\ref{proof: scc} and provide some intuition here. 
Recognizing that 
\$
Y_t \in \hat\cC_{\alpha,t}(\Pi_t^\star)  \quad \Leftrightarrow \quad p_{t}(Y_t; \Pi_t^\star) > \alpha,
\$
the key idea is to show the permutation test with $\hat{R}_t(Y_{t}; \Pi_t^\star)$ is valid conditional on the selection event. The specific techniques slightly differ for the two choices of $\Pi_t^\star$, and we discuss them separately. 

For the validity of the full-reference variant $\widehat{\mathcal{C}}_{\alpha,t}(\Pi_t)$, we prove a stronger result 
\begin{equation*}
    \mathbb{P}\big(p_t(Y_t;\Pi_t) \le \alpha \mid [\mathcal D_t], S_t=1\big)\le \alpha,
\end{equation*}
where we define the unordered bag of full observations $[\mathcal D_t]=[Z_1, \dots, Z_t]$. 
For any realized values of the unordered set as $[d_t] = [z_1, \ldots, z_{t-1}, z_t]$, conditional on $[\mathcal D_t]=[d_t]$,  the only randomness is in the order of $(Z_1, \dots, Z_t)$ as a permutation of $(z_1, \ldots, z_t)$. Due to  exchangeability, one sees that $\hat\pi$, the random permutation such that $(Z_1,\dots,Z_t)=(z_{\hat\pi(1)},\dots,z_{\hat\pi(t)})$, follows a uniform distribution on $\Pi_t$. Then, we show that $\hat\pi\sim \text{Unif}(\hat{R}_t(Y_{t};\Pi_t))$ conditional on $[\cD_t]=[d_t]$ and $S_t=1$ (indeed, our construction ensures that $\hat{R}_t(Y_t;\Pi_t)$ is determined by $[\cD_t]$ only). This leads to (a) in Theorem~\ref{thm: general scc} via the Bayes' rule and tower property.

The validity of $\hat{\cC}_{\alpha,t}(\Pi_t^{(M)})$ and $\hat{\cC}_{\alpha,t}^{\text{rand}}(\Pi_t^{(M)})$ relies on a slightly different structure: the exchangeability between $\hat\pi$ (the permutation such that $(Z_1,\dots,Z_t)=(z_{\hat\pi(1)},\dots,z_{\hat\pi(t)})$) and the randomly sampled $\{\pi^{(m)}\}_{m=1}^M$. 
Due to the independent and uniform sampling of $\{\pi^{(m)}\}_{m=1}^M$ from $\Pi_t$, we are able to show that the permutations $\{ \widehat{\pi},\; \pi^{(1)}\circ \widehat{\pi},\; \ldots,\; \pi^{(M)}\circ \widehat{\pi}\}$ are exchangeable. 
We then leverage this exchangeability and rely on the selection-conditional distribution of $\hat\pi$ among the corresponding reference set of permutations among $\{ \widehat{\pi},\; \pi^{(1)}\circ \widehat{\pi},\; \ldots,\; \pi^{(M)}\circ \widehat{\pi}\}$ to prove their validity in Theorem~\ref{thm: general scc}.  

Having completed the core general framework, the next three subsections are dedicated to several natural generalizations of PEMI. These results help connect to existing methods and show that PEMI strictly generalizes earlier methods. Readers interested in practical implementations of PEMI with concrete asymmetric selection rules can move to Section~\ref{sec:instantiation} without missing key concepts. 

\subsection{Incorporating offline data}
\label{subsec:offline_data}

Several existing works in the online setting incorporate certain offline data, i.e., labeled data that are not involved in any selection decisions~\citep{bao2024cap,sale2025online}. 
While PEMI works without offline data, it handles these settings by simply extending the permutations to the entire data sequence.

Following the literature, we denote the offline data as \(\mathcal{D}_{\mathrm{off}}=\{Z_{-n_{\mathrm{off}}+1},\ldots,Z_{0}\}\), where \(n_{\mathrm{off}}\in \NN^+\) is the number of offline samples. They are collected prior to the online selection procedure.
The online data up to time \(t\) are denoted by \(\mathcal{D}_{t, \mathrm{on}}=\{Z_{1},\ldots, X_{t}\}\) and we denote the complete data as \(\mathcal{D}_{t} = \{Z_{-n_{\mathrm{off}}+1}, \ldots, Z_{t-1}, X_{t}\}\). We assume that data in $\cD_t$ are exchangeable. 
We then define the selection indicator $S_t = \mathcal{S}_t(\mathcal{D}_{t})$,
where $\mathcal{S}_t$ is a fixed mapping $\mathcal{S}_t\colon (\mathcal{X}\times\mathcal{Y})^{n_{\mathrm{off}}+t-1}\times\mathcal{X} 
\to \{0,1\}$. Note that this general setup covers our default setting. Accordingly, in this section, the conformity score $\cV_t\colon (\cX\times\cY)^{n_{\mathrm{off}}+t}\to \RR$ is a pre-specified function that takes the entire data sequence as the arguments.

Here, we extend PEMI via permuting the entire set of online and offline data. For a hypothesized $y\in\cY$ and a permutation $\pi$ of $\{-n_{\mathrm{off}}+1,\dots, t\}$, we denote the permuted dataset $\mathcal{D}_{\pi, t}^y = (Z_{\pi(-n_{\mathrm{off}}+1)}(y), \ldots, X_{\pi(t)})$. 
Similar to before, $Z_{i}(y)=Z_i$ for $i<t$ and $Z_{t}(y)=(X_t,y)$. 
The selection decision under the permutation is $S_t^{y}(\pi)=\cS_t(\mathcal{D}_{\pi, t}^y)$, 
and the conformity score is  $V_t^{y}(\pi)=\cV_t(Z_{\pi(-n_{\mathrm{off}}+1)}(y), \ldots, Z_{\pi(t)}(y))$. Then, we define
\begin{equation}\label{eq: rand dtm offline pvalue}
    p_t^\text{off}(y)=\frac{ \sum_{\pi \in \widehat{R}_t^{\text{off}}(y)}\mathds{1}\{V_t^y(\pi_0) \leq V_t^y(\pi)\}}{ \vert \widehat{R}_t^{\text{off}}(y)\vert},\quad \text{where} \quad \widehat{R}_t^{\text{off}}(y) = \big\{\pi \in \Pi_{t,\text{off}}^{(M)} :  S_t^y(\pi) = 1\big\}\cup\{\pi_0\},
\end{equation}
and $\Pi_{t,\text{off}}^{(M)}=\{\pi^{(1)},\dots,\pi^{(M)}\}$ are permutations  independently and uniformly drawn from the set of permutations over $\{-n_{\mathrm{off}}+1, \dots, t\}$.
Finally, we define our prediction set as  
\begin{equation}\label{eq: rand dtm offline interval}
    \widehat{\mathcal{C}}_{\alpha,t}^{\text{off}}=\{y \in \mathcal{Y}:p_t^\text{off}(y) > \alpha\}.
\end{equation}

\begin{theorem}\label{thm: offline scc}
Suppose that $\{Z_i\}_{i=-n_{\mathrm{off}}+1}^0 \cup \{Z_j\}_{j=1}^t $ are exchangeable. Then $\widehat{\mathcal{C}}_{\alpha,t}^{\text{off}}$ defined in~\eqref{eq: rand dtm offline interval} obeys
\begin{equation}
\label{eq:thm offline dtm}
    \PP( Y_t \in \widehat{\mathcal{C}}_{\alpha,t}^{\text{off}} \given S_t=1) \geq 1-\alpha, \quad \forall ~t\geq 1. 
\end{equation}
\end{theorem}

The detailed proof of Theorem~\ref{thm: offline scc} is in Appendix~\ref{proof: offline scc}. 
This extension is intuitive: when more data is available, we simply apply  the same ideas in Section~\ref{sec:general_method} to the enlarged dataset. Under exchangeability, the same theoretical arguments still go through. Applying the randomization trick in~\eqref{eq: rand randomized pvalue} also leads to exact coverage, which we omit here for brevity.

The practical benefit of including offline data is that a large fraction of permutations may be included in the reference set, which leads to more stable and efficient prediction set. For example, consider all the permutations that keep the online data $\{1,\dots,t-1\}$ invariant and only permutes $\{-n_{\mathrm{off}}+1,\dots,-1,t\}$. Under such permutation, all the selection decisions before time $t$ remain the same, and the permutation would be included in $\hat{R}_t^{\text{off}}(y)$ whenever switching the offline data into position $t$ leaves the selection decision the same. 
Finally, we remark that our framework does not necessarily require the offline data to operate, and this extension is introduced primarily for completeness and to bridge with existing literature. 

\subsection{General selection taxonomy and FCR control}
\label{subsec:generalize_FCR}

Besides the SCC, many existing works also consider the false coverage rate (FCR) when multiple test points are selected. While our discussion so far has been focused on the SCC, in this part, we introduce the concept of selection taxonomy and extend PEMI to provide FCR control in certain situations. 

Similar to~\cite{jin2024confidence}, we define a selection taxonomy $\mathfrak{S} \subseteq \{0,1\}^{t}$ as any pre-specified collection of binary sequences of length $t$. Intuitively, the taxonomy is a set of selection 
trajectories satisfying desired pre-specified properties. 
Given a taxonomy $\mathfrak{S}$, the selection-conditional coverage can be generalized to
\begin{equation}\label{eq:scc_taxonomy}
     \mathbb{P}(Y_t \in \widehat{\mathcal C}_{\alpha, t} \mid S_t = 1, (S_1, \dots, S_t) \in \mathfrak S) \geq 1 - \alpha,
\end{equation}
where $S_i = \cS_i(Z_1,\dots,Z_{i-1},X_i)\in\{0,1\}$ is the selection decision at time $i$. 

To achieve~\eqref{eq:scc_taxonomy}, we construct the reference set that keeps the selection decision in the taxonomy: 
\begin{equation}
\widehat{R}_t(y) = \{\pi \in \Pi_t^{(M)} : S_t^y(\pi) = 1,\ (S_1^y(\pi), \dots, S_t^y(\pi)) \in \mathfrak{S}\} \cup \{\pi_0\}.
\end{equation}
The PEMI prediction set, denoted as $\hat\cC_{\alpha,t}(\mathfrak{S})$, can then be constructed with $\hat{R}_t(y)$ similar to~\eqref{eq: deter interval} or~\eqref{eq: rand randomized pvalue}.  
The proof of~\eqref{eq:scc_taxonomy} is a straightforward extension of the proof of Theorem~\ref{thm: general scc}.

\paragraph{FCR Control.} We now demonstrate that we can achieve FCR control by setting an appropriate taxonomy. The FCR  quantifies the expected proportion of selected prediction sets that fail to cover the true outcome:
\begin{equation}
\label{eq:thm fcr}
    \text{FCR}=\mathbb E\Bigg[\frac{\sum_{t=1}^TS_t \ind\{Y_t \notin \widehat{\mathcal{C}}_{\alpha,t} \}}{1  \lor  \sum_{t=1}^TS_t } \Bigg].
\end{equation}
Here $T\in \NN^+$ is the total number of time points. 

Existing methods achieve the FCR control for the so-called \emph{decision-driven} selection rules, where the selection rules are deterministic functions of past selection decisions~\citep{sale2025online,bao2024cap}. By introducing an appropriate selection taxonomy, PEMI achieves FCR control in similar settings. At each time $t$, we record the observed selection trajectory as $s_1 = \mathcal{S}_1(X_1)$, $s_2 = \mathcal{S}_2(Z_1, X_2)$, $\dots$, $s_t = \mathcal{S}_t(Z_1, \ldots, Z_{t-1}, X_t)$, and define the taxonomy as $\mathfrak{S} := \{(s_1, \ldots, s_t)\}$.  

The following theorem shows that under a conditional independence assumption on future selection decisions, this construction of PEMI prediction set yields FCR control. 

\begin{theorem}\label{thm: fcr}
Suppose for each $t\in\NN^+$, the selection decisions $\{S_{t'}\}_{t'> t}$ are independent of $\ind\{Y_t\in \widehat{\mathcal{C}}_{\alpha,t}(\mathfrak{S})\}$ conditional on $(S_1, \dots, S_{t})$. Then $\{\widehat{\mathcal{C}}_{\alpha,t}(\mathfrak{S})\}_{t=1}^T$ for  $\mathfrak{S} := \{(s_1, \ldots, s_t)\}$ above obeys $\text{FCR}\leq \alpha$. 
\end{theorem}

Theorem~\ref{thm: fcr} is proved in Appendix~\ref{proof: fcr}. The core idea here follows~\citet[Proposition 5.2]{sale2025online}. By limiting the selection event to the taxonomy, the prediction sets $\{\widehat{\mathcal{C}}_{\alpha,t}(\mathfrak{S})\}_{t=1}^T$ achieve  coverage conditional on the entire selection trajectory, which implies FCR control. A natural scenario where the conditional independence assumption holds is in the decision-driven selection rules. In general, however, we remark that exact FCR control without being overly conservative is challenging. With arbitrary selection rules, the decisions and label uncertainty over multiple time points have complicated dependence. We thus leave the problem as an open future direction.

\subsection{Generalization to multiple test samples}
\label{subsec:generalize_multiple_test}

In this subsection, we generalize PEMI to an offline-like setting where multiple test points can be selected. For the ease of clarity, we adopt notations similar to the JOMI framework~\citep{jin2024confidence}.

Assume access to the  calibration data $\mathcal D_{\mathrm{calib}}=(Z_1,\dots,Z_n)$ and test data $\mathcal D_{\mathrm{test}}=(X_{n+1},\dots,X_{n+m})$. 
We consider a general selection mechanism for multiple test samples $\mathcal{S}\colon (\mathcal{X} \times \mathcal{Y})^n \times \mathcal{X}^m \to 2^{\{1,\dots,m\}}$, which can be sensitive to data ordering. 
We also rely on a pre-specified conformity score $\cV\colon (\cX\times\cY)^{n+1}\to \RR$ sensitive to data ordering (e.g., being a function of only the last data point). 
The goal is to construct prediction sets only for the selected samples with selection-conditional coverage. 
% In fact, our major online setting at time $t$ can be viewed as a special case with a single test sample. 
This setting also reduces to that of JOMI when $\cS$ is assumed to be permutation invariant to the first $n$ arguments (calibration data). 

Our method follows the same principle as in Section~\ref{sec: procedure}. Consider the selection set $\widehat{S} = \mathcal{S}(\mathcal{D}_{\mathrm{calib}}, \mathcal{D}_{\mathrm{test}})\subseteq\{1,\dots,m\}$. 
Let $\Pi_{n+j}^{(M)}=(\pi^{(1)}, \dots, \pi^{(M)})$ be permutations randomly and uniformly drawn from the set of all permutations over $\{1, \dots, n, n + j\}$. For each $j \in \widehat{S}$ and any hypothesized value $y$, we define the permuted calibration data $\mathcal{D}
_{\mathrm{calib}}^{\pi}(y)$ and permuted test data $\mathcal{D}_{\mathrm{test}}^{\pi}(y)$ in a similar way:
\begin{align*}
    \mathcal{D}_{\mathrm{calib}}^{\pi}(y)
        &= \left( Z_{\pi(1)}(y), Z_{\pi(2)}(y), \ldots, Z_{\pi(n-1)}(y), Z_{\pi(n)}(y) \right), \\
    \mathcal{D}_{\mathrm{test}}^{\pi}(y)
        &= \left( X_{n+1}, \ldots, X_{n+j-1}, X_{\pi(n+j)}, X_{n+j+1}, \ldots, X_{n+m} \right),
\end{align*}
with $Z_{i}(y)=(X_i,Y_i)$ for $i\leq n$ and $Z_{n+j}(y) = (X_{n+j},y)$.  
Accordingly, the selection set after permutation is $\hat{S}^\pi(y) = \cS(\cD_{\calib}^\pi(y), \cD_{\test}^\pi(y))$, and the conformity score is $V_{n+j}^{y}(\pi)=\cV(Z_{\pi(1)}(y), \ldots, Z_{\pi(n)}(y),Z_{\pi(n+j)}(y))$. We define the reference set 
\begin{equation}
    \widehat R_{n+j}^M(y)=\{\pi \in \Pi_{n+j}^{(M)}: j \in \widehat{S}^\pi(y)\} \cup\{\pi_0\}, 
\end{equation}
and construct the PEMI  prediction set 
\begin{equation}\label{eq: most general prediction set}
    \widehat{\cC}_{\alpha, n+j}=\left\{y \in \mathcal{Y} : p_{n+j}(y) > \alpha \right\},\quad \text{where} \quad p_{n+j}(y)=\frac{ \sum_{\pi \in \widehat R_{n+j}^M(y)}\mathds{1}\{V_{n+j}^y(\pi_0) \leq V_{n+j}^y(\pi)\}}{ \vert \widehat R_{n+j}^M(y) \vert}
\end{equation}

\begin{theorem}\label{thm: most general scc}
Suppose $\{Z_i\}_{i=1}^n \cup \{Z_{n+j}\} $ are exchangeable conditional on $\{X_{n+l}\}_{l \in [m]\backslash\{j\}}$ for any $j \in [m]$. Then $\widehat{\mathcal{C}}_{\alpha,n+j}$ defined in~\eqref{eq: most general prediction set} obeys
\begin{equation}
\label{eq: thm most general dtm}
    \PP( Y_{n+j} \in \widehat{\mathcal{C}}_{\alpha,n+j} \given j \in \widehat{S}) \geq 1-\alpha.
\end{equation}
\end{theorem}

We defer the detailed proof of Theorem~\ref{thm: most general scc} to Appendix~\ref{proof: most general scc}. 
The conditional exchangeability condition is implied by the standard joint exchangeability of $\{Z_i\}_{i=1}^{n+m}$. Similar to the proof of Theorem~\ref{thm: general scc}, we show a stronger result that 
$ 
\mathbb{P}\big(p_{n+j}(Y_{n+j}) \le \alpha \mid j \in \widehat{S},[\mathcal D_j],\mathcal D_j^c, [\hat{\pi}^M]\big)\le \alpha,
$  
where we define the unordered set of full observations $[\mathcal D_j]=[Z_1, \dots, Z_n, Z_{n+j}]$, the remaining test data $\mathcal D_j^c=\mathcal D_{\text{test}}\backslash\{X_{n+j}\}$ and the unordered permutation set $[\hat{\pi}^M]=[\widehat{\pi},\; \pi^{(1)}\circ \widehat{\pi},\; \ldots,\; \pi^{(M)}\circ \widehat{\pi}]$. The general arguments still follow the same ideas, except that we now condition on all the other test points.  

%!tex root = main.tex

\section{Computationally efficient instantiations}
\label{sec:instantiation}

So far, we have established a general framework for predictive inference with selection-conditional coverage for arbitrary selection rules. In classification problems, one can solve for the PEMI prediction set by enumerating all classes $y\in \cY$. However, when \(|\mathcal{Y}|\) is infinite,  such an enumeration is computationally intractable. 

In this section, we present efficient algorithms tailored to a wide range of common selection rules with special structures. We focus on three broad classes of selection rules: covariate-dependent rules, conformal selection, and selection based on earlier outcomes. For efficient computation, throughout, we limit the scope to conformity score functions that involve only the last data point, i.e., $\cV_t(z_1,\dots,z_t)=v(x_t,y_t)$, for a pre-specified function \(v:\mathcal{X}\times\mathcal{Y}\to\mathbb{R}\). We also only focus on the Monte-Carlo  variants $\hat\cC_{\alpha,t}(\Pi_t^{(M)})$ and $\hat\cC_{\alpha,t}^{\text{rand}}(\Pi_t^{(M)})$ since enumeration over the entire permutation set is expensive.

\subsection{Covariate-dependent selection rules}
We first consider a broad class of selection rules known as \textit{covariate-dependent} rules~\citep{jin2024confidence}. In this setting, the selection decision at time $t$ only depends on the covariates  $\{X_i\}_{i=1}^t$ and not the response values $\{Y_i\}_{i=1}^{t-1}$.  
This covers many commonly used selection rules including those in existing work:

\vspace{0.25em}
\begin{enumerate}
\item \textit{Decision-driven selection.} This is the class of rules studied by~\cite{bao2024cap,sale2025online}. Here, the selection $S_t$ is a function of $X_t$ and all the past selection decisions $\{S_i\}_{i=1}^{t-1}$ only. Such rules are covariate-dependent, since they do not involve the responses $\{Y_i\}_{i=1}^{t-1}$. 
    \item \textit{Selection based on weighted quantile.} A unit $t$ is selected if its predicted value is above the weighted quantile of the predicted values (or any scores) of the past $t-1$ data at level $q$. Such ideas have been used to improve robustness of distribution shifts~\citep{barber2023conformalpredictionexchangeability}. In online settings, the weights $\{w_i\}_{i=1}^{t-1}$ may vary over time, and the selection rule is asymmetric to the labeled data.
    \item \textit{Selection based on weighted average.} A unit $t$ is selected if its predicted value is larger than the weighted average of the predicted values of the past $t-1$ data. Similar to the quantile case, the selection rule can be naturally asymmetric.
    \item \textit{Selection based on black-box online optimization.} Many online algorithms produce decisions based on streaming data, which fall within this category if the optimization problem involves covariates only, such as the online Knapsack problems in resource allocation~\citep{chakrabarty2008online}.
\end{enumerate}
\vspace{0.25em}

We now present the procedure applicable to any covariate-dependent selection rule. The key to simplification is that the reference set remains the same for any hypothesized value \(y\in \cY\).  
In words, for covariate-dependent selection rules, the reference set $\hat{R}_t(y;\Pi_t^{(M)})$ can be computed independent of $y$. With a last-point score $v(x_t,y_t)$, the PEMI p-value~\eqref{eq: deter interval} compares $v(X_t,y)$ to $v(X_{\pi(t)},Y_{\pi(t)})$'s in this single reference set, and  the prediction set relies on a single threshold. Since the selection rule does not involve the labels, we denote the selection decision under any imputed response as \(S_t^y(\pi) \equiv S_t(\pi)\). 

The following proposition provides the explicit form of the PEMI prediction set, with proof in Appendix~\ref{proof: prop covariate}.  Algorithm~\ref{alg: covariate dependent jomi online} summarizes the procedure, where we only present $\widehat{\mathcal{C}}_{\alpha, t}(\Pi_t^{(M)})$ for brevity.

\begin{prop}
\label{prop: covariate}
Assume $S_t = 1$.  
Let $\Pi_t^{(M)}$ be i.i.d.~samples from the uniform distribution over $\Pi_t$. 
Define 
\begin{equation}
\label{eq: AtBt_main}
\widehat{R}_t = \{\pi \in \Pi_t^{(M)} \colon S_t(\pi) = 1\} \cup\{\pi_0\}, \qquad
B_t = \{\pi \in \Pi_t^{(M)} \colon S_t(\pi) = 1,\ \pi(t) \neq t\}.
\end{equation}
The PEMI prediction set is   
    $ 
    \widehat{\mathcal{C}}_{\alpha, t}(\Pi_t^{(M)}) = \left\{ y \in \mathcal{Y} : v(X_t, y) \leq \mathrm{Quantile}\left(\beta; \{ v(X_{\pi(t)}, Y_{\pi(t)}) \}_{\pi \in B_t} \cup \{+\infty\} \right) \right\}$ for $\beta = \frac{\lceil (1-\alpha) \cdot |\widehat{R}_t| \rceil}{|B_t|}.$ 
\end{prop}

Here, the set $B_t$ collects the informative permutations that move a labeled point into the last position and yield $v(X_{\pi(t)},Y_{\pi(t)})\neq v(X_t,y)$ appearing in the reduction form of the PEMI p-value. The randomized set $\hat\cC_{\alpha,t}^{\textnormal{rand}}(\Pi_t^{(M)})$ also admits a closed form, which we defer to Proposition~\ref{prop: covariate_full} in Appendix~\ref{app:subsec_cov_full} for brevity.

\begin{algorithm}[t]
\caption{PEMI with Random Permutation for Covariate-dependent Selection}
\label{alg: covariate dependent jomi online}
\begin{algorithmic}[1]
\REQUIRE Data sequence $(Z_1,\dots,Z_{t-1},X_t)$, confidence level $\alpha\in(0,1)$, number of permutations $M$, covariate-dependent selection rule $\mathcal{S}$, conformity score function $v(\cdot, \cdot)$.
% \FOR{$t = 1$ to $N$}
    \STATE Compute $S_t = \mathcal{S}_t(Z_1, \dots, Z_{t-1}, X_t)$;
    \IF{$S_t = 1$} 
        \STATE Randomly sample $M$ permutations $\Pi_t^{(M)}=(\pi^{(1)}, \dots, \pi^{\text{(M)}})$ from \text{Unif}($\Pi_t$); 
         \STATE Compute $\widehat{R}_t$ and $B_t$ as~\eqref{eq: AtBt_main};
        \STATE Compute $\hat\eta = \mathrm{Quantile} \bigl(\frac{\lceil (1-\alpha) \cdot |\widehat{R}_t| \rceil}{|B_t|};\ 
        \{ v(X_{\pi(t)}, Y_{\pi(t)}) \}_{\pi \in B_t}\cup\{+\infty\}\big)$;
         \STATE Construct prediction set $
        \widehat{\mathcal{C}}_{\alpha, t}(\Pi_t^{(M)})
   = \{ y \in \mathcal{Y} :
        v(X_t, y)\leq \hat\eta\}$. 
    \ENDIF
% \ENDFOR
\ENSURE Prediction sets $\widehat{\mathcal{C}}_{\alpha, t}(\Pi_t^{(M)})$.
\end{algorithmic}
\end{algorithm}

\subsection{Conformal selection}
\label{subsec:conformal_selection}

The second class of selection rules are known as \emph{conformal selection}, which aims to select units whose outcomes satisfy certain conditions while controlling certain type-I error~\citep{jin2023selectionpredictionconformalpvalues,jin2025model,xu2023onlinemultipletestingevalues}. 
Such decisions are produced by hypothesis testing procedures based on certain conformal p-values or e-values constructed for each test point.  
With online testing procedures, this class of selection rules is often asymmetric and involves earlier outcomes. %Yet, we are able to establish a simplified implementation due to the structure of the conformal p/e-values. 

We follow the setting of~\citet{jin2023selectionpredictionconformalpvalues,jin2024confidence} with slight modifications to accommodate the online context. At each time $t$, we observe the test point $X_t$ along with a threshold $c_t\in \RR$ of interest. The goal is to select points with $Y_t>c_t$ based on the labeled data $\{(X_i, Y_i, c_i)\}_{i=1}^{t-1}$. 

The selection relies on any score function $F: \mathcal{X} \times \mathbb{R} \to \mathbb{R}$ such that $F(x,y)$ is non-increasing in $y\in \RR$ for any $x \in \mathcal{X}$. Denoting $\widehat{F}_i = F(X_i, c_i)$, we define a general \emph{weighted conformal p-value} as  
\begin{equation}\label{eq: selection pvalue}
p_t^w = \frac{w_t + \sum_{i=1}^{t-1} w_i \cdot \ind\{\widehat{F}_i \geq \widehat{F}_t , Y_i\leq c_i\}}{\sum_{i=1}^t w_i},
\end{equation}
where $\{w_i\}_{i=1}^\infty$ is a pre-fixed sequence of non-negative weights. 
This p-value is a weighted version of the conformal p-value introduced in~\citet{jin2023selectionpredictionconformalpvalues} with the powerful ``clipped'' score~\citep{jin2024confidence}. Following similar arguments in~\citet{barber2023conformalpredictionexchangeability} and~\cite{jin2023selectionpredictionconformalpvalues}, when $\{(X_i,Y_i,c_i)\}_{i=1}^t$ are exchangeable, such a p-value is valid, i.e., 
$\mathbb{P}(p_t^w \leq \alpha,\, Y_t \leq c_t) \leq \alpha$  for all $\alpha \in [0,1]$.  
The conformal selection framework has been applied to identifying promising drug candidates in drug discovery~\citep{bai2025conformal} and reliable language model outputs~\citep{gui2024conformal}, and extended to online settings~\citep{xu2023onlinemultipletestingevalues}, where the need for uncertainty quantification after selection naturally arises. 

We provide the implementations of PEMI for two broad classes of selection rules that (1) threshold the p-value $p_t^w$, and (2) threshold certain e-values constructed based on the conformal p-values. For the conformal selection rules we consider, the selection decision depends on past outcomes only through the binary indicators $\ind\{Y_i\leq c_i\}$. Therefore, after imputing a candidate label 
$y\in \cY$ for time $t$, the induced selection event depends on $y$ only through whether $y\leq c_t$ or $y>c_t$. Consequently, the reference set $\hat{R}(y;\Pi_t^{(M)})$ is constant within each region, leading to two regime-specific reference sets $\hat{R}_t^{(0)}$ and $\hat{R}_t^{(1)}$ for $y> c_t$ and $y\leq c_t$, respectively, and the PEMI prediction set takes a shared form 
\$
\hat\cC_{\alpha,t}(\Pi_t^{(M)}) = \{y>c_t, v(X_t,y)\leq \hat{q}_0\}\cup \{y\leq c_t, v(X_t,y)\leq \hat{q}_1\},
\$
where $\hat{q}_0,\hat{q}_1\in \RR$ are suitable quantiles of the conformity scores in $\hat{R}_t^{(0)}$ and $\hat{R}_t^{(1)}$. 
The remaining of this subsection details how these quantities can be computed for commonly used online testing procedures.

\subsubsection{p-value-based selection rules}

We first consider selection procedures based on thresholding the conformal p-values at certain values. As data arrive sequentially, one can either test each unit individually, or test multiple units across the entire sequence. We thus consider two major types of thresholding methods:
\begin{enumerate}
        \item \emph{Fixed threshold:} Unit $t$ is selected if $p_t^w$ is below a fixed threshold $q \in (0,1)$.
        \item \emph{Adaptive threshold:} Unit $t$ is selected if $p_t^w$ is below an adaptive threshold determined by an online multiple testing algorithm~\citep{javanmard2015onlinecontrolfalsediscovery,ramdas2019saffronadaptivealgorithmonline}.
\end{enumerate}

The second class covers a broad range of online testing methods such as LOND~\citep{javanmard2015onlinecontrolfalsediscovery}, SAFFRON~\citep{ramdas2019saffronadaptivealgorithmonline}, and ADDIS~\citep{tian2019addisadaptivediscardingalgorithm}. These methods often set an adaptive threshold \(\alpha_t\) at each time point to allocate error rates, which is dynamically updated to balance discovery power and error control. 
The thresholds in these methods can be unified as 
\begin{equation}
    \alpha_t = G\left(t;\, \mathcal{F}_{t-1},\, \boldsymbol{\theta}\right),\quad \text{where} \quad \mathcal{F}_{t-1}
=
\sigma\bigl(\{p_i,\, \alpha_i \}_{i=1}^{t-1}\bigr),
\end{equation}
for some function $G(\cdot)$. 
Above, $\mathcal{F}_{t-1}$ is 
the $\sigma$-algebra generated by the history of realized p-values $\{p_i\}_{i=1}^{t-1}$ and adaptive thresholds $\{\alpha_i\}_{i=1}^{t-1}$. 
Finally, $\boldsymbol{\theta}$ collects all the  hyper-parameters such as the weights $\{\gamma_t\}_{t\ge1}$, initial alpha-wealth $W_0$ and threshold parameters in SAFFRON and ADDIS.
% \footnote{We note that these online testing procedures typically work for independent p-values, and thus it is unclear yet whether they provide type-I error control with our weighted p-values. Although we conjecture this might be true due to the known positive dependence structures~\citep{jin2023selectionpredictionconformalpvalues}, a rigorous study is beyond the scope of this work. Instead, one may view them as working procedures one may use in practice to derive the selection decisions.} 

Proposition~\ref{prop: pvalue threshold} provides the explicit form of the PEMI set under this selection rule, where we only present $\widehat{\mathcal{C}}_{\alpha, t}(\Pi_t^{(M)})$ for brevity. The proof is in Appendix~\ref{proof: prop pvalue} and the procedure is summarized in Algorithm~\ref{alg: conformal p-value jomi online}.

\begin{prop}\label{prop: pvalue threshold}
Let $\alpha_t \in (0,1)$ denote the selection threshold at time $t$, which can be either fixed or adaptive as above. Consider a selected time point $t\in \NN^+$ with 
$p_t^w \leq \alpha_t$. 
For any permutation $\pi\in\Pi_t^{(M)}$, we define the permuted score $\hat{F}_{\pi(i)}=F(X_{\pi(i)},c_{\pi(i)})$ for $i\in [t]$. For $k\in\{0,1\}$ and $j \in [t]$, we define %the permuted conformal p-values  
\begin{equation}
    p_{j,\pi}^{w,k}
=
\frac{
  w_j
  + \sum_{1\le i \le j-1,i\neq \pi^{-1}(t)}
    w_i
    \mathds{1}\bigl\{\widehat F_{\pi(i)}\ge\widehat F_{\pi(j)},\;Y_{\pi(i)}\le c_{\pi(i)}\bigr\}
  + k\cdot w_{\pi^{-1}(t)}
    \mathds{1}\bigl\{\widehat F_{t}\ge\widehat F_{\pi(j)} \bigr\} \cdot \ind\{\pi^{-1}(t)< j\}
}{
  \sum_{i=1}^j w_i
},
\end{equation}
and let $\alpha_{t, \pi}^{k}$ be the threshold at time $t$ computed with the permuted p-values and the given online testing procedure.  
Then, it holds that $\widehat{R}_t(y)=\widehat {R}_t^{\ind\{y\leq c_t\}}$ where we define 
\begin{equation}\label{eq:pvalue reference set}
\widehat{R}_t^{k} = \big\{\,\pi\in\Pi_t^{(M)} :\; p_{t,\pi}^{w,k} \leq \alpha_{t, \pi}^{k}\,\big\}\cup \{\pi_0\}, \quad B_t^k = \big\{\pi \in \widehat{R}_t^k : \pi(t) \neq t\big\},\quad  k\in\{0,1\}.
\end{equation} 
Then, the PEMI prediction set is given by 
\begin{equation}\label{eq: pvalue prediction set}
    \widehat{\mathcal{C}}_{\alpha, t}(\Pi_t^{(M)})
=
\left\{\,y\in\mathcal Y: y>c_{t},\,v(X_t,y)\le\widehat{q}_0\,\right\}
\cup
\left\{\,y\in\mathcal Y: y\le c_{t},\,v(X_t,y)\le\widehat{q}_1\,\right\},
\end{equation}
where $\widehat{q}_k
=\mathrm{Quantile} \bigl(\frac{\lceil (1-\alpha)\cdot |\hat{R}_t^k| \rceil}{|B_t^k|};\ 
        \{ v(X_{\pi(t)}, Y_{\pi(t)}) \}_{\pi \in B_t^k}\cup\{+\infty\}\bigr)$  
for $k\in\{0,1\}$. 
\end{prop}

\begin{algorithm}[htbp]
\caption{PEMI with Random Permutation for conformal p-value and e-value}
\label{alg: conformal p-value jomi online}
\begin{algorithmic}[1]
\REQUIRE Data sequence $(Z_1,\dots,Z_{t-1},X_t)$, confidence level $\alpha \in (0,1)$, number of permutations $M$, selection rule $\mathcal{S}$, conformity score function $v(\cdot,\cdot)$, rule $\in \{\text{p-value}, \text{e-value}\}$.
    \STATE Compute $S_t = \mathcal{S}_t(Z_1, \dots, Z_{t-1}, X_t)$;
    \IF{$S_t = 1$}
        \STATE Randomly sample $M$ permutations $\Pi_t^{(M)}=(\pi^{(1)}, \dots, \pi^{\text{(M)}})$ from \text{Unif}($\Pi_t$);
        \STATE Compute the threshold $\alpha_t$;
        \STATE \textbf{if} rule = \texttt{p-value} \textbf{then}
        \STATE \hspace{1em} Compute $\widehat{R}_t^{k}$ and $\widehat{B}_t^{k}$ as~\eqref{eq:pvalue reference set}  $\text{for }k=0,1$;
        \STATE \textbf{if} rule = \texttt{e-value} \textbf{then}
        \STATE \hspace{1em} Compute $\widehat{R}_t^{k}$ and $\widehat{B}_t^{k}$ as in Proposition~\ref{prop: elond threshold} in Appendix~\ref{app:subsec_elond_detail} $\text{for }k=0,1$;
        %\STATE Compute $\widehat{R}_t^{k}$ and $\widehat{B}_t^{k}$ as~\eqref{eq:pvalue reference set} and~\eqref{eq: pvalue Bt} $\text{for }k=0,1$;
        \STATE Compute $\widehat{q}_k
        =\mathrm{Quantile} \Bigl(\frac{\lceil (1-\alpha)\cdot |R_t^k| \rceil}{|B_t^k|};\ 
        \{ v(X_{\pi(t)}, Y_{\pi(t)}) \}_{\pi \in B_t^k}\cup\{+\infty\}\Bigr) 
        \text{for }k=0,1$;
         \STATE Compute $\widehat{\mathcal{C}}_{\alpha, t}(\Pi_t^{(M)})$ as in~\eqref{eq: pvalue prediction set}. 
        % \[
        % \widehat{\mathcal{C}}_{\alpha, t}^\text{mc}=\;\bigl\{\,y\in\mathcal Y:\:y>c_{t},\;v(X_t,y)\le\widehat{q}_0\bigr\}\cup\bigl\{\,y\in\mathcal Y:\;y\le c_{t},\;v(X_t,y)\le\widehat{q}_1\bigr\}
        % \]
    \ENDIF
\ENSURE Prediction set $\widehat{\mathcal{C}}_{\alpha, t}(\Pi_t^{(M)})$.
\end{algorithmic}
\end{algorithm}

\subsubsection{e-value-based selection}
\label{subsec: evalue}

\emph{E-values} provide an alternative to p-value-based hypothesis testing~\citep{vovk2005algorithmic,ramdas2025hypothesis} with recent applications in conformal selection~\citep{xu2023onlinemultipletestingevalues}.  
Here we present PEMI sets for selections from e-value-based procedures, with e-values based on the conformal p-values.

We study the e-LOND algorithm along with the form of e-values proposed in~\citet{xu2023onlinemultipletestingevalues}, which ensure FDR control for online conformal selection, thereby providing a useful tool when test samples arrive sequentially and the selection decisions need to be issued in real-time.  
In addition to the common conformal selection setup, we assume access to $n\in \NN^+$ offline data which will be used in selection, and our PEMI prediction set follows the generalized approach in Section~\ref{subsec:offline_data}.  
At each time \(t\), for every past index \(j\in[t-1]\), we first construct two leave-one-out conformal p-values:
\[
p_{j}^{-}
:=\frac{\sum_{i=-n+1}^{0} \ind\{\widehat F_i \geq\widehat F_{j}, Y_i\leq c_i\}}
        {n+1},
\qquad
p_{j}^{+}
:=\frac{1+\sum_{i=-n+1}^{0} \ind\{\widehat F_i \geq\widehat F_{j}, Y_i\leq c_i\}}
        {n+1}.
\]
We then apply the LOND procedure~\citep{javanmard2015onlinecontrolfalsediscovery} separately to \((p_{j}^{ -})_{j\in[t-1]}\) and \((p_{j}^{+})_{j\in[t-1]}\) to obtain two selection sets
\(\widehat{\mathfrak R}_{t-1}^{-}\) and \(\widehat{\mathfrak R}_{t-1}^{+}\), which yield adaptive levels
$ 
\widehat{\alpha}_t^{\text{LOND},-}
:=\alpha\,\gamma_t (|\widehat{\mathfrak R}_{t-1}^{-}|+1 )
$ and $
\widehat{\alpha}_t^{\text{LOND},+}
:=\alpha\,\gamma_t (|\widehat{\mathfrak R}_{t-1}^{+}|+1).
$ 
The two sets are used only to construct the e-values, defined as
$ 
E_t^{\mathrm{e-LOND}}
:= \ind\! \{p_t\le \widehat{\alpha}_t^{\text{LOND},+} \} / \widehat{\alpha}_t^{\text{LOND},-}$, where %\quad \text{where} \quad 
$p_{t}
:=\frac{1+\sum_{i=-n+1}^{0} \ind\{\widehat F_i \geq\widehat F_{t}, Y_i\leq c_i\}}
        {n+1}$.
Finally, the selection set of e-LOND is determined sequentially: 
for every time point $i$, letting $\widehat{\mathfrak{R}}_{i-1}^\text{e-LOND}$ be the set of time points selected by e-LOND up to time $t-1$, the unit $t$ is selected if and only if %$E_t^{\text{e-LOND}}$ surpasses an adaptive threshold:  
\begin{equation}
 E_t^\text{e-LOND}\geq 1/\alpha_t^\text{e-LOND},\quad \text{where}\quad    \alpha_t^\text{e-LOND}:=\alpha \gamma_t \cdot (\vert \widehat{\mathfrak R}_{t-1}^\text{e-LOND} \vert +1).
\end{equation} 

To compute the PEMI set after e-LOND selection, the idea is still based on the fact that $\hat{R}_t(y;\Pi_t^{(M)})$ depends on $y$ only through whether $y\leq c_t$. The e-LOND selection event can also be rewritten as a p-value thresholding event with a complicated data-dependent cutoff; the calculation can therefore also follow similar ideas. We defer  the exact formulas to Proposition~\ref{prop: elond threshold} in Appendix~\ref{app:subsec_elond_detail} for brevity. 

\subsection{Selection based on earlier outcomes}

The final class of selection rules we consider is based on weighted quantiles of earlier outcomes. 
Specifically, a test point is selected if a user-specified score $\widehat{\mu}(X_t)$ exceeds the weighted quantile of $\{Y_i\}_{i=1}^{t-1}$, where the weights $\{w_i\}_{i=1}^t$ are a fixed sequence and $\hat{\mu}\colon \cX\to \RR$ is a pre-trained score function. 
With proper transformations it addresses selection rules based on weighted quantiles of given functions of the outcomes, such as $Y_i-f(X_i)$, $Y_i\cdot f(X_i)$, and so on, for any function $f(\cdot)$.

Given weights $\{w_i\}_{i<t}$, we denote the weighted quantile of a set of values $\{y_i\}_{i<t}$ as $\text{wQ}(\beta; \{y_i\}_{i=1}^{t-1}) := \inf\{z\in \RR\colon \sum_{i=1}^{t-1}w_i \ind\{y_i\leq z\}/\sum_{i=1}^{t-1}w_i \geq \beta\}$ for any $\beta\in[0,1]$. 
The key to efficient computation is as follows. Partition the range of $y\in \RR$ into intervals via the scores $\{\hat{\mu}_1, \dots, \hat{\mu}_{t-1}\}$. Within each interval, the ordering of $y$ and $\{\hat{\mu}_i\}_{i=1}^{t-1}$ remains the same, so the reference set $\hat{R}_t(y)$ is constant. We can thus compute $\hat{R}_t(y)$ for a representative value of $y$ in each region and merge them to obtain the final prediction set. 

Proposition~\ref{prop: earlier outcome partition} provides the explicit form of the prediction set, where we only present $\widehat{\mathcal{C}}_{\alpha,t}(\Pi_t^{(M)})$ for simplicity. The detailed proof is in Appendix~\ref{proof: prop earlier outcomes}, and Algorithm~\ref{alg: earlier outcomes jomi online} summarizes the procedure.

\begin{prop}\label{prop: earlier outcome partition}
 Suppose the $t$-th data is selected as $\widehat{\mu}(X_t) \geq \textnormal{wQ}(1-\beta; \{Y_i\}_{i=1}^{t-1}). $ We denote the sorted sequence of $\{\hat\mu(X_i)\}_{i=1}^{t-1}$ by $\widehat{\mu}_{(1)} \leq \widehat{\mu}_{(2)} \leq \cdots \leq \widehat{\mu}_{(t-1)}$, which partitions the range of $y$ into open intervals $I_1, \dots, I_t$, where
$ 
I_1 = (-\infty, \widehat{\mu}_{(1)})$, 
$I_j = (\widehat{\mu}_{(j-1)}, \widehat{\mu}_{(j)})$, for $j=2,\dots,t-1$, 
and $I_t = (\widehat{\mu}_{(t-1)}, +\infty)$. 
Let $\Pi_t^{(M)}$ denote the random permutation set over $\{1, \dots, t\}$. Then, the PEMI prediction set at time $t$ is given by
\begin{equation}\label{eq: final set for earlier outcomes}
    \widehat{\mathcal{C}}_{\alpha, t}(\Pi_t^{(M)}) =
    \left( \textstyle{\bigcup_{j=1}^{t}} \widehat{\mathcal{C}}_{\alpha, t}^{(j)} \right) \;\cup\; \widehat{\mathcal{C}}_{\alpha, t}^{\mathrm{bd}},
\end{equation}
Here, we define the boundary prediction subset as 
\begin{equation}\label{eq: prediction set for bd}
    \widehat{\mathcal{C}}_{\alpha, t}^{\mathrm{bd}} = \left\{\, y_k \in \{\widehat{\mu}_{(k)}\}_{k=1}^{t-1} : p_t(y_k;\Pi_t^{(M)}) > \alpha \right\},
\end{equation}
where $p_t(y;\Pi_t^\star)$ is defined in~\eqref{eq: deter interval}. 
For each interval $I_j$, the interval-wise prediction subset is  
\begin{equation}\label{eq: prediction set for open intervals}
    \widehat{\mathcal{C}}_{\alpha, t}^{(j)} = \Big\{\, y \in I_j : v(X_t, y) \leq \mathrm{Quantile} \bigl(\textstyle{\frac{\lceil (1-\alpha)\cdot|\widehat{R}_t^{(j)}| \rceil}{|B_t^{(j)}|}};\ 
        \{ v(X_{\pi(t)}, Y_{\pi(t)}) \}_{\pi \in B_t^{(j)}}\cup\{+\infty\}\bigr) \Big\},
\end{equation}
where $B_t^{(j)} = \{\pi \in \widehat{R}_t^{(j)} :  \pi(t) \neq t\}$, and, denoting we $\widehat{\mu}_{(0)}=-\infty$ and $\widehat{\mu}_{(t)}=+\infty$,
\begin{equation}\label{eq: ref set of open intervals}
    \widehat{R}_t^{(j)} = \big\{\, \pi \in \Pi_t^{(M)} : \widehat{\mu}({X_{\pi(t)}}) \le\widehat{\mu}_{(j-1)},p_\pi^{(1)} \le \alpha \big\}\cup\big\{\, \pi \in \Pi_t^{(M)} : \widehat{\mu}({X_{\pi(t)}}) \ge\widehat{\mu}_{(j)},p_\pi^{(0)} \le \alpha \big\}\cup\{\pi_0\},
\end{equation}
where 
$p_\pi^{(k)}=\frac{\sum_{1\leq i\leq t-1, i\neq \pi^{-1}(t)} w_i\ind\{Y_{\pi(i)}\geq\widehat{\mu}({X_{\pi(t)}})\}+w_{\pi^{-1}(t)}\cdot k}{\sum_{i=1}^{t-1}w_i}$, $k=0,1$. 
\end{prop}

\begin{algorithm}[H]
\caption{PEMI for  Selection Based on Earlier Outcomes}
\label{alg: earlier outcomes jomi online}
\begin{algorithmic}[1]
\REQUIRE Data sequence $(Z_1,\dots,Z_{t-1},X_t)$, confidence level $\alpha\in(0,1)$, number of permutations $M$, selection rule $\mathcal{S}$, conformity score function $v(\cdot, \cdot)$.
    \STATE Compute $S_t = \mathcal{S}_t(Z_1, \dots, Z_{t-1}, X_t)$;
    \IF{$S_t = 1$}
        % \STATE Construct the full permutation set $\Pi_t$ over the indices $\{1, \dots, t\}$;
        \STATE Randomly sample $M$ permutations $\Pi_t^{(M)}=(\pi^{(1)}, \dots, \pi^{\text{(M)}})$ from \text{Unif}($\Pi_t$);
        %\STATE Randomly sample $\Pi_t^M=\{\pi^{(1)}, \dots, \pi^{\text{(M)}}\}$ from $\Pi_t$, the full permutation set over $\{1,\dots,t\}$;
        \STATE Partition the range of $y$ via the sorted sequence $\hat{\mu}_{(1)}  \le \cdots \le \hat{\mu}_{(t-1)}$ into open intervals $I_1, \dots, I_t$;

        \STATE Construct the reference set $R_t^{(j)} $ for $j=1, \dots, t $ as~\eqref{eq: ref set of open intervals};

        \STATE Compute prediction sets $\widehat{\mathcal{C}}_{\alpha,t}^{(j)}$ for $j=1, \dots, t $ and $\widehat{\mathcal{C}}_{\alpha,i}^{\mathrm{bd}}$ as~\eqref{eq: prediction set for open intervals} and~\eqref{eq: prediction set for bd};

        \STATE Merge interval-wise and boundary prediction set to form the final prediction set $\widehat{\mathcal{C}}_{\alpha, t}(\Pi_t^{(M)})$ in~\eqref{eq: final set for earlier outcomes}.
 
    \ENDIF
\ENSURE Prediction set $\widehat{\mathcal{C}}_{\alpha, t}(\Pi_t^{(M)})$.
\end{algorithmic}
\end{algorithm}

%!tex root = main.tex

\section{Real application in drug discovery}
\label{sec:real}

We demonstrate the utility of PEMI through a real drug discovery dataset under a wide range of asymmetric selection rules following Section~\ref{sec:instantiation}. 
In drug discovery, predictive models are widely employed to assist in selecting promising drug candidates. In such high-stakes scenarios, quantifying the uncertainty of predictions is critical for understanding and controlling errors in later developments. 
Real selection processes are often online, with drug candidates evaluated sequentially over time, leading to various asymmetric selection rules.

We focus on a drug-target interactions (DTI) task, where the goal is to construct prediction sets for the unknown (continuously-valued) binding affinities of pairs selected in an online process. Our results show that across various selection rules, PEMI consistently delivers precise selection-conditional coverage. In contrast, vanilla conformal prediction  fails to achieve target coverage in most cases, and in the few cases where the coverage is attained, the size of the resulting prediction sets is exceedingly large.

\vspace{-1em}

\paragraph{Dataset and model.} We utilize the DAVIS dataset~\citep{PMID:22037378}, which contains real-valued binding affinities for 30,060 drug-target pairs. The covariates consist of the encoded structures of drugs and disease targets. We randomly sample 20\% of the dataset to train a three-layer neural network using DeepPurpose library~\citep{Huang_2020}, which serves as our regression model $\widehat{\mu}\colon \mathcal X \to \mathbb R$ used in the residual conformity score $\cV_t(z_1,\dots,z_t)=|y_t-\hat\mu(x_t)|$. The remaining 80\% of the data form the online data sequence. %For each experiment, we randomly sample an appropriate number of instances from this pool as test data.

\vspace{-1em}

\paragraph{Selection rules.} We consider three types of realistic selection rules $\mathcal{S}$ (details to be introduced later):
\begin{enumerate}
    \item \emph{Covariate-dependent selection}: selecting drugs based solely on covariate information.
    \begin{enumerate}
        \item Decision-driven selection rules. When the scientist operates under a fixed budget for investigating drug candidates, the selection strategy may become increasingly stringent as more pairs are chosen, so that later selections are restricted to candidates with higher predicted binding affinity.
        \item Weighted quantile/average. When the scientist prioritizes predictions from more recent time points, the selection strategy can be based on a weighted quantile or average of the previously predicted affinities, where recent drug-target pairs receive higher weights.
        \item Selection based on model uncertainty. The scientist select drug candidates when their predicted affinities disagree between multiple models (high variance), as such disagreement may suggest the need for further experimental investigation.
    \end{enumerate}
    \item \emph{Conformal selection}: selecting drugs based on conformal p-values or e-values.
    \begin{enumerate}
        \item Fixed threshold. When the scientist evaluates each drug-target pair individually, the selection rule can be based on a conformal p-value with a fixed threshold, so that a candidate is pursued only if its conformal p-value falls below a pre-specified level.
        \item E-LOND. When the scientist aims to identify a sequence of promising drug-target pairs, the selection strategy can follow the e-LOND procedure to control the FDR below some $q  \in (0,1)$.
    \end{enumerate}
    \item \emph{Selection based on earlier outcomes}: the scientist selects drugs based on past outcomes, i.e., the current drug-target pair is selected only if its predicted affinity $\widehat{\mu}_t$ exceeds a quantile or weighted quantile of the realized affinities $\{y_s\}_{s<t}$; building upon true labels rather than past predictions may better reflect empirical performance and mitigate  model misspecification.
\end{enumerate}

\vspace{-0.5em}

\paragraph{Evaluation metrics.} The selection-conditional coverage at each time point $t$ is evaluated by a consistent estimator $\widehat{\text{Cov}}_t=\frac{\widehat{P}(Y_t \in \widehat{\mathcal{C}}_{\alpha,t},  S_t=1)}{\widehat{P}(S_t=1)}$, where $\widehat{P}$ is the empirical fraction across repeated experiments. To assess the efficiency, we evaluate at each time point the \emph{median} length of the prediction sets over repeated experiments in which the selection event occurs (the \emph{mean} is inapplicable as infinite-size sets may occur due to small labeled sample size in earlier time steps). Furthermore, we compute at each time point the proportion of selected samples whose sets have infinite length to evaluate the effectiveness of the prediction sets.

\subsection{Covariate-dependent selection}
\subsubsection{Decision-driven selection rules}

We first consider a decision-driven selection rule, which incorporates both past selections and the current covariate. Specifically, we set  $S_t=\ind\{\widehat{\mu}_t \geq \tau_1+\frac{1}{\tau_0}\sum_{i=1}^{t-1} S_i\}$, where $\tau_0=200,\ \tau_1=5.5$. The experimental setting is fully online with $T=200$ and no offline data.
This rule is used in the EXPRESS paper~\citep{sale2025online} to demonstrate validity gap in the CAP method~\citep{bao2024cap}; the latter is thus excluded from the comparison. 
We adopt the confidence level of $\alpha=0.4$ used in~\cite{sale2025online}. The methods under comparison are $\hat{\cC}_{\alpha,t}(\Pi_t^{(M)})$ (\texttt{Ours\_Upper}) and $\hat{\cC}_{\alpha,t}^{\text{rand}}(\Pi_t^{(M)})$ (\texttt{Ours\_Randomize}),  {EXPRESS},  OnlineSCI with three parameter settings (\texttt{OnlineSCI1/2/3}) from~\cite{humbert2025online}, and vanilla conformal prediction (\texttt{Vanilla})  that uses all past time points as calibration data.

Figure~\ref{fig:real decision} presents the empirical selection-conditional coverage, the median prediction set size, and the fraction of prediction sets with infinite length for the four methods across $N=10000$ runs. For selection-conditional coverage, the green curve (vanilla CP) remains well below the target level throughout. The orange curve (EXPRESS), although with theoretical guarantees, is overly conservative because of its stringent requirements on reference set construction (based on data point swapping), and thus also stays far above the target level. As OnlineSCI relies on certain hyper-parameters, it can be unclear a priori how to choose them, and its behavior is sensitive to such choice: OnlineSCI1 (pink) substantially under-covers in the early stage due to an unfavorable initialization; OnlineSCI2 (light blue), with well-chosen parameters, remains close to the target level throughout; and OnlineSCI3 (gray), with an overly small $\gamma$ sequence, is overly conservative and fails to converge. 
In contrast, PEMI (blue and red) consistently stays above and close to the target level.  

Regarding prediction set size, EXPRESS is overly conservative: when $t>180$, the median length reaches infinity (i.e., no reference calibration data can be found as the criterion becomes too stringent), whereas our methods and OnlineSCI remain reasonably efficient. 
The pattern in the last panel is revealing: the fraction of infinite-length sets for EXPRESS increases over time as its requirements on the reference set become more stringent. 
In contrast, with a moderate number of labeled points such as 40, PEMI always leads to non-vacuous prediction sets. 
This shows the statistical benefit of building on permutation tests rather than data swapping.

\begin{figure} 
  \centering
  \includegraphics[width=1.0\linewidth]{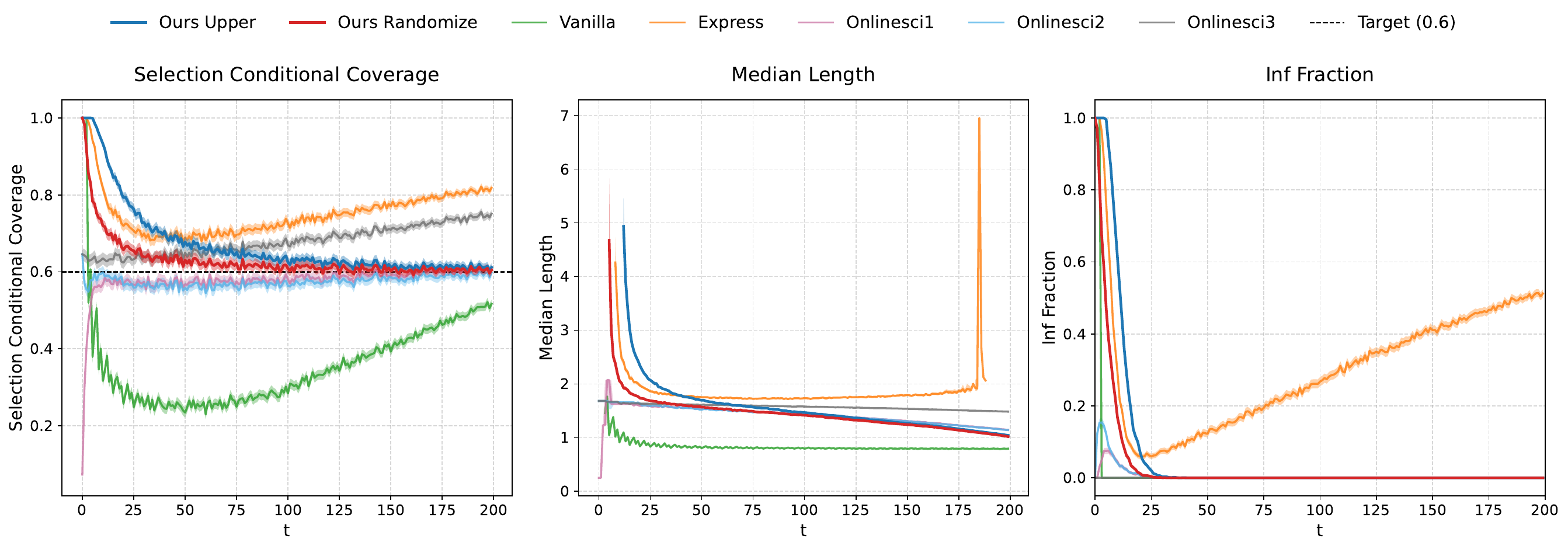}
  \caption{{\small Selection-conditional coverage (left), median prediction set size (middle), and the fraction of prediction sets  with infinite length (right) in drug property prediction under a decision-driven selection rule for PEMI (\texttt{Ours Upper}), randomized PEMI (\texttt{Ours randomize}), vanilla conformal prediction (\texttt{Vanilla}), EXPRESS (\texttt{EXPRESS}) and OnlineSCI with three parameter choices (\texttt{OnlineSCI1/2/3}). The purple dashed line is the target coverage level $1-\alpha=0.6$.}}
  \label{fig:real decision}
\end{figure}

\subsubsection{Weighted quantile or average of predictions}
\label{subsubsec:w_quantile_pred}

From now on, we demonstrate the application of PEMI to selection rules that cannot be addressed by the EXPRESS framework. Therefore, the comparison is between our method, vanilla conformal prediction, and OnlineSCI (even though the theoretical results in~\cite{humbert2025online} might not apply to the selection rules here, the procedures can be executed anyway). 

The selection rule is based on the weighted quantile or average of the predictions, with  $S_t=\ind\{\widehat{\mu}_t > \text{wQ}(1-q_{sel}, \{\widehat{\mu}_i\}_{i=1}^{t-1})\}$ and $S_t=\ind\{\widehat{\mu}_t > \frac{\sum_{i=1}^{t-1}w_i \widehat{\mu}_i}{\sum_{i=1}^{t-1}w_i}\}$, where we define $q_{sel}=0.1$ and the pre-specified weights $w_i \propto {0.5}^{(t-i)}$, which assigns higher importance to predictions closer to the current test point. 
For consistency, we keep the confidence level at $\alpha=0.4$ and use the absolute residual conformity score.

Figure~\ref{fig:real wq of mu} and Figure~\ref{fig:real wa of mu} respectively present the results under these two selection rules across 10,000 runs. For selection-conditional coverage, our methods remain above the target level at all time points. In comparison, vanilla CP severely undercovers for most time points. OnlineSCI1 (with a well-tuned parameter) falls below the target level in the early stage but quickly adapts and converges to coverage close to the target, whereas OnlineSCI2 (with less well-tuned parameter) fails to adapt to the target level in the time horizon evaluated. Regarding the interval length, our methods and OnlineSCI moderately enlarge the prediction sets compared with vanilla CP so as to ensure valid coverage.

% Figure~\ref{fig:real wq of mu} and Figure~\ref{fig:real wa of mu} respectively present the results under these two selection rules across 10,000 runs. For selection-conditional coverage, our methods remain above the target level at all time points, whereas vanilla CP fails to achieve the target at most points. Regarding the interval length, our methods moderately enlarge the prediction sets compared with vanilla CP so as to ensure valid coverage. 

\begin{figure}[htbp]
  \centering
  \includegraphics[width=0.9\linewidth]{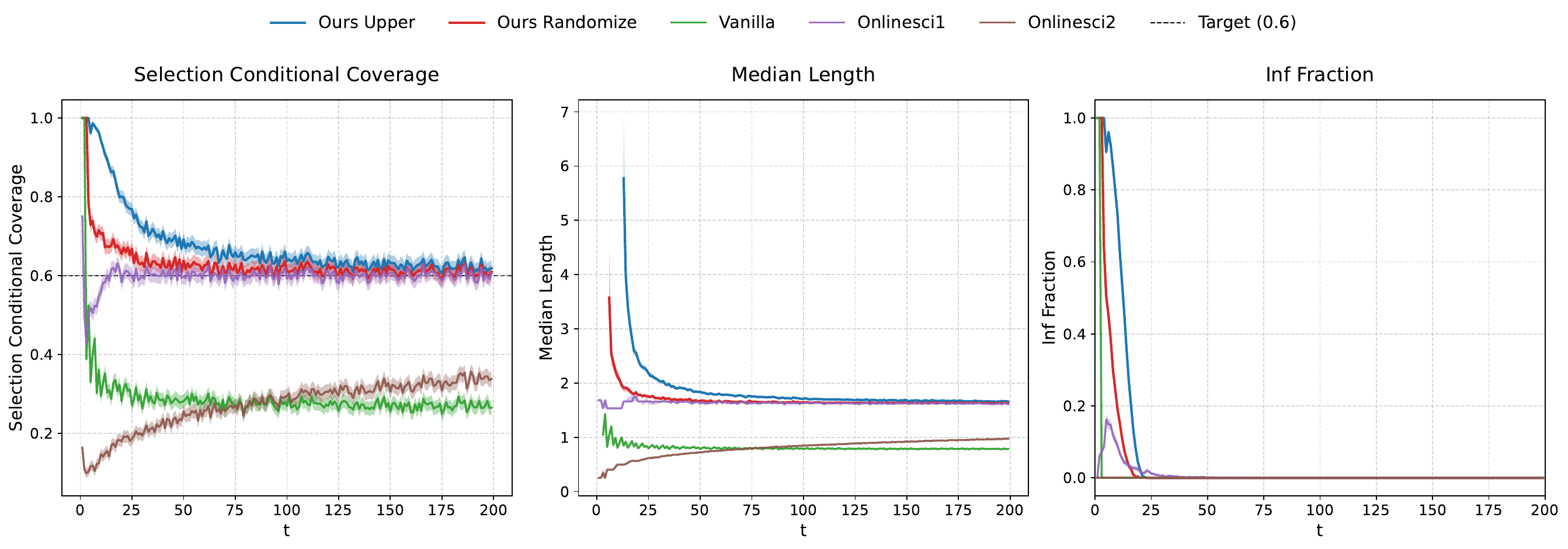}
  \caption{{\small Results under a selection rule based on weighted quantile of predictions when applying PEMI (\texttt{Ours Upper}), randomized PEMI (\texttt{Ours randomize}), vanilla conformal prediction (\texttt{Vanilla}) and  OnlineSCI \texttt{OnlineSCI}). The purple dashed line is the target coverage level $1-\alpha=0.6$. Details are otherwise as Figure~\ref{fig:real decision}.}}
  \label{fig:real wq of mu}
\end{figure}

\begin{figure}[htbp]
  \centering
  \includegraphics[width=0.9\linewidth]{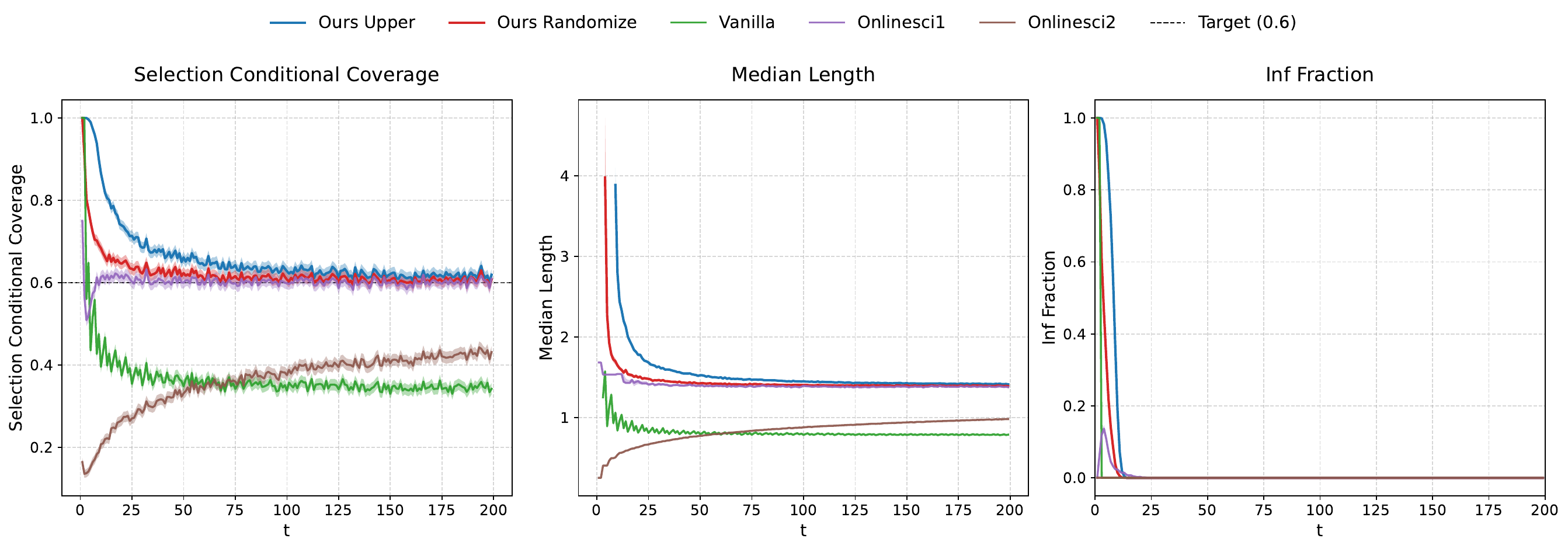}
  \caption{{\small Results under selection  based on weighted average of earlier predictions; details otherwise as Figure~\ref{fig:real wq of mu}.}}
  \label{fig:real wa of mu}
\end{figure}

\subsubsection{Selection based on model uncertainty and online optimization}

Third, we consider selection rules based on model uncertainty from multiple models. Using the DeepPurpose library, we pre-train three different models  $\{f^{(j)}\}_{j=1}^3$ with identical neural network architecture and different featurizations of the drugs and targets (Morgan and Conjoint Triad, PubChem and AAC, and Daylight and PseudoAAC). For each test point, we compute the variance $s_t = \text{Var}(\{f^{(j)}(X_t)\}_{j=1}^3)$ and select the point if $s_t\geq \tau_t$. The threshold $\tau_t$ is determined through the optimal solution to an online optimization program: $\underset{\tau}{\text{max}}\sum_{i=1}^{t-1} s_i \ind\{s_i \geq \tau\}$ subject to $\frac{1}{t-1}\sum_{i=1}^{t-1} \ind\{s_i\geq \tau\} \leq \gamma$. Operationally, this online optimization adaptively sets $\tau_t$ to maximize the cumulative uncertainty of admitted points while enforcing the throughput constraint $\gamma$, thereby allocating limited assay capacity to the most informative candidates at each time. Finally, we use the first model $\hat\mu=f^{(1)}$  to construct prediction sets at all selected time points.

\begin{figure} 
  \centering
  \includegraphics[width=0.9\linewidth]{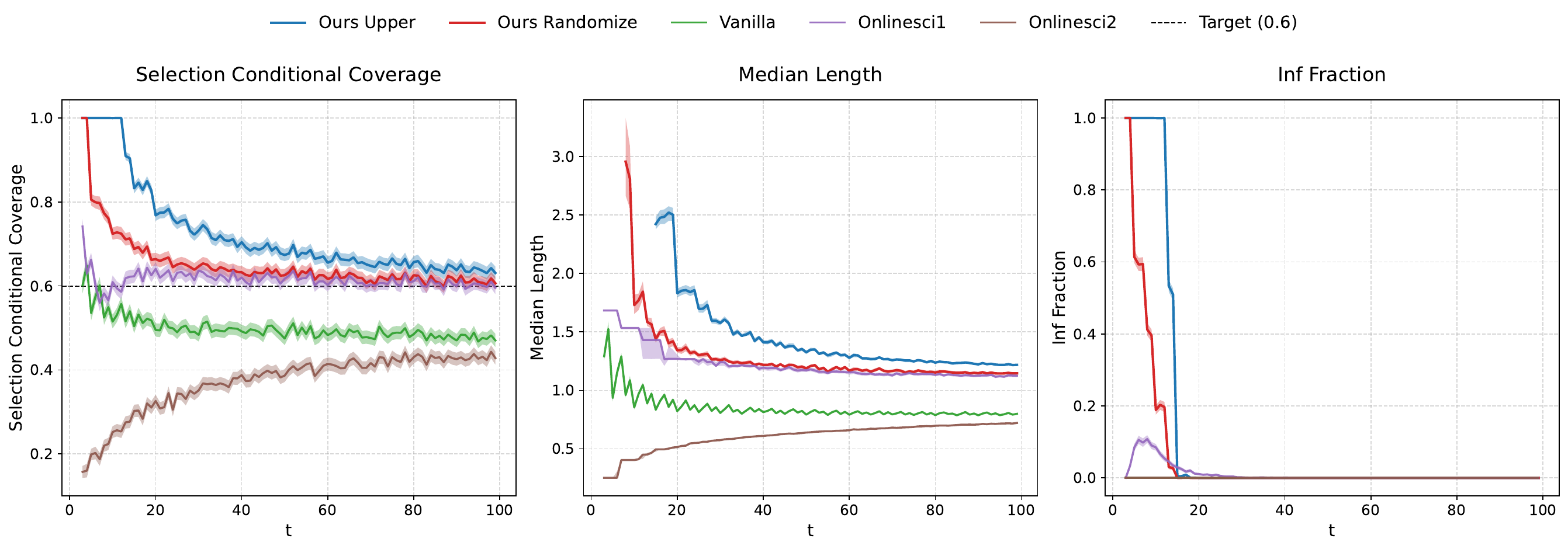}
  \caption{{\small Results under  selection  based on model uncertainty with online optimization; details otherwise as Figure~\ref{fig:real wq of mu}.}}
  \label{fig:real online optimization}
\end{figure}

Figure~\ref{fig:real online optimization} presents the result for this selection rule across $N=10000$ runs. While vanilla CP is over-confident at nearly all time points since the high-uncertainty points typically need larger prediction sets and the performance of OnlineSCI again relies on proper choice of hyperparameters, our methods maintain coverage that is consistently above or exactly at the target level. Indeed, PEMI slightly enlarges the prediction set size to ensure validity.

% Figure~\ref{fig:real online optimization} presents the result for this selection rule across $N=10000$ runs. While vanilla CP is over-confident at nearly all time points since the high-uncertainty points typically need larger prediction sets, our methods maintain coverage that is consistently above or exactly at the target level. Indeed, PEMI slightly enlarges the prediction set size to ensure validity.

\subsection{Conformal selection rules}
\paragraph{Fixed threshold.}
We first study conformal selection with a fixed threshold. To demonstrate the generality of our approach, we employ the weighted conformal p-value, originally designed to mitigate coverage gap due to distribution shift~\citep{jin2025model,barber2023conformalpredictionexchangeability}. The pre-fixed weights are $w_i=0.99^{t+1-i}$. At each test point, a candidate is selected if its p-value is below a pre-specified threshold $q=0.3$. Following the setting in Section~\ref{subsec:conformal_selection}, the thresholds  $c_i$'s are taken as the $0.7$-th quantile of the training pairs' true binding affinities with the same target as sample $i$. The experiment is repeated 10000 times.

Figure~\ref{fig:real fixed} presents the results. For selection-conditional coverage, vanilla CP fails to achieve the target level at most time points, and OnlineSCI relies on hyperparameter choice. In contrast, PEMI stays very close to the target level at most points (except for earlier time points where infinite-length sets have to occur). The PEMI set is larger than vanilla CP to ensure valid coverage, while still being smaller than the well-tuned OnlineSCI1 sets.

% Figure~\ref{fig:real fixed} presents the results. For selection-conditional coverage, vanilla CP is overly conservative. In contrast, PEMI stay very close to the target level at most points (except earlier time points where infinite-length sets occur). The PEMI set is larger than vanilla CP to ensure valid coverage. 

\begin{figure}[htbp]
  \centering
  \includegraphics[width=0.9\linewidth]{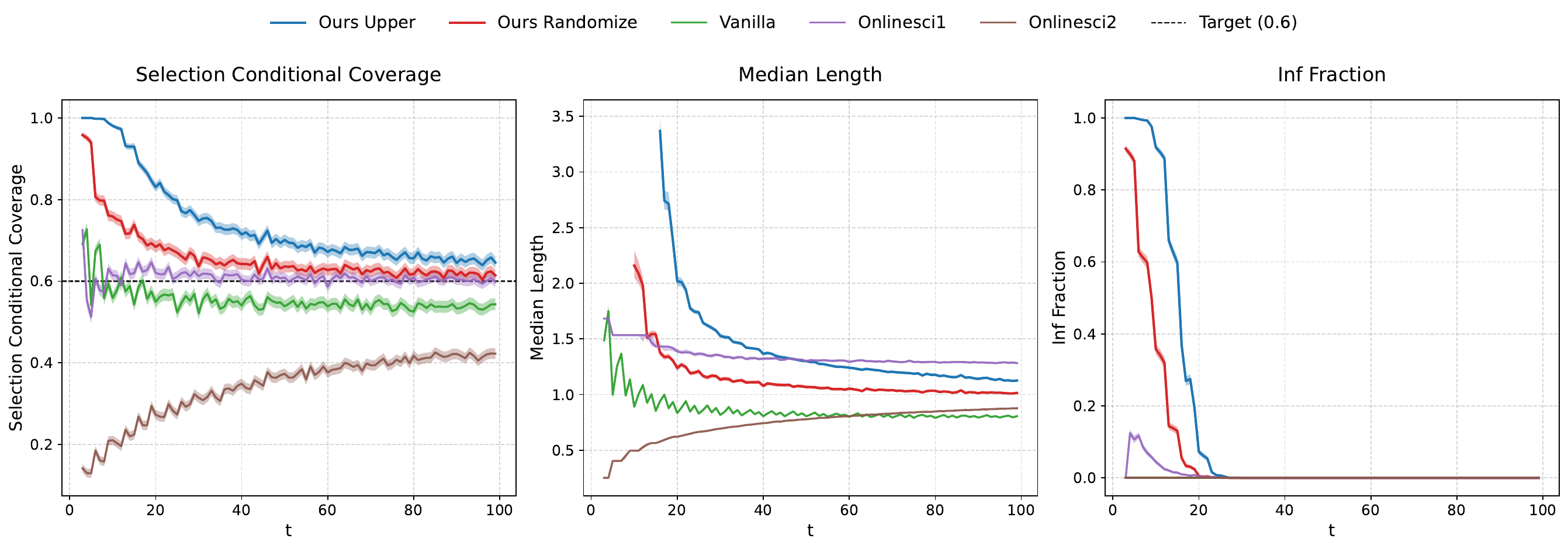}
  \caption{{\small Results under conformal selection with a fixed threshold; details are otherwise the same as Figure~\ref{fig:real wq of mu}.}}
  \label{fig:real fixed}
\end{figure}

\paragraph{E-LOND threshold.}
\label{subsubsec:real_elond}
Second, we consider conformal selection with online multiple testing procedures. As a demonstration, we adopt the Ue-LOND algorithm in \citet{xu2023onlinemultipletestingevalues}---the randomized version of e-LOND---as our online multiple testing procedure, which provides rigorous FDR control. The e-values used in this procedure are defined as in Section~\ref{subsec: evalue}.
To ensure sufficient selection frequency for evaluation, we set aside 50 data points as an offline calibration set. 
Consistent with Section~\ref{subsec: evalue}, the Ue-LOND procedure computes an e-value $e_t$ and a threshold $\alpha_t$ for each time point $t$, which is selected whenever $e_t\geq 1/\alpha_t$. For each sample $i$, we take $c_i$ as the 0.3-th quantile of the binding affinities of training pairs with the same disease target.
It is worth noting that the Ue-LOND procedure has low selection power here. 
As such, we lower the thresholds $c_i$'s and scale up the experiment to  100,000 independent runs for reliable evaluation.

Figure~\ref{fig:real e-lond} presents the results for this selection rule. The limited power of Ue-LOND makes the selection into the reference set difficult, and therefore the results appear conservative. Nevertheless, PEMI methods maintain valid selection-conditional coverage, whereas vanilla CP and OnlineSCI2 (with less well-tuned parameters) fall well below the target level.  
The better-tuned OnlineSCI1 curve also stays slightly below the target level  since the selection time points, in which it makes adjustment to its thresholds, are rare. For both the  prediction set size and the fraction of infinite-length sets, we observe an initial increase followed by a gradual decrease over time. This pattern is due to the nature of the selection procedure: early in the sequence, \(\{\alpha_t\}\) decreases rapidly and makes selection into the reference set difficult, thereby driving up both the set size and the proportion of infinite-length sets. As time progresses, the reference set of permutations grows and the decay of $\alpha_t$ becomes slower, leading to a slow decline in both metrics.

\begin{figure}[htbp]
  \centering
  \includegraphics[width=0.9\linewidth]{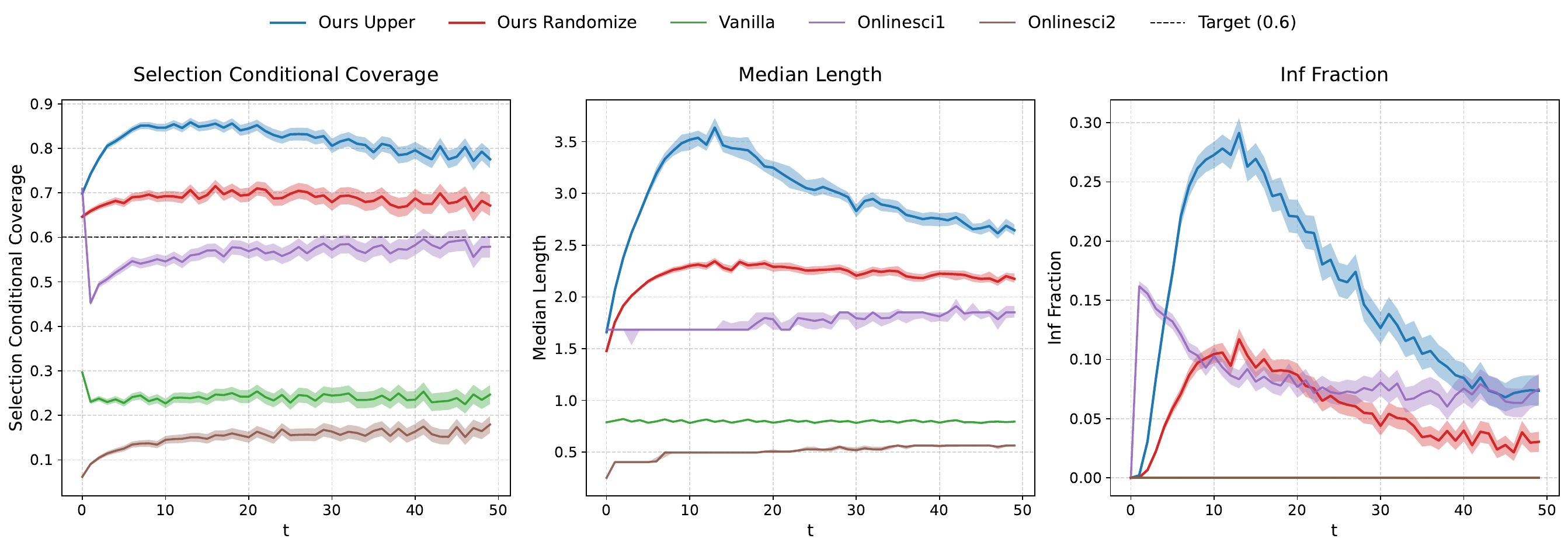}
  \caption{{\small Results under conformal selection with a e-LOND threshold; details are otherwise the same as Figure~\ref{fig:real wq of mu}.}}
  \label{fig:real e-lond}
\end{figure}

\subsection{Selection based on earlier outcomes}
%\subsubsection{Quantile of earlier outcomes}

Finally, we evaluate selection rules constructed directly from earlier labels. We select the point $t$ whenever its predicted value $\mu_t=\hat\mu(X_t)$ exceeds the quantile and the weighted quantile of $\{Y_i\}_{i<t}$. The weights are set as $w_i \propto {0.5}^{(t-i)}$, giving higher importance to more recent candidates.

Figure~\ref{fig:real q of y} and Figure~\ref{fig:real wq of y} present the results for these two selection rules across 10,000 runs, respectively.
PEMI maintains valid coverage, whereas vanilla CP and OnlineSCI2 (less well-tuned) fail at most points. OnlineSCI1 with better-tund parameters still suffers from undercoverage at very early points. The middle and right panels show the reasonable efficiency of our methods with a quickly decaying frequency of infinite-length prediction sets. 

\begin{figure}[H]
  \centering
  \includegraphics[width=0.9\linewidth]{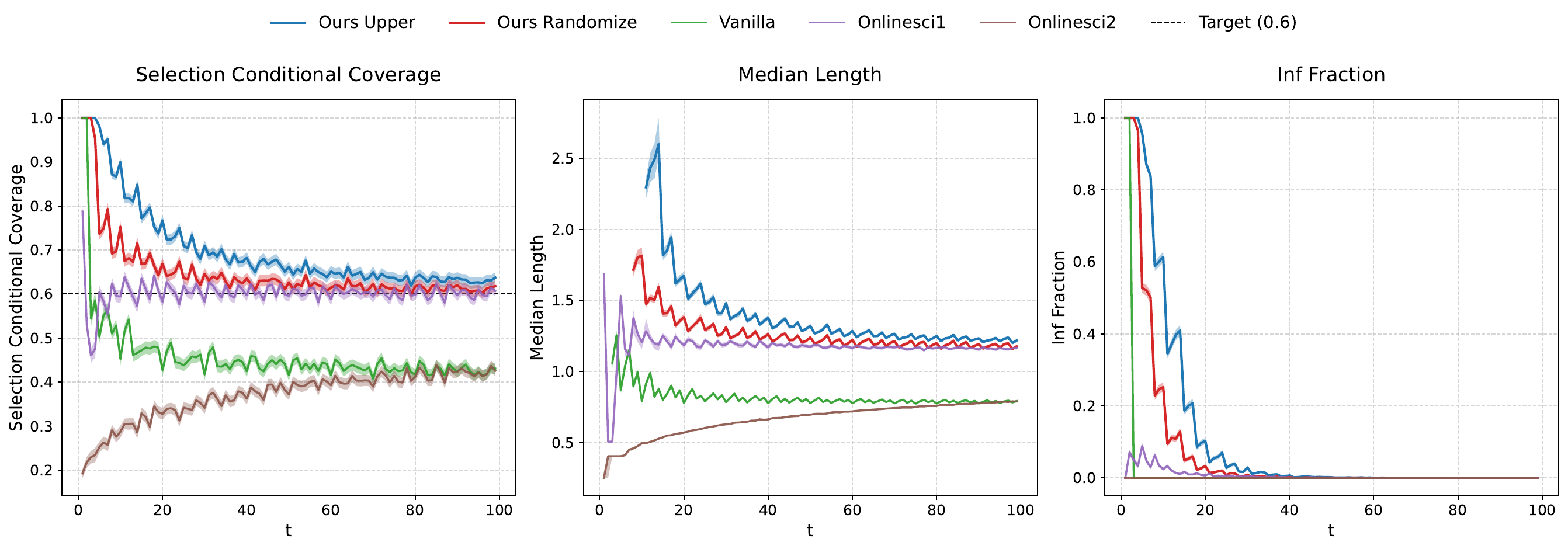}
  \caption{{\small Results under a selection rule based on quantile of earlier outcomes; details are otherwise  as Figure~\ref{fig:real wq of mu}.}}
  \label{fig:real q of y}
\end{figure}

%\subsubsection{Weighted quantile of earlier outcomes}

\begin{figure}[H]
  \centering
  \includegraphics[width=0.9\linewidth]{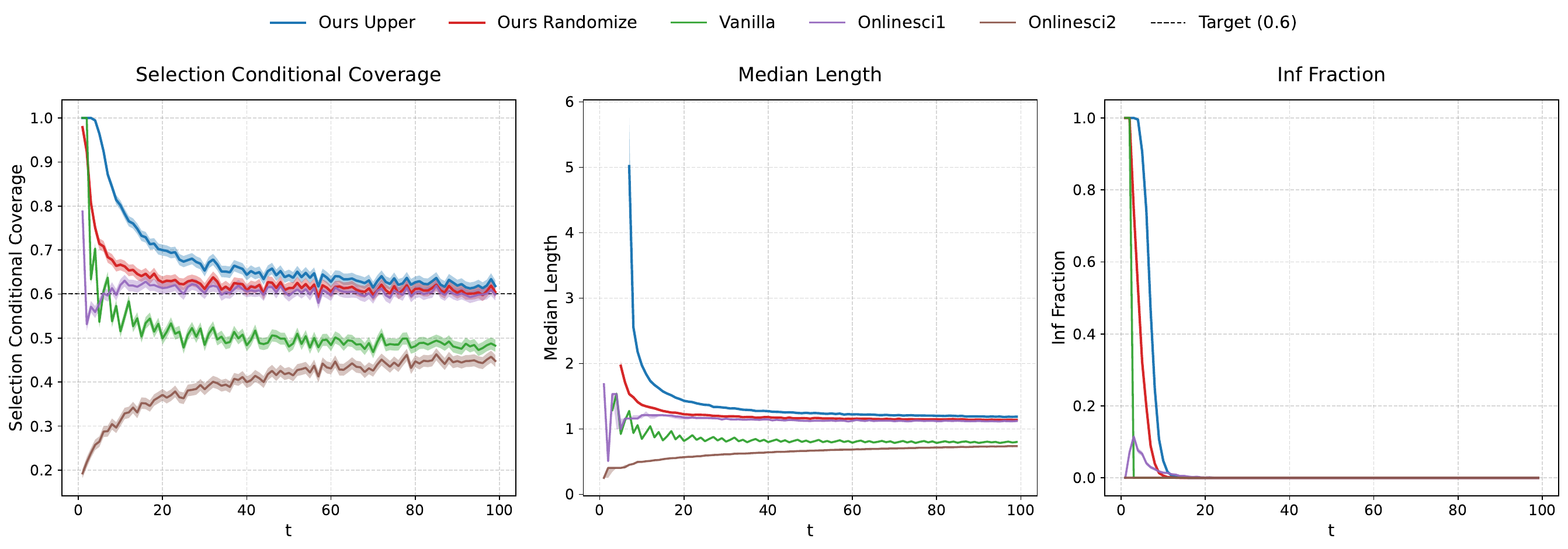}
  \caption{{\small Results under a selection rule based on weighted quantile of earlier outcomes; details otherwise as Figure~\ref{fig:real wq of mu}.}}
  \label{fig:real wq of y}
\end{figure}

%!tex root = main.tex

\section{Simulation studies}
\label{sec:simu}

In this part, we design simulation studies to investigate the impact of various factors, such as the quality of the predictive model, the noise level of the data, the decay rate of weights in certain selection rules, and the quality of the score function, on the performance of our methods. 

\subsection{Model quality and noise level}
\label{subsec:model quality}
We first evaluate our method with controlled prediction models and noise levels so the discrepancy between  \(\widehat{\mu}(X)\) and the true label \(Y\) varies.
We generate i.i.d.~covariates \(X_i \sim \mathrm{Unif}[-1,1]\) and responses
$Y_i=\mu(X_i)+\epsilon_i$, where $  
\mu(x) = 3\,\sin (4\pi x )
    + 4 [\max\{0,\, x - 0.3\}]^2
    - 4 [\max\{0,\, -(x + 0.4)\}]^2 
$,
where \(\epsilon_i \given X_i \sim \mathcal{N}\! (0,\; \bigl[\sigma \cdot(0.5+\vert X_i \vert)\bigr]^2 )\), and \(\sigma \in\{0.1,\;1.0,\;10.0\}\) controls the noise level. For model quality, we use linear model and random forest to fit the regression model \(\widehat{\mu}(\cdot)\). With a nonlinear true model \(\mu(x)\), linear model suffers from pronounced model misspecification while random forest can capture the nonlinearity and provide a good fit.

For conciseness, we focus exclusively on the weighted quantile rule with weights \(w_i \propto 0.5^{\,t-i}\) similar to Section~\ref{subsubsec:w_quantile_pred}. Figure~\ref{fig:simu model and noise} presents the results under two models of different quality and three noise levels.  

\vspace{-.75em}
\paragraph{Model quality.}  Comparing across models, for selection-conditional coverage, our method performs well with both models while vanilla CP becomes less stable when model misspecification is severe, showing PEMI's robustness to model quality. Regarding prediction set size, when the noise level is low, random forest naturally yields shorter prediction sets. However, when the noise level is high, the difference between them becomes negligible as most of the size contributes to covering the noise instead of signal.

\vspace{-.75em}
\paragraph{Noise level.} Figure~\ref{fig:simu model and noise} shows the results under three different noise levels. 
The coverage of vanilla CP is sensitive to noise: it is overly conservative when the noise level is low and fails to reach the target when the noise level is high. In contrast, our method remains robust and achieves the desired coverage across all noise levels. Regarding prediction set size, higher noise levels lead to larger prediction sets. In addition, our method enlarges the prediction sets to maintain validity when the noise is high.

\begin{figure}[H]
  \centering
  \includegraphics[width=1.0\linewidth]{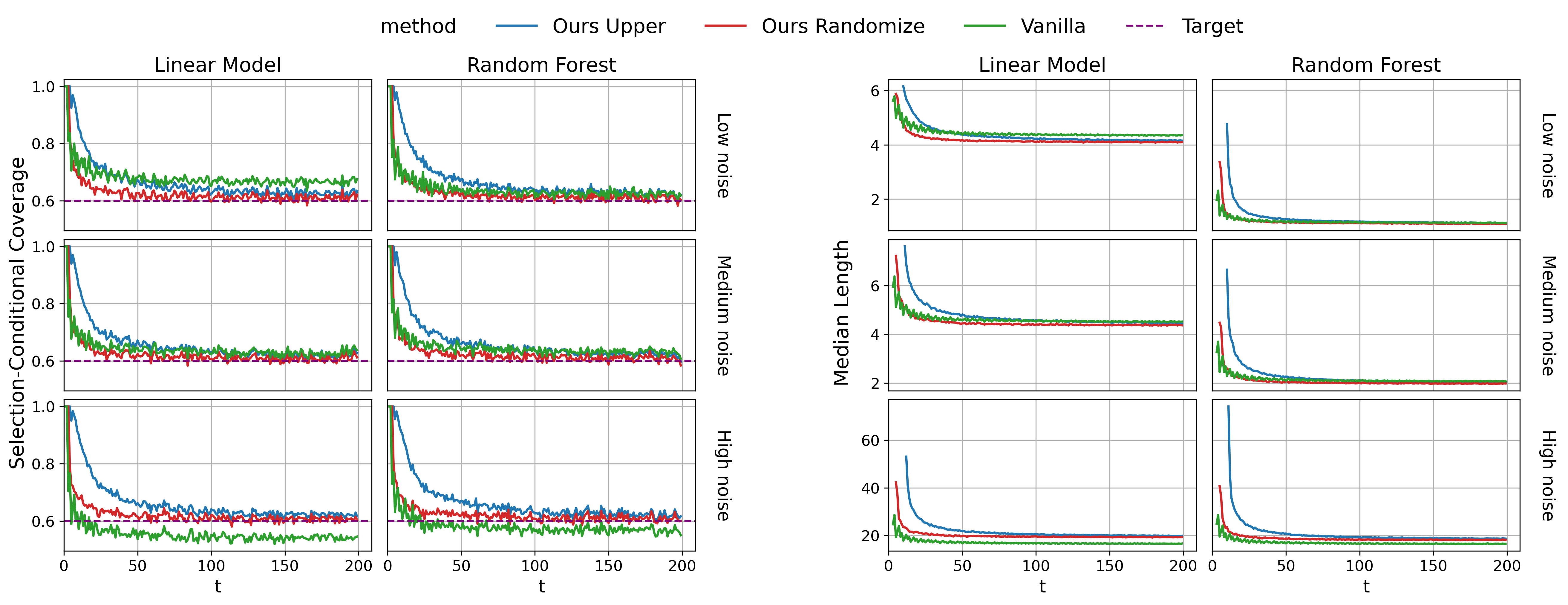}
  \caption{{\small Empirical selection-conditional coverage (left) and median prediction set size (right) with different model qualities and noise levels, under a selection rule based on weighted quantile of predictions when applying PEMI (\texttt{Ours Upper}), randomized PEMI (\texttt{Ours Randomize}), and vanilla conformal prediction (\texttt{Vanilla}). The purple dashed line is the target coverage level $1-\alpha=0.6$. The fraction of infinite-length sets is low and not displayed.}}
  \label{fig:simu model and noise}
\end{figure}

\subsection{Weight decay rate}
\label{subsec: weight decay}

We then evaluate the performance of our method under the weighted quantile selection rule in the last part with different decay rates, which affect the selection power. 
Following the data-generation process in Setting~3 of \citet{jin2023selectionpredictionconformalpvalues}, we generate i.i.d.~covariates \(X_i \sim \mathrm{Unif}[-1,1]^{20}\) and responses $Y_i=\mu(X_i)+\epsilon_i, \  
\mu(x) = 5\bigl(x_{1}x_{2} + e^{x_{4}-1}\bigr),$
where \(\epsilon_i \sim \mathcal{N}\!(0,\; \bigl[\sigma \cdot (5.5 - |\mu(x_i)|)/2\bigr]^2)\), and \(\sigma\) controls the noise level. Here we use a linear model and fix \(\sigma = 1.0\). We use  the same selection rule as before, with weights \(w_i \propto \beta^{\,t-i}\), where \(\beta\in \{0.1, 0.5, 0.9\}\), corresponding to fast, medium, and slow decay rates, respectively.

Figure~\ref{fig:simu weight} shows the results under the three weight decay rates. In this setting, the weight decay rate  determines the power of the selection rule: the faster the decay, the higher the frequency of selection. This further affects our selection reference set: an overly stringent selection rule can make it difficult for permutations to enter the reference set, resulting in overly conservative prediction sets. 
Indeed, we observe that a faster decay leads to selection-conditional coverage that more quickly approaches the target level, accompanied by a reduction in prediction set length and a faster decay of the infinite-length fraction to zero.

% \begin{figure}[H]
%   \centering
%   \includegraphics[width=1.0\linewidth]{figure/simu weight.png}
%   \caption{weight decay}
%   \label{fig:simu weight}
% \end{figure}

\begin{figure}[htbp]
  \centering
  \includegraphics[width=0.8\linewidth]{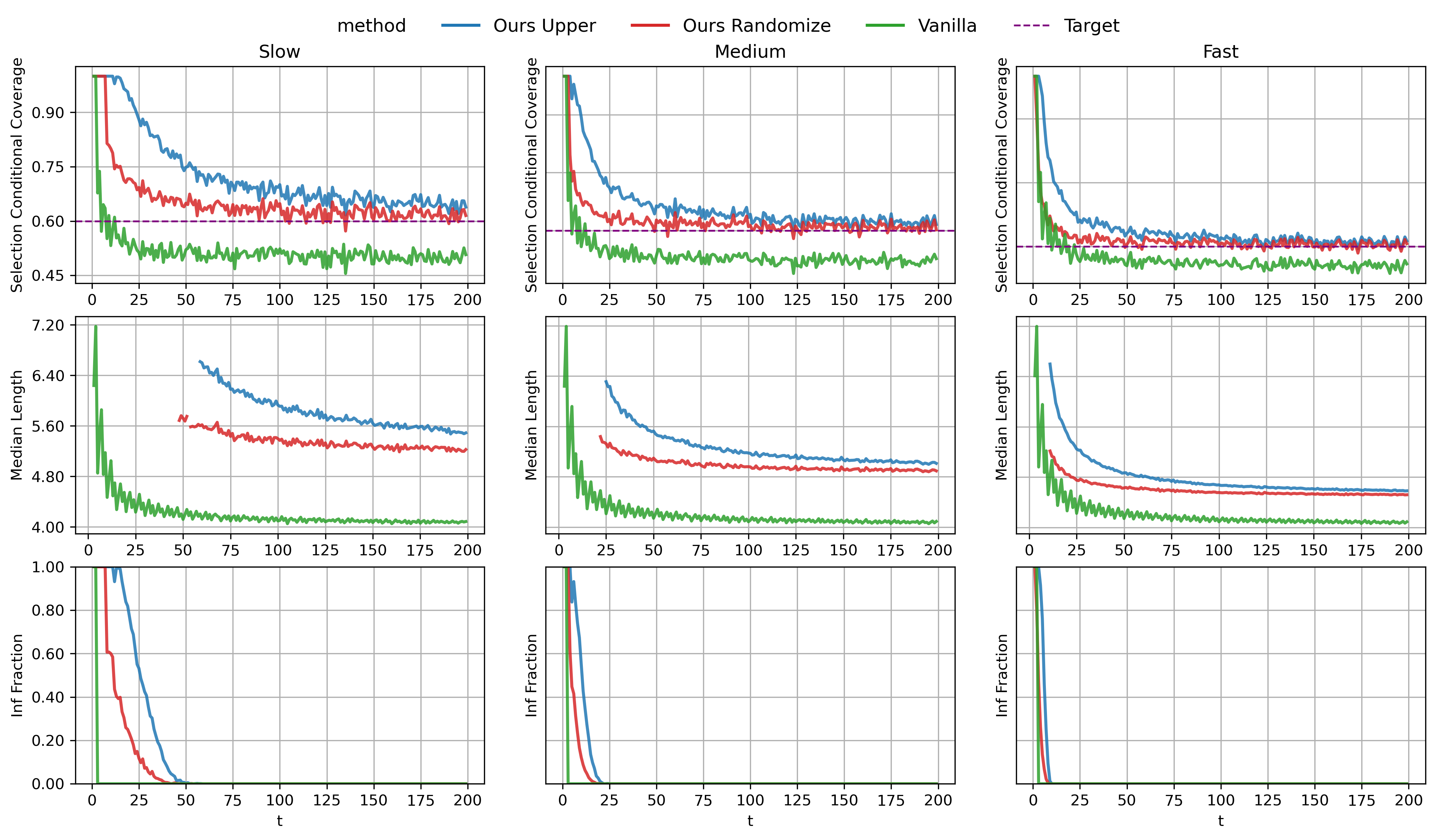}
  \caption{{\small Empirical selection-conditional coverage (top), median prediction set size (middle), and the fraction of prediction sets with infinite length (bottom) in different weight decay rate settings (columns) under a selection rule based on weighted quantile of predictions. Details are otherwise the same as Figure~\ref{fig:simu model and noise}.}}
  \label{fig:simu weight}
\end{figure}

\subsection{Conformity score function}
Finally, we evaluate our methods under different conformity score functions. 
We follow the data-generating processes of~\ref{subsec:model quality} and~\ref{subsec: weight decay}, which are highly nonlinear and heteroskedastic data settings, respectively, with two noise levels $\sigma \in \{1,5\}$. We use the same weighted quantile selection rule as before. 
We vary the score function over the absolute residual score $|y-\hat\mu(x)|$ with $\hat\mu(\cdot)$ being a linear model and random forests, and Conformalized Quantile Regression (CQR) score~\citep{romano2019conformalizedquantileregression}.

\begin{figure}[htbp]
  \centering
  \includegraphics[width=1.0\linewidth]{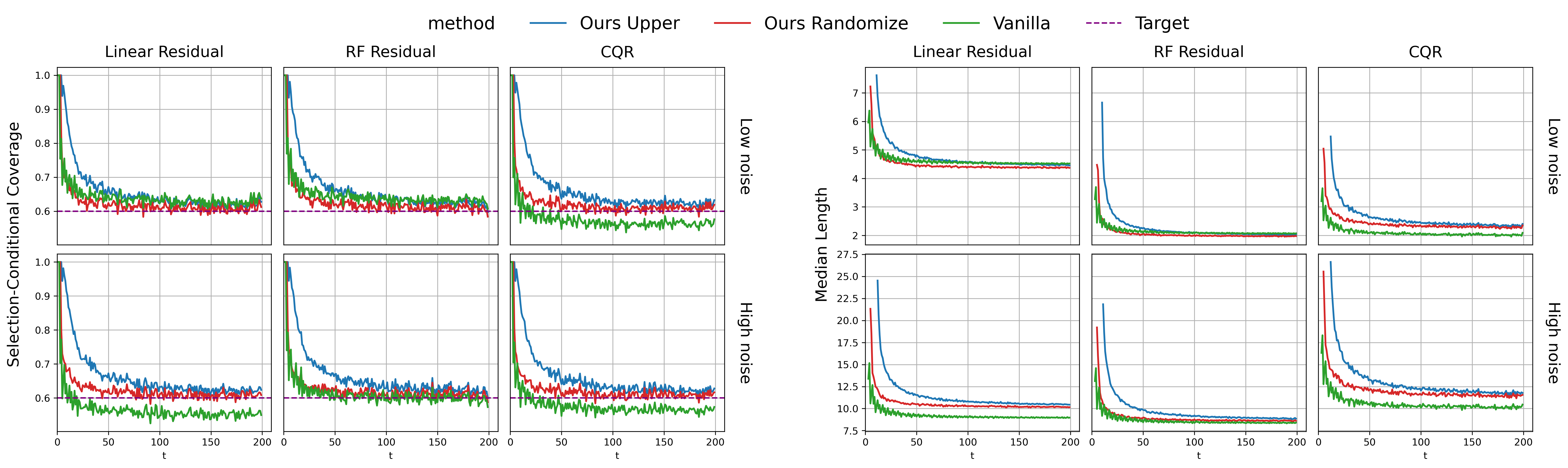}
  \caption{{\small Results under the settings in Section~\ref{subsec:model quality}  with three different score functions (columns) and two noise levels (rows). Details are otherwise the same as Figure~\ref{fig:simu model and noise}.}}
  \label{fig:simu nonlinear}
\end{figure}

\begin{figure}[htbp]
  \centering
  \includegraphics[width=1.0\linewidth]{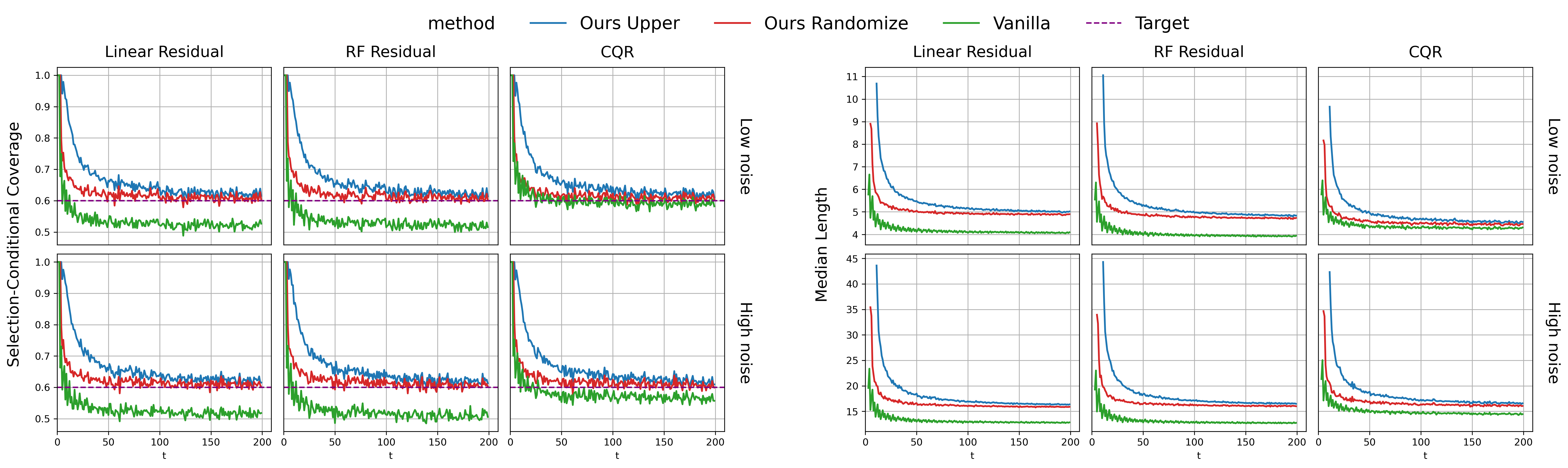}
  \caption{{\small Results under the settings in Section~\ref{subsec: weight decay}  with three different score functions (columns) and two noise levels (rows). Details are otherwise the same as Figure~\ref{fig:simu model and noise}.}}
  \label{fig:simu hetero}
\end{figure}

Figure~\ref{fig:simu nonlinear} and Figure~\ref{fig:simu hetero} present the results for the three score functions in the settings of Section~\ref{subsec:model quality} (nonlinear) and Section~\ref{subsec: weight decay} (heteroskedasticity), respectively. 
In vanilla CP, CQR-based sets lead to more stable performance metrics across runs than residual-based ones. However, in most cases, vanilla CP with CQR still under-covers. Even though CQR is designed to provide approximate $X$-conditional coverage, it is  insufficient for selection-conditional coverage with dynamic selection.  In contrast, our method consistently achieves valid coverage across all score functions. 
We also observe that CQR typically yields slightly smaller sets in the low-noise setting and wider ones in noisy settings.

\FloatBarrier
\vspace{-0.5em}
\section{Discussion}
\vspace{-0.25em}

This paper presents PEMI, a general framework for selective conformal prediction with any selection rules that is asymmetric to the labeled data. The key idea is post-selection inference for permutation test: we use a subset of permutations that lead to the same selection event  to calibrate a prediction set, thereby eliminating the only symmetry condition needed in earlier methods. 
Based on a Monte-Carlo permutation test, we derive computationally efficient implementations of PEMI under various commonly used online selection rules. The efficacy of PEMI is demonstrated on a real drug discovery dataset and extensive simulation studies. 

Several research questions remain open. 
First, many conformal prediction methods for online data come with the feature of addressing distributional drift~\citep{gibbs2021adaptive,barber2023conformalpredictionexchangeability,bao2024cap}, and similar ideas have recently been used for selective coverage guarantees~\citep{humbert2025online},  
motivating the question of how PEMI may adapt to distribution shifts. 
Second, the selection-conditional coverage may be achieved with a  more general permutation test. Besides the entire set of permutations or a uniformly drawn subset, other choices---such as a set of random swaps---can also lead to valid permutation tests~\citep{ramdas2023permutation,barber2023conformalpredictionexchangeability}. In particular, the JOMI method~\citep{jin2024confidence} can be viewed as setting $\Pi_t$ as the set of permutations that swaps the data point $t$ with all others; although its symmetry condition fails,  \cite{barber2023conformalpredictionexchangeability} shows a \emph{randomly-drawn} swap indeed leads to valid conformal prediction set even with asymmetric conformity scores, mirroring the settings here. It is thus interesting to study whether similar  ideas may lead to valid selection-conditional coverage. Finally, for selection procedures with rare selection event, such as those in Section~\ref{subsubsec:real_elond}, it is challenging to obtain a sufficiently large reference set and non-infinite-length set. This might reflect the fundamental difficulty of selection-conditional coverage with limited labeled data and rare selection event, which warrants a theoretical investigation.

\vspace{-0.5em}
\section*{Acknowledgments}
\vspace{-0.25em}
Y.~J.~would like to thank Rina Barber for helpful discussion on the generality of permutation tests.

\newpage
%%%%%%%%%%%%%%%%%%%%
%%% Bibliography %%%
%%%%%%%%%%%%%%%%%%%%
\bibliographystyle{apalike}
\bibliography{reference}

\clearpage
\appendix

\section{Deferred details}

\subsection{Complete version of Proposition~\ref{prop: covariate}}
\label{app:subsec_cov_full}

\begin{prop}
\label{prop: covariate_full}
Assume $S_t = 1$.  
Let $\Pi_t^{(M)}$ be i.i.d.~samples from the uniform distribution over $\Pi_t$. 
Define 
\begin{equation}
\label{eq: AtBt}
\widehat{R}_t = \{\pi \in \Pi_t^{(M)} : S_t(\pi) = 1\} \cup\{\pi_0\}, \qquad
B_t = \{\pi \in \Pi_t^{(M)} : S_t(\pi) = 1,\ \pi(t) \neq t\}.
\end{equation}
The PEMI prediction sets are given by: 
\begin{itemize}
    \item[(a)]   
    $ 
    \widehat{\mathcal{C}}_{\alpha, t}(\Pi_t^{(M)}) = \left\{ y \in \mathcal{Y} : v(X_t, y) \leq \mathrm{Quantile}\left(\beta; \{ v(X_{\pi(t)}, Y_{\pi(t)}) \}_{\pi \in B_t} \cup \{+\infty\} \right) \right\}$ for $\beta = \frac{\lceil (1-\alpha)\cdot |\widehat{R}_t| \rceil}{|B_t|}.$  

    \item[(b)]   
    $\widehat{\mathcal{C}}_{\alpha, t}^{\mathrm{rand}} (\Pi_t^{(M)})= \left\{ y \in \mathcal{Y} : v(X_t, y) \leq q \right\}$, 
    where the random threshold $q\in \RR$ is sampled from
    \[
    q \sim
    \begin{cases}
        \mathrm{Quantile}\left( \frac{r^* - 1}{|B_t|};\, \{ v(X_{\pi(t)}, Y_{\pi(t)}) \}_{\pi \in B_t} \cup \{+\infty\} \right), & \text{with probability } p; \\
        \mathrm{Quantile}\left( \frac{r^*}{|B_t|};\, \{ v(X_{\pi(t)}, Y_{\pi(t)}) \}_{\pi \in B_t} \cup \{+\infty\} \right), & \text{with probability } 1-p.
    \end{cases}
    \]
    Here, we define the probability $p = \frac{\alpha \cdot |\widehat{R}_t| - (|B_t| - r^* + 1 - e_B^*)}{ e_R^*}$ for $e_R^* = \#\{ \pi \in \widehat{R}_t : v(X_{\pi(t)}, Y_{\pi(t)}) = v_{(r^*)} \}$, 
    $e_B^* = \#\{ \pi \in B_t : v(X_{\pi(t)}, Y_{\pi(t)}) = v_{(r^*)} \}$, and $r^* = \min\{ i : v_{(i)} = v_{(\lceil (1-\alpha)\cdot |\widehat{R}_t| \rceil)} \}$, where $v_{(1)} \leq \cdots \leq v_{(|B_t|)}$ are the order statistics of $\{ v(X_{\pi(t)}, Y_{\pi(t)}) : \pi \in B_t \}$. 
    % Then, we can compute the probability of taking the lower quantile as follows: \yingcomment{what's the lower quantile?}
    % \[
    % p = \frac{\alpha(1+|\widehat{R}_t|) - (|B_t| - r^* + 1 - e_B^*)}{1 + e_R^*}.
    % \]
    
\end{itemize}
\end{prop}

\subsection{PEMI set under e-LOND selection}
\label{app:subsec_elond_detail}

Proposition~\ref{prop: elond threshold} presents the explicit form of the PEMI set $\widehat{\mathcal{C}}_{\alpha,t}(\Pi_t^{(M)})$. The detailed proof is in Appendix~\ref{proof: prop evalue}.

\begin{prop}\label{prop: elond threshold}
Consider a selected unit $t\in \NN^+$ obeying $E_t\geq 1/\alpha_t^{\text{e-LOND}}$. For any permutation $\pi$ over $\{-n+1,\dots,t\}$, we denote the score after permutation $\hat{F}_{\pi(i)}=F(X_{\pi(i)},c_{\pi(i)})$ for $i\in \{-n+1,\dots,t\}$. 
    For $k\in\{0,1\}$, we define the two leave-one-out conformal p-values for each \(j\in[t]\) under permutation $\pi$ by 
    \[
   p_{j,\pi}^{k,-}
    =
    \frac{
        \sum_{-n+1\leq i\leq 0,i\neq \pi^{-1}(t)}
        \mathds{1}\bigl\{\widehat F_{\pi(i)}\ge\widehat F_{\pi(j)},\;Y_{\pi(i)}\le c_{\pi(i)}\bigr\}
      + k\cdot
        \mathds{1}\bigl\{\widehat F_{t}\ge\widehat F_{\pi(j)}\bigr\} \cdot \ind{\bigl\{\pi^{-1}(t) \leq 0\bigr\}}
    }{
      n+1
    }\,,\]
    \[
    p_{j,\pi}^{k,+}
    =
    \frac{
      1
      + \sum_{-n+1\leq i\leq 0,i\neq \pi^{-1}(t)}
        \mathds{1}\bigl\{\widehat F_{\pi(i)}\ge\widehat F_{\pi(j)},\;Y_{\pi(i)}\le c_{\pi(i)}\bigr\}
      + k\cdot
        \mathds{1}\bigl\{\widehat F_{t}\ge\widehat F_{\pi(j)}\bigr\} \cdot \ind{\bigl\{\pi^{-1}(t) \leq 0\bigr\}}
    }{
      n+1
    }\,.
    \]
    For each \(i\in[t]\), let $\widehat{\alpha}_{i,\pi}^{k,-}$ and $\widehat{\alpha}_{i,\pi}^{k,+}$ be the corresponding thresholds computed with $\{p_{j,\pi}^{k,-}\}_{1\leq j\leq i-1}$ and $\{p_{j,\pi}^{k,+}\}_{1\leq j\leq i-1}$ by the LOND algorithm. Then we define the e-values under permutation $\pi$ by
    $ 
    E_{i,\pi}^{k}
:= \ind\! \{p_{i,\pi}^{k,+}\le \widehat{\alpha}_{i,\pi}^{k,+} \} / \widehat{\alpha}_{i,\pi}^{k,-} 
    $ for all $i\in[t]$. 
    Similarly, based on these e-values and e-LOND algorithm, we can compute the threshold $\widehat{\alpha}_{t,\pi}^{k}$ at time $t$.
    Then, we have   $\widehat{R}_t(y)=\widehat {R}_t^{0}$ for $y > c_t$ and $\widehat{R}_t(y)=\widehat {R}_t^{1}$ for $y \leq c_t$, where
    \begin{equation}\label{eq: evalue reference set}
    \widehat{R}_t^{k} \;=\; \bigl\{\pi\in\Pi_t^{(M)} : E_{t,\pi}^{k} \geq 1/\widehat{\alpha}_{t,\pi}^{k} \bigr\}\cup \{\pi_0\}, \quad B_t^k = \{\pi \in \widehat{R}_t^k :  \pi(t) \neq t\}.
    \end{equation}
    Finally, the PEMI prediction set can be computed via 
    \begin{equation}\label{eq: evalue prediction set}
        \widehat{\mathcal{C}}_{\alpha, t}(\Pi_t^{(M)})
    =
    \left\{\,y\in\mathcal Y: y>c_{t},\,v(X_t,y)\le\widehat{q}_0\,\right\}
    \cup
    \left\{\,y\in\mathcal Y: y\le c_{t},\,v(X_t,y)\le\widehat{q}_1\,\right\},
    \end{equation}
    where $\widehat{q}_k
    =\mathrm{Quantile} \bigl(\frac{\lceil (1-\alpha)\cdot\widehat{R}_t^k| \rceil}{|B_t^k|};\ 
        \{ v(X_{\pi(t)}, Y_{\pi(t)}) \}_{\pi \in B_t^k}\cup\{+\infty\}\bigr)
    \text{for }k=0,1$.
\end{prop}

\section{Technical proofs}
\label{sec:appendix}

\subsection{Proof of Theorem~\ref{thm: general scc}}
\label{proof: scc}

Recall that $\mathcal{D}_t = (Z_1, Z_2, \ldots, Z_t)$ is the ordered dataset up to time $t$ with the unknown true label $Y_t$. Let $[\mathcal{D}_t] = [Z_1, \ldots, Z_{t-1}, Z_t]$ denote the unordered set of $\mathcal{D}_t$. We denote any realized value  of the unordered set as $[d_t] = [z_1, \ldots, z_{t-1}, z_t]$. 
Here, given $[\cD_t]=[d_t]$, the only randomness lies in which element of $\mathcal{D}_t$ corresponds to which value in $[d_t]$. 
We denote $\Pi_t$ as the entire set of permutations over $\{1,\dots, t\}$. 
For notational simplicity, we also write $S_t(\pi;[d_t])=\cS_t(z_{\pi(1)}, z_{\pi(2)}, \ldots, x_{\pi(t)})$ and $V_t(\pi;[d_t])=\cV_t(z_{\pi(1)}, z_{\pi(2)}, \ldots, z_{\pi(t)})$. 
Given any realized value $[d_t]$, we let $\widehat{\pi}$ denote the random permutation corresponding to the observed data, so that $(Z_1, \ldots, Z_t) = (z_{\widehat{\pi}(1)}, \ldots, z_{\widehat{\pi}(t)})$. The selection decision of the observed data thus obeys 
\$
S_t = \cS_t(Z_1,\dots,Z_{t-1},X_t) =S_t(\widehat{\pi}; [d_t]).
\$
We prove the three statements in Theorem~\ref{thm: general scc} separately. 

\paragraph{Proof for the case $\Pi_t^\star=\Pi_t$.} 
We first focus on $\cC_{\alpha,t}(\Pi_t)$. It suffices to show that for any realized value $[d_t]$, 
\begin{equation}\label{eq:proofa}
    \mathbb{P}\big(p_t(Y_t;\Pi_t) \le \alpha \mid [\cD_t]=[d_t], S_t=1\big)\le \alpha.
\end{equation}
First, we define the reference set of permutations satisfying the selection rule as
\[
R_t = \left\{ \pi \in \Pi_t : S_t(\pi;[d_t]) = 1 \right\}.
\]
Here, by definition, $R_t$ is fully determined by $[d_t]$, not the ordering of the data.  
Then we claim that conditional on $[\cD_t]=[d_t]$,
\begin{equation}\label{eq: veq}
\big\{V_t^{Y_t}(\pi): \pi \in \hat{R_t}(Y_t;\Pi_t)\big\}=\big\{V_t(\pi;[d_t]): \pi \in R_t\big\},
\end{equation}
where repetition of elements is allowed on both sides and the sets are both unordered. In other words, ~\eqref{eq: veq} means that no matter which permutation $\hat{\pi}$ corresponds to, the left-handed side always only depends on $[d_t]$.
Conditional on $[\cD_t]=[d_t]$, we have
\begin{equation}\label{eq: vt}
    V_t^{Y_t}(\pi)
    =\mathcal V_t(Z_{\pi(1)}, Z_{\pi(2)}, \ldots, Z_{\pi(t)})
    =\mathcal V_t(z_{\pi\circ \hat{\pi}(1)}, z_{\pi\circ \hat{\pi}(2)}, \ldots, z_{\pi\circ \hat{\pi}(t)})
    =V_t(\pi \circ \hat{\pi};[d_t]),
\end{equation}
\begin{equation}\label{eq: st}
    S_t^{Y_t}(\pi)
    =\mathcal S_t(Z_{\pi(1)}, Z_{\pi(2)}, \ldots, X_{\pi(t)})
    =\mathcal S_t(z_{\pi\circ \hat{\pi}(1)}, z_{\pi\circ \hat{\pi}(2)}, \ldots, x_{\pi\circ \hat{\pi}(t)})
    =S_t(\pi \circ \hat{\pi};[d_t]).
\end{equation}

Then we can prove the claim above:  
\begin{align*}
    \{V_t^{Y_t}(\pi): \pi \in \hat{R_t}(Y_t;\Pi_t)\}&\overset{(a)}{=} \{V_t(\pi \circ \hat{\pi};[d_t]): \pi \in \hat{R_t}(Y_t;\Pi_t)\}\\
    &\overset{(b)}{=} \{V_t(\pi \circ \hat{\pi};[d_t]): \pi \in \Pi_t, \ S_t^{Y_t}(\pi)=1\}\\
    &\overset{(c)}{=} \{V_t(\pi \circ \hat{\pi};[d_t]):  \pi \in \Pi_t,\ S_t(\pi \circ \hat{\pi};[d_t])=1\}\\
    &\overset{(d)}{=} \{V_t(\sigma;[d_t]):\sigma \in \Pi_t,\ S_t(\sigma;[d_t])=1\}\\
    &\overset{(e)}{=} \{V_t(\pi;[d_t]): \pi \in R_t\}.
\end{align*}
Above, step (a) is due to~\eqref{eq: vt}, step (b) follows from the definition of $\hat{R_t}(Y_t;\Pi_t)$, step (c) is due to~\eqref{eq: st}, step (d) follows from the group structure of $\Pi_t$,  and step (e) follows from the definition of $R_t$. We thus complete the proof of Equation~\eqref{eq: veq}.

For convenience, we denote the p-value of an arbitrary permutation $\pi$ on $[d_t]$ as
\[
p_t(\pi;[d_t])=\frac{\sum_{\pi' \in {R}_t}\mathds{1}\{V_t(\pi;[d_t]) \leq V_t(\pi';[d_t])\}}{\vert{R}_t \vert}
\]
and drop the dependence on $\Pi_t$. 
Recall that $\pi_0$ denotes the identity permutation, thus $\pi_0\circ \hat{\pi}=\hat{\pi}$. Following Equation~\eqref{eq: veq}, we can also conclude that $\vert \widehat{R}_t(Y_{t};\Pi_t) \vert=\vert R_t \vert$. Therefore, given $[\cD_t]=[d_t]$, we have 
\begin{align}
p_t(Y_t;\Pi_t)&=\frac{\sum_{\pi \in \widehat{R}_t(Y_{t};\Pi_t)}\mathds{1}\{V_t^{Y_{t}}(\pi_0) \leq V_t^{Y_{t}}(\pi)\}}{\vert \widehat{R}_t(Y_{t};\Pi_t) \vert}=\frac{\sum_{\pi \in R_t}\mathds{1}\{V_t(\hat\pi;[d_t]) \leq V_t(\pi ;[d_t])\}}{\vert R_t \vert}=p_t(\hat\pi;[d_t])\label{eq: yt to pihat}.
\end{align}
Returning to the proof of~\eqref{eq:proofa}, we have  
\begin{align*}
&\hspace{1.7em}\mathbb{P}\big(p_t(Y_t;\Pi_t) \le \alpha \biggiven [\cD_t]= [d_t], S_t=1\big)\\
&=\mathbb{P}\big(p_t(Y_t;\Pi_t) \le \alpha \biggiven [\cD_t]= [d_t], S_t(\widehat{\pi}; [d_t])=1\big)\\
&\overset{(a)}{=} \frac{\mathbb{P} \left( p_t (Y_{t};\Pi_t) \le \alpha,\; S_t(\widehat{\pi}; [d_t]) = 1 \given [\cD_t]= [d_t] \right)}{ \mathbb{P} \left( S_t(\widehat{\pi}; [d_t]) = 1 \given [\cD_t]=  [d_t] \right)} \\
&\overset{(b)}{=} \frac{\mathbb{P} \left( p_t (\hat{\pi};[d_t]) \le \alpha,\; S_t(\widehat{\pi}; [d_t]) = 1 \given [\cD_t]= [d_t] \right)}{ \mathbb{P} \left( S_t(\widehat{\pi}; [d_t]) = 1 \given [\cD_t]=  [d_t] \right)} \\
&\overset{(c)}{=} \frac{\frac{1}{\vert \Pi_t \vert} \sum_{\pi \in \Pi_t} \mathds{1} \left\{p_t (\pi;[d_t]) \le \alpha \right\}S_t(\pi;[d_t]) }{\frac{1}{\vert \Pi_t \vert} \sum_{\pi \in \Pi_t} S_t(\pi;[d_t]) } \\
% &\overset{(c)}{=} \frac{ \sum_{\pi \in \Pi_t} \mathds{1} \left\{\frac{ \sum_{\pi'\in \Pi_t} \mathds{1} \left\{ V_t(\pi';[d_t]) \le V_t(\pi;[d_t]) \right\} \cdot S_t(\pi';[d_t]) }{ \sum_{\pi'\in \Pi_t} S(\pi';[d_t]) } \le \alpha \right\}S_t(\pi;[d_t]) }{ \sum_{\pi \in \Pi_t} S_t(\pi;[d_t]) } \\
&\overset{(d)}{=} \frac{1}{ |R_t| } \sum_{\pi \in R_t} \mathds{1} \left\{ \frac{ \sum_{\pi' \in R_t} \mathds{1} \left\{ V_t(\pi;[d_t]) \le V_t(\pi';[d_t]) \right\} }{ |R_t| } \le \alpha \right\} \\
&\overset{(e)}{\le} \alpha.
\end{align*}

Above, step (a) follows from Bayes' rule, step (b) is due to~\eqref{eq: yt to pihat}, step (c) follows from the fact that $\widehat{\pi}$ is drawn uniformly from the full permutation set $\Pi_t$, step (d) plug in the explicit definition of the p-value and reference set $R_t$ and step (e) follows from its definition.

Then by tower property, we conclude~\eqref{eq:proofa}, and thus the coverage guarantee. The validity of the randomized version follows the same idea as our proof for the randomized version with $\Pi_t^\star=\Pi_t^{(M)}$ below, which we omit here. 

\paragraph{Proof for the case $\Pi_t^\star = \Pi_t^{(M)}$.} Our proof  relies on the following lemma. 

\begin{lemma}\label{lemma: perm exchangeability}
% Suppose that at time $t$ we have the full permutation set over indices $[t]$, 
% denoted by $\Pi_t$. 
Let $\widehat{\pi}, \pi^{(1)}, \ldots, \pi^{(M)}$ be random permutations i.i.d.~from $\textnormal{Unif}(\Pi_t)$.  Then the random variables
$ 
(\widehat{\pi},\; \pi^{(1)} \circ \widehat{\pi},\; \ldots,\; \pi^{(M)} \circ \widehat{\pi}) 
$
are exchangeable.
\end{lemma}

\begin{proof}[Proof of Lemma~\ref{lemma: perm exchangeability}]
To show the exchangeability of $(\widehat{\pi},\; \pi^{(1)} \circ \widehat{\pi},\; \ldots,\; \pi^{(M)} \circ \widehat{\pi})$, it suffices to prove that their joint distribution is invariant 
under any reordering of the indices. Formally, let $[\hat{\pi}^M]=[\widehat{\pi},\; \pi^{(1)} \circ \widehat{\pi},\; \ldots,\; \pi^{(M)} \circ \widehat{\pi}]$ denote the unordered set of these random variables and $[\sigma^M]=[\sigma^{(0)}, \sigma^{(1)}, \ldots, \sigma^{(M)}]$ denote any realized value of this unordered set. Then, it suffices to show that under any permutation $\tau$ of $\{0,1,\ldots,M\}$, we have
\begin{equation}\label{eq: perm exchangeability}
    \mathbb{P} \left( \widehat{\pi} = \sigma^{(0)},\; \pi^{(1)} \circ \widehat{\pi} = \sigma^{(1)},\; \ldots,\; \pi^{(M)} \circ \widehat{\pi}  = \sigma^{(M)} \right)
=
\mathbb{P} \left( \widehat{\pi} = \sigma^{\tau(0)},\; \pi^{(1)} \circ \widehat{\pi} = \sigma^{\tau(1)},\; \ldots,\; \pi^{(M)} \circ \widehat{\pi} = \sigma^{\tau(M)} \right).
\end{equation}

Since $\widehat{\pi}, \pi^{(1)}, \ldots, \pi^{(M)}$ are independently and uniformly sampled from $\Pi_t$, we can directly compute the value of left-handed side:
\begin{align*}
    &\hspace{1.55em}\mathbb{P} \left( \widehat{\pi} = \sigma^{(0)},\; \pi^{(1)} \circ \widehat{\pi} = \sigma^{(1)},\; \ldots,\; \pi^{(M)} \circ \widehat{\pi}  = \sigma^{(M)} \right) \\
    &= \mathbb{P} \left( \widehat{\pi} = \sigma^{(0)},\; \pi^{(1)} = \sigma^{(1)} \circ (\sigma^{(0)})^{-1},\; \ldots,\; \pi^{(M)} = \sigma^{(M)} \circ (\sigma^{(0)})^{-1} \right) \\
    &\overset{(a)}{=} \mathbb{P}(\widehat{\pi} = \sigma^{(0)}) \cdot \prod_{m=1}^M \mathbb{P}\left(\pi^{(m)} =  \sigma^{(m)} \circ (\sigma^{(0)})^{-1} \right)\\
    &\overset{(b)}{=} \left( \frac{1}{|\Pi_t|} \right)^{M+1}.
\end{align*}

Step (a) follows from the independence of $(\widehat{\pi}, \pi^{(1)}, \ldots, \pi^{(M)})$. Step (b) holds because each of 
$\widehat{\pi}, \pi^{(1)}, \ldots, \pi^{(M)}$ 
is a random variable following a uniform distribution on $\Pi_t$, and therefore the probability 
of any one of them taking a specific value in $\Pi_t$ is $1/|\Pi_t|$. 
Moreover, by the group structure of $\Pi_t$, the permutations 
$\{
\sigma^{(0)},\; \sigma^{(1)} \circ (\sigma^{(0)})^{-1},\; \ldots,\; 
\sigma^{(M)} \circ (\sigma^{(0)})^{-1}
\}$ 
are all fixed elements of $\Pi_t$. Thus we can directly compute the probability as step (b).

Similarly, for any permutation $\tau$, we have
\begin{align*}
    &\hspace{1.55em}\mathbb{P} \left( \widehat{\pi} = \sigma^{\tau(0)},\; \pi^{(1)} \circ \widehat{\pi} = \sigma^{\tau(1)},\; \ldots,\; \pi^{(M)} \circ \widehat{\pi} = \sigma^{\tau(M)} \right) \\
    &= \mathbb{P} \left( \widehat{\pi} = \sigma^{\tau(0)},\; \pi^{(1)} =  \sigma^{\tau(1)} \circ (\sigma^{\tau(0)})^{-1},\; \ldots,\; \pi^{(M)} = \sigma^{\tau(M)} \circ (\sigma^{\tau(0)})^{-1} \right) \\
    &= \mathbb{P}(\widehat{\pi} = \sigma^{\tau(0)}) \cdot \prod_{m=1}^M \mathbb{P}\left(\pi^{(m)} = \sigma^{\tau(m)} \circ (\sigma^{\tau(0)})^{-1}\right)\\
    &= \left( \frac{1}{|\Pi_t|} \right)^{M+1}.
\end{align*} 
Therefore, we can conclude~\eqref{eq: perm exchangeability}. Then the joint distribution of $(\widehat{\pi},\; \pi^{(1)} \circ \widehat{\pi},\; \ldots,\; \pi^{(M)} \circ \widehat{\pi})$ is invariant under permutations, and hence these random permutations are exchangeable. We thus complete the proof of Lemma~\ref{lemma: perm exchangeability}.
\end{proof}
\noindent\textbf{Detailed proof for $\cC_{\alpha,t}(\Pi_t^{(M)})$.} Recall that $\Pi_t^{(M)}=(\pi^{(1)}, \dots, \pi^{(M)} )$. Let $[\hat\pi^M]=[\widehat{\pi},\; \pi^{(1)} \circ \widehat{\pi},\; \ldots,\; \pi^{(M)} \circ \widehat{\pi}]$ denote the unordered set of the random permutations and $[\sigma^M]=[\sigma^{(0)}, \sigma^{(1)}, \ldots, \sigma^{(M)}]$ denote any realization of this unordered set.
To prove~\eqref{eq:validity_dtm}, it suffices to show that
\begin{equation}\label{eq:proofb}
    \mathbb{P}\big(p_t(Y_t; \Pi_t^{(M)}) \le \alpha \mid [\mathcal D_t]=[d_t],[\hat\pi^M]=[\sigma^M], S_t=1\big)\le \alpha.
\end{equation}

First, we define the reference set of permutations  in the unordered set $[\sigma^M]$ as
\[
R_t^{M} = \left\{ \sigma \in [\sigma^M] : S_t(\sigma;[d_t]) = 1 \right\}.
\] 
Here, $R_t^{M}$ is fully determined by $[d_t]$ and $[\sigma^M]$. Then, conditional on $[\mathcal D_t]=[d_t]$ and $[\hat\pi^M]=[\sigma^M]$, we have:
\begin{align*}
    \{V_t^{Y_t}(\pi): \pi \in \hat{R}_t(Y_t;\Pi_t^{(M)})\}&= \{V_t(\pi \circ \hat{\pi};[d_t]): \pi \in \hat{R}_t(Y_t;\Pi_t^{(M)})\}\\
    &= \{V_t(\pi \circ \hat{\pi};[d_t]): \pi \in  \Pi_t^{(M)}\cup \{\pi_0\}, \ S_t^{Y_t}(\pi)=1\}\\
    &= \{V_t(\pi \circ \hat{\pi};[d_t]):  \pi \in  \Pi_t^{(M)}\cup \{\pi_0\},\ S_t(\pi \circ \hat{\pi};[d_t])=1\}\\
    &= \{V_t(\sigma;[d_t]):\sigma \in [\sigma^M],\ S_t(\sigma;[d_t])=1\}\\
    &= \{V_t(\pi;[d_t]): \pi \in R_t^M\}.
\end{align*}
For convenience, we define the p-value on $[d_t]$ as
\[
p_t(\pi;[d_t])=\frac{\sum_{\pi' \in R_t^M}\mathds{1}\{V_t(\pi;[d_t]) \leq V_t(\pi' ;[d_t])\}}{\vert R_t^M \vert}.
\] 

Similar to the proof for $\Pi_t^\star = \Pi_t$, conditional on $[\mathcal D_t]=[d_t]$ and $[\hat\pi^M]=[\sigma^M]$, we have:
\begin{align}
p_t(Y_t;\Pi_t^{(M)})&=\frac{\sum_{\pi \in \widehat{R}_t(Y_{t};\Pi_t^{(M)})}\mathds{1}\{V_t^{Y_{t}}(\pi_0) \leq V_t^{Y_{t}}(\pi)\}}{\vert \widehat{R}_t (Y_{t};\Pi_t^{(M)}) \vert}\\ 
% =\frac{\sum_{\pi \in \widehat{R}_t(Y_{t};\Pi_t^{(M)})\cup\{\pi_0\}}\mathds{1}\{V_t^{Y_{t}}(\pi_0) \leq V_t^{Y_{t}}(\pi)\}}{\vert \widehat{R}_t(Y_{t};\Pi_t^{(M)})\cup\{\pi_0\} \vert}\\
&=\frac{\sum_{\pi \in R_t^M}\mathds{1}\{V_t(\hat\pi;[d_t]) \leq V_t(\pi ;[d_t])\}}{\vert R_t^M \vert}=p_t(\hat\pi;[d_t])\label{eq: mc yt to pihat}.
\end{align}

Returning to the proof of~\eqref{eq:proofb}, we have
\begin{align*}
&\hspace{1.7em}\mathbb{P}\big(p_t(Y_t;\Pi_t^{(M)}) \le \alpha \mid [\mathcal D_t]=[d_t], [\hat\pi^M]=[\sigma^M], S_t=1\big)\\
&=\mathbb{P}\big(p_t(Y_t;\Pi_t^{(M)}) \le \alpha \mid [\mathcal D_t]=[d_t], [\hat\pi^M]=[\sigma^M], S_t(\widehat{\pi}; [d_t])=1\big)\\
&\overset{(a)}{=} \frac{\mathbb{P} \left( p_t(Y_{t};\Pi_t^{(M)}) \le \alpha,\; S_t(\widehat{\pi}; [d_t])=1 \mid [\mathcal D_t]=[d_t], [\hat\pi^M]=[\sigma^M] \right)}{ \mathbb{P} \left( S_t(\widehat{\pi}; [d_t])=1 \mid [\mathcal D_t]=[d_t], [\hat\pi^M]=[\sigma^M] \right)} \\
&\overset{(b)}{=} \frac{\mathbb{P} \left( p_t(\hat\pi;[d_t]) \le \alpha,\; S_t(\widehat{\pi}; [d_t])=1 \mid [\mathcal D_t]=[d_t], [\hat\pi^M]=[\sigma^M] \right)}{ \mathbb{P} \left( S_t(\widehat{\pi}; [d_t])=1 \mid [\mathcal D_t]=[d_t], [\hat\pi^M]=[\sigma^M] \right)} \\
&\overset{(c)}{=} \frac{\frac{1}{1+M} \sum_{m=0}^M \mathds{1} \left\{p_t(\sigma^{(m)};[d_t]) \le \alpha \right\}S_t(\sigma^{(m)};[d_t]) }{\frac{1}{1+M} \sum_{m=0}^M S_t({\sigma^{(m)}};[d_t]) } \\
&\overset{(d)}{=} \frac{1}{ |R_t^M|}\sum_{\sigma^{(m)} \in R_t^M} \mathds{1} \left\{ \frac{ \sum_{\sigma^{(n)} \in R_t^M} \mathds{1} \left\{ V_t(\sigma^{(n)};[d_t]) \le V_t(\sigma^{(m)};[d_t]) \right\} }{ |R_t^M| } \le \alpha \right\} \\
&\overset{(e)}{\le} \alpha.
\end{align*}

Above, step (a) follows from Bayes' rule, step (b) is due to~\eqref{eq: mc yt to pihat}, step (c) follows from the fact that $(\widehat{\pi},\; \pi^{(1)} \circ \widehat{\pi},\; \ldots,\; \pi^{(M)} \circ \widehat{\pi})$ are exchangeable, step (d) plug in the explicit definition of the p-value and reference set $R_t^M$ and step (e) follows from its definition.

Then by tower property, we conclude~\eqref{eq:proofb}, and thus~\eqref{eq:validity_dtm}.

\paragraph{Proof of $\cC_{\alpha,t}^{\textnormal{rand}}(\Pi_t^{(M)})$.} 
The arguments here are similar to those in proof of $\cC_{\alpha,t}(\Pi_t^{(M)})$. Our goal is to prove that
\begin{equation}\label{eq:proofc}
    \mathbb{P}\big(p_t^\text{rand}(Y_t;\Pi_t^{(M)}) \le \alpha \mid [\mathcal D_t]=[d_t],[\hat\pi^M]=[\sigma^M], S_t=1\big)\le \alpha.
\end{equation}

Here, we follow the same form of $R_t^M$ as proof (b). Thus, conditional on $[\mathcal D_t]=[d_t]$ and $[\hat\pi^M]=[\sigma^M]$, we have:
\begin{align*}
    \{V_t^{Y_t}(\pi): \pi \in \hat{R}_t(Y_t; \Pi_t^{(M)})\}
    &= \{V_t(\pi;[d_t]): \pi \in R_t^M\}.
\end{align*}
Then we define the randomized p-value on $[d_t]$ as: 
\[
p_t^\text{rand}(\pi;[d_t])=\frac{\sum_{\pi' \in R_t^M}\mathds{1}\{V_t(\pi;[d_t]) < V_t(\pi' ;[d_t])\}+U_t \cdot \sum_{\pi' \in R_t^M}\mathds{1}\{V_t(\pi;[d_t]) = V_t(\pi' ;[d_t])\}}{\vert R_t^M \vert}.
\]
Similar to the proof for $\cC_{\alpha,t}( \Pi_t^{(M)})$, we have
\begin{align}
p_t^\text{rand}(Y_t; \Pi_t^{(M)})=p_t^\text{rand}(\hat\pi;[d_t])\label{eq: rand yt to pihat}.
\end{align}
Returning to the proof of~\eqref{eq:proofc}, we have
\begin{align*}
&\hspace{1.7em}\mathbb{P}\big(p_t^{\text{rand}}(Y_t; \Pi_t^{(M)}) \le \alpha \mid [\mathcal D_t]=[d_t], [\hat\pi^M]=[\sigma^M], S_t=1\big)\\
&=\mathbb{P}\big(p_t^{\text{rand}}(Y_t; \Pi_t^{(M)}) \le \alpha \mid [\mathcal D_t]=[d_t], [\hat\pi^M]=[\sigma^M], S_t(\widehat{\pi}; [d_t])=1\big)\\
&\overset{(a)}{=} \frac{\mathbb{P} \left( p_t^{\text{rand}}(Y_{t}; \Pi_t^{(M)}) \le \alpha,\; S_t(\widehat{\pi}; [d_t])=1 \mid [\mathcal D_t]=[d_t], [\hat\pi^M]=[\sigma^M] \right)}{ \mathbb{P} \left( S_t(\widehat{\pi}; [d_t])=1 \mid [\mathcal D_t]=[d_t], [\hat\pi^M]=[\sigma^M] \right)} \\
&\overset{(b)}{=} \frac{\mathbb{P} \left( p_t^{\text{rand}}(\hat\pi;[d_t]) \le \alpha,\; S_t(\widehat{\pi}; [d_t])=1 \mid [\mathcal D_t]=[d_t], [\hat\pi^M]=[\sigma^M] \right)}{ \mathbb{P} \left( S_t(\widehat{\pi}; [d_t])=1 \mid [\mathcal D_t]=[d_t], [\hat\pi^M]=[\sigma^M] \right)} \\
&\overset{(c)}{=} \frac{ \sum_{\sigma^{(m)} \in R_t^M}\mathbb{P} \left\{ \frac{ \sum_{\sigma^{(n)} \in R_t^M} \mathds{1} \left\{ V_t(\sigma^{(n)};[d_t]) < V_t(\sigma^{(m)};[d_t]) \right\}+U_t \cdot \sum_{\sigma^{(n)} \in R_t^M} \mathds{1} \left\{ V_t(\sigma^{(n)};[d_t]) = V_t(\sigma^{(m)};[d_t]) \right\} }{ |R_t^M| } \le \alpha \mid [d_t], [\sigma^M]\right\} }{ |R_t^M| }\\
&\overset{(d)}{=} \alpha.
\end{align*}

Above, step (a), (b), (c) follow the same arguments as those in the proof of $\cC_{\alpha,t}(\Pi_t^{(M)})$, and step (d) follows from a standard result for randomized p-value~\citep[Proposition 2.4]{vovk2005algorithmic}.

Then similarly, by tower property, we conclude~\eqref{eq:proofc}, and thus~\eqref{eq:validity_rand}. 

\subsection{Proof of Theorem~\ref{thm: offline scc}}
\label{proof: offline scc}

For notational simplicity, in this proof we write $n=n_{\mathrm{off}}$. We define the unordered ``bag" of the complete data $[\mathcal{D}_t]=[Z_{-n+1}, \dots, Z_t]$ and its realized values $[d_t]=[z_{-n+1}, \dots, z_t]$. Recall that $\widehat{\pi}$ denote the random permutation corresponding to the observed data, so that $(Z_1, \ldots, Z_t) = (z_{\widehat{\pi}(1)}, \ldots, z_{\widehat{\pi}(t)})$. Let $[\hat{\pi}^M]=[\widehat{\pi},\; \pi^{(1)} \circ \widehat{\pi},\; \ldots,\; \pi^{(M)} \circ \widehat{\pi}]$ denote the unordered set of random permutations over indices $\{-n+1, \dots, t\}$ and $[\sigma^{(0)}, \sigma^{(1)}, \ldots, \sigma^{(M)}]$ denote its realization. We also write $S_t(\pi;[d_t])=\cS_t(z_{\pi(-n+1)}, z_{\pi(2)}, \ldots, x_{\pi(t)})$ and $V_t(\pi;[d_t])=\cV(z_{\pi(-n+1)}, z_{\pi(2)}, \ldots, z_{\pi(t)})$. The selection decision of the observed data thus obeys 
\$
S_t = \cS_t(Z_{-n+1},\dots,Z_{t-1},X_t) =S_t(\widehat{\pi}; [d_t]).
\$

To prove~\eqref{eq:thm offline dtm}, it suffices to show that
\begin{equation}\label{eq:proof off}
    \mathbb{P}\big(p_t^\text{off}(Y_t) \le \alpha \mid [\mathcal D_t]=[d_t], [\hat\pi^M]=[\sigma^M], S_t=1\big)\le \alpha.
\end{equation}

First, we define the reference set of permutations satisfying the selection rule in the unordered set $[\tau^M]$ as
\[
R_{t}^{\text{off}} =  \left\{ \sigma \in [\sigma^M] : S_t(\sigma;[d_t]) = 1 \right\}.
\] 
Here, $R_{t}^{\text{off}}$ is fully determined by $[d_t]$ and $[\sigma^M]$. Similar to the arguments in proof of (b) of Theorem 1, we have that conditional on $[\mathcal D_t]=[d_t]$ and $[\hat\pi^M]=[\sigma^M]$,
\begin{align*}
    \{V_t^{Y_t}(\pi): \pi \in \hat{R}_t^\text{off}(Y_t)\}
    &= \{V_t(\pi;[d_t]): \pi \in R_t^\text{off}\}.
\end{align*}

With this claim, we have
\begin{align}
p_t^\text{off}(Y_t)=p_t^\text{off}(\hat\pi;[d_t])\label{eq: off yt to pihat}, \quad \text{where} \quad p_t^\text{off}(\pi;[d_t])=\frac{\sum_{\pi' \in R_t^\text{off}}\mathds{1}\{V_t(\pi;[d_t]) \leq V_t(\pi' ;[d_t])\}}{\vert R_t^\text{off} \vert}.
\end{align}

Returning to the proof of~\eqref{eq:proof off}, we have
\begin{align*}
&\hspace{1.7em}\mathbb{P}\big(p_t^{\text{off}}(Y_t) \le \alpha \mid [\mathcal D_t]=[d_t], [\hat\pi^M]=[\sigma^M], S_t=1\big)\\
&=\mathbb{P}\big(p_t^{\text{off}}(Y_t) \le \alpha \mid [\mathcal D_t]=[d_t], [\hat\pi^M]=[\sigma^M], S_t(\widehat{\pi}; [d_t])=1\big)\\
&= \frac{\mathbb{P} \left( p_t^{\text{off}}(Y_{t}) \le \alpha,\; S_t(\widehat{\pi}; [d_t])=1 \mid [\mathcal D_t]=[d_t], [\hat\pi^M]=[\sigma^M] \right)}{ \mathbb{P} \left( S_t(\widehat{\pi}; [d_t])=1 \mid [\mathcal D_t]=[d_t], [\hat\pi^M]=[\sigma^M] \right)} \\
&=\frac{\mathbb{P} \left( p_t^{\text{off}}(\hat\pi;[d_t]) \le \alpha,\; S_t(\widehat{\pi}; [d_t])=1 \mid [\mathcal D_t]=[d_t], [\hat\pi^M]=[\sigma^M] \right)}{ \mathbb{P} \left( S_t(\widehat{\pi}; [d_t])=1 \mid [\mathcal D_t]=[d_t], [\hat\pi^M]=[\sigma^M] \right)} \\
&= \frac{\frac{1}{1+M} \sum_{m=0}^M \mathds{1} \left\{p_t^\text{off}(\sigma^{(m)};[d_t]) \le \alpha \right\}S_t(\sigma^{(m)};[d_t]) }{\frac{1}{1+M} \sum_{m=0}^M S_t({\sigma^{(m)}};[d_t]) } \\
&= \frac{1}{ |R_t^\text{off}|}\sum_{\sigma^{(m)} \in R_t^\text{off}} \mathds{1} \left\{ \frac{ \sum_{\sigma^{(n)} \in R_t^\text{off}} \mathds{1} \left\{ V_t(\sigma^{(n)};[d_t]) \le V_t(\sigma^{(m)};[d_t]) \right\} }{ |R_t^\text{off}| } \le \alpha \right\} \\
& \le \alpha.
\end{align*}

Above, each step follows the same idea as proof (b) in Theorem 1.

Then by tower property, we conclude~\eqref{eq:proof off}, and thus~\eqref{eq:thm offline dtm}.

\subsection{Proof of Theorem~\ref{thm: fcr}}
\label{proof: fcr}
First, we recall the construction of the prediction set for FCR control. The taxonomy is defined as $\mathfrak{S} = \{(s_1, \ldots, s_t)\}$, where $s_1 = \mathcal{S}(X_1)$, $s_2 = \mathcal{S}(Z_1, X_2)$, $\dots$, $s_t = \mathcal{S}(Z_1, \ldots, Z_{t-1}, X_t)$. Recall that $\Pi_t^{(M)}=(\pi^{(1)}, \dots, \pi^{(M)} )$ are the random permutations uniformly drawn from the full permutation set over the indices $\{1, \dots, t\}$. The reference set is thus given by
\begin{equation*}
\widehat{R}_t^\text{str}(y) = \{\pi \in \Pi_t^{(M)} : S_1^y(\pi) = s_1,\ \ldots,\ S_t^y(\pi) = s_t\}\cup\{\pi_0\}.
\end{equation*}
Substituting this reference set into our framework, the resulting prediction set is
\begin{equation}\label{eq: fcr prediction set}    \widehat{\mathcal{C}}_{\alpha,t}^\text{str} = \left\{y \in \mathcal{Y} : p_t^\text{str}(y) > \alpha\right\}, \quad \text{where}\quad p_t^\text{str}(y)=\frac{\sum_{\pi \in \widehat{R}_t^\text{str}(y)}\mathds{1}\{V_t^y(\pi_0) \leq V_t^y(\pi)\}}{\vert \widehat{R}_t^\text{str}(y) \vert}.
\end{equation}

Besides, we follow the notations of Appendix~\ref{proof: scc}.  We define the unordered set of the data $[\mathcal{D}_t]=[Z_{1}, \dots, Z_t]$ and its realized values $[d_t]=[z_{1}, \dots, z_t]$. Given any realized value $[d_t]$, we let $\widehat{\pi}$ denote the random permutation corresponding to the observed data, so that $(Z_1, \ldots, Z_t) = (z_{\widehat{\pi}(1)}, \ldots, z_{\widehat{\pi}(t)})$. The selection decision of the observed data thus obeys 
$
S_t =S_t(\widehat{\pi}; [d_t]).
$ Let $[\hat{\pi}^M]=[\widehat{\pi},\; \pi^{(1)} \circ \widehat{\pi},\; \ldots,\; \pi^{(M)} \circ \widehat{\pi}]$ denote the unordered set of the random permutations over the  indices $\{1, \dots, t\}$ and $[\sigma^{(0)}, \sigma^{(1)}, \ldots, \sigma^{(M)}]$ denote its realization. we also write $S_t(\pi;[d_t])=\cS_t(z_{\pi(1)}, z_{\pi(2)}, \ldots, x_{\pi(t)})$ and $V_t(\pi;[d_t])=\cV_t(z_{\pi(1)}, z_{\pi(2)}, \ldots, z_{\pi(t)})$. 

\begin{lemma}
    The prediction set defined in~\eqref{eq: fcr prediction set} obeys strong selection-conditional coverage: for any $t \geq1$ and any $(s_1, \dots, s_{t-1})\in \{0,1\}^{t-1}$, we have
    \begin{equation}\label{eq: strong scc}
        \mathbb{P}(Y_t \in \widehat{\mathcal C}_{\alpha, t}^\text{str}\mid S_1=s_1, \dots, S_{t-1}=s_{t-1},S_t=1)\geq 1-\alpha.
    \end{equation}
\end{lemma}

\begin{proof}
    Following the proof before, to prove~\eqref{eq: strong scc}, it suffices to show that
\begin{equation}\label{eq: strong scc proof}
    \mathbb{P}\big(p_t^\text{str}(Y_t) \le \alpha \mid [\mathcal D_t]=[d_t], [\hat\pi^M]=[\sigma^M], S_1=s_1, \dots, S_{t-1}=s_{t-1},S_t=1\big)\le \alpha.
\end{equation}
First, we define the reference set of permutations satisfying the selection rule in the unordered set $[\sigma^M]$ as
\[
R_t^{\text{str}} = \left\{ \sigma \in [\sigma^M] : S_1(\sigma;[d_t]) = s_1, \dots,S_{t-1}(\sigma;[d_t]) = s_{t-1}, S_t(\sigma;[d_t]) = 1 \right\}.
\] 

Here, $R_t^{\text{str}}$ is fully determined by $[d_t]$ and $[\sigma^M]$. Then, conditional on $[\mathcal D_t]=[d_t]$ and $[\hat\pi^M]=[\sigma^M]$, we have:
\begin{align*}
    &\hspace{1.7em}\{V_t^{Y_t}(\pi): \pi \in \hat{R}_t^\text{str}(Y_t)\}\\
    &= \{V_t(\pi \circ \hat{\pi};[d_t]): \pi \in \hat{R}_t^\text{str}(Y_t)\}\\
    &= \{V_t(\pi \circ \hat{\pi};[d_t]): \pi \in \Pi_t^{(M)}\cup \{\pi_0\}, S_1^{Y_t}(\pi) = s_1,\ \ldots,\ S_{t-1}^{Y_t}(\pi) = s_{t-1}, S_t^{Y_t}(\pi)=1\}\\
    &= \{V_t(\pi \circ \hat{\pi};[d_t]):  \pi \in \Pi_t^{(M)}\cup \{\pi_0\},S_1(\pi \circ \hat{\pi};[d_t])=s_1,\dots, S_{t-1}(\pi \circ \hat{\pi};[d_t])=s_{t-1},S_t(\pi \circ \hat{\pi};[d_t])=1\}\\
    &= \{V_t(\sigma;[d_t]):\sigma \in [\sigma^M],S_1(\sigma;[d_t])=s_1,\dots,S_{t-1}(\sigma;[d_t])=s_{t-1}, S_t(\sigma;[d_t])=1\}\\
    &= \{V_t(\pi;[d_t]): \pi \in R_t^\text{str}\}.
\end{align*}
With this claim, following the same arguments in Proof of (b) in Theorem 1, we have that conditional on $[\mathcal D_t]=[d_t]$ and $[\hat\pi^M]=[\sigma^M]$,
\begin{align}
p_t^\text{str}(Y_t)=p_t^\text{str}(\hat\pi;[d_t])\label{eq: off yt to pihat}, \quad \text{where} \quad p_t^\text{str}(\pi;[d_t])=\frac{\sum_{\pi' \in R_t^\text{str}}\mathds{1}\{V_t(\pi;[d_t]) \leq V_t(\pi' ;[d_t])\}}{\vert R_t^\text{str} \vert}.
\end{align}

Returning to the proof of~\eqref{eq: strong scc proof}, we have
\begin{align*}
&\hspace{1.7em}\mathbb{P}\big(p_t^\text{str}(Y_t) \le \alpha \mid [\mathcal D_t]=[d_t], [\hat\pi^M]=[\sigma^M], S_1=s_1, \dots, S_{t-1}=s_{t-1},S_t=1\big)\\
&=\mathbb{P}\big(p_t^{\text{str}}(Y_t) \le \alpha \mid [\mathcal D_t]=[d_t], [\hat\pi^M]=[\sigma^M],S_1(\widehat{\pi}; [d_t])=s_1, \dots, S_{t-1}(\widehat{\pi}; [d_t])=s_{t-1},S_t(\widehat{\pi}; [d_t])=1\big)\\
&= \frac{\mathbb{P} \left( p_t^{\text{str}}(Y_{t}) \le \alpha,\; S_1(\widehat{\pi}; [d_t])=s_1, \dots,S_t(\widehat{\pi}; [d_t])=1 \mid [\mathcal D_t]=[d_t], [\hat\pi^M]=[\sigma^M] \right)}{ \mathbb{P} \left( S_1(\widehat{\pi}; [d_t])=s_1, \dots,S_t(\widehat{\pi}; [d_t])=1 \mid [\mathcal D_t]=[d_t], [\hat\pi^M]=[\sigma^M] \right)} \\
&= \frac{\mathbb{P} \left( p_t^{\text{str}}(\hat\pi;[d_t]) \le \alpha,\; S_1(\widehat{\pi}; [d_t])=s_1, \dots,S_t(\widehat{\pi}; [d_t])=1 \mid [\mathcal D_t]=[d_t], [\hat\pi^M]=[\sigma^M]  \right)}{ \mathbb{P} \left( S_1(\widehat{\pi}; [d_t])=s_1, \dots,S_t(\widehat{\pi}; [d_t])=1 \mid [\mathcal D_t]=[d_t], [\hat\pi^M]=[\sigma^M]  \right)} \\
&= \frac{\frac{1}{1+M} \sum_{m=0}^M \mathds{1} \left\{p_t^{\text{str}}(\sigma^{(m)};[d_t]) \le \alpha \right\}\ind\left\{S_1(\sigma^{(m)};[d_t]) = s_1, \dots, S_t(\sigma^{(m)};[d_t]) = 1 \right\} }{\frac{1}{1+M} \sum_{m=0}^M \ind\left\{S_1(\sigma^{(m)};[d_t]) = s_1, \dots, S_t(\sigma^{(m)};[d_t]) = 1 \right\} } \\
&= \frac{1}{ |R_t^\text{str}|}\sum_{\sigma^{(m)} \in R_t^\text{str}} \mathds{1} \left\{ \frac{ \sum_{\sigma^{(n)} \in R_t^\text{str}} \mathds{1} \left\{ V_t(\sigma^{(n)};[d_t]) \le V_t(\sigma^{(m)};[d_t]) \right\} }{ |R_t^\text{str}| } \le \alpha \right\} \\
&\le \alpha.
\end{align*}

Then by tower property, we conclude~\eqref{eq: strong scc proof}, and thus~\eqref{eq: strong scc}.
\end{proof}

\noindent \textbf{Proof of Theorem~\ref{thm: fcr}.} 
Similar to the proof of~\citet[Proposition 5.2]{sale2025online}, FCR control follows directly from the strong selection-conditional coverage established above. Suppose at time $T$, we have an arbitrary sequence $(s_1, \dots, s_T)$. Then,

\begin{align*}
&\mathbb{E}\!\left[
\frac{\sum_{t=1}^{T} S_t \mathds{1}\{Y_t \notin \widehat{\mathcal{C}}_{\alpha, t}^{\text{str}}\}}
{\sum_{j=1}^{T} S_j}
\,\middle|\, S_1 = s_1, \ldots, S_T = s_T
\right] \\[6pt]
&= \frac{1}{\sum_{j=1}^{T} s_j}
\mathbb{E}\!\left[
\sum_{t=1}^{T} S_t \mathds{1}\{Y_t \notin \widehat{\mathcal{C}}_{\alpha, t}^{\text{str}}\}
\,\middle|\, S_1 = s_1, \ldots, S_T = s_T
\right] \\[6pt]
&= \frac{1}{\sum_{j=1}^{T} s_j}
\mathbb{E}\!\left[
\sum_{t \le T : s_t = 1}
\mathds{1}\{Y_t \notin \widehat{\mathcal{C}}_{\alpha, t}^{\text{str}}\}
\,\middle|\, S_1 = s_1, \ldots, S_T = s_T
\right] \\[6pt]
&= \frac{1}{\sum_{j=1}^{T} s_j}
\sum_{t \le T : s_t = 1}
\mathbb{P}\!\left(
Y_t \notin \widehat{\mathcal{C}}_{\alpha, t}^{\text{str}}
\,\middle|\, S_1 = s_1, \ldots, S_T = s_T
\right) \\[6pt]
&\overset{(a)}{=}
\frac{1}{\sum_{j=1}^{T} s_j}
\sum_{t \le T : s_t = 1}
\mathbb{P}\!\left(
Y_t \notin \widehat{\mathcal{C}}_{\alpha, t}^{\text{str}}
\,\middle|\, S_1 = s_1, \ldots, S_{t-1} = s_{t-1}, S_t = 1
\right) \\[6pt]
&\overset{(b)}{\le}
\frac{1}{\sum_{j=1}^{T} s_j}
\sum_{t \le T : s_t = 1} \alpha \\[6pt]
&= \alpha.
\end{align*}

Above, step (a) follows from our assumption that the selection decisions after time t are independent of whether $Y_t$ is covered by $\widehat{C}_{\alpha,t}$ given $(S_1, \dots, S_{t})$ and step (b) follows from the definition of strong selection-conditional coverage.

Then by tower property, we conclude the FCR control~\eqref{eq:thm fcr}.

\subsection{Proof of Theorem~\ref{thm: most general scc}}
\label{proof: most general scc}

Let $\mathcal D_j=\mathcal D_{\text{calib}}\cup\{Z_{n+j}\}$ and $\mathcal D_j^c=\{X_{n+l}\}_{l \in [m] \backslash \{j\}}$. Similar to Appendix~\ref{proof: scc}, we denote the unordered set $[\mathcal D_j]=[Z_1, \dots, Z_n, Z_{n+j}]$ and its realized values $[ d_j]=[z_1, \dots, z_n, z_{n+j}]$. Here $z_i=(x_i, y_i)$. Given any realized value $[d_j]$, we let $\widehat{\pi}$ denote the random permutation corresponding to the observed data, so that $(Z_1, \ldots, Z_n, Z_{n+j}) = (z_{\widehat{\pi}(1)}, \ldots, z_{\widehat{\pi}(n)},z_{\widehat{\pi}(n+j)})$. The selection subset of the observed data thus obeys 
\$
\hat{S} = \mathcal{S}(\mathcal{D}_{\mathrm{calib}}, \mathcal{D}_{\mathrm{test}}) =S^{\widehat{\pi}}([d_j]).
\$
Recall that $\Pi_{n+j}^{(M)}=(\pi^{(1)}, \dots, \pi^{(M)} )$ are the random permutations uniformly drawn from the full permutation set over the indices $\{1, \dots, t\}$. Let $[\hat{\pi}^M]=[\widehat{\pi},\; \pi^{(1)} \circ \widehat{\pi},\; \ldots,\; \pi^{(M)} \circ \widehat{\pi}]$ denote the unordered set of the random permutations, and its realized permutations are $[\sigma^M]=[\sigma^{(0)}, \sigma^{(1)}, \ldots, \sigma^{(M)}]$. Then we denote $d_\text{calib}^\pi=(z_{\pi(1)}, \dots, z_{\pi(n)})$ and $d_\text{test}^\pi=(X_{n+1}, \dots, x_{\pi(n+j)}, \dots, X_{n+m})$. We also write the selection subset of permuted data as $S^\pi=\cS(d_\text{calib}^\pi,d_\text{test}^\pi )$ and $V_{n+j}(\pi;[d_j])=\cV(z_{\pi(1)}, \dots, z_{\pi(n)}, z_{\pi(n+j)})$.

To prove~\eqref{eq: thm most general dtm}, it suffices to show that
\begin{equation}\label{eq: proof most general}
    \mathbb{P}\big(p_{n+j}(Y_{n+j}) \le \alpha \mid j \in \widehat{S},[\mathcal D_j]=[d_j],\mathcal D_j^c, [\hat{\pi}^M]=[\sigma^M]\big)\le \alpha.
\end{equation}

We define the reference set of permutations satisfying the selection rule in the unordered set $[\sigma^M]$ as
\[
R_{n+j}^{{M}} = \left\{ \sigma \in [\sigma^M] : j \in S^\sigma \right\}.
\] 
Here, $R_{n+j}^{M}$ is fully determined by $[d_j]$, $\mathcal D_j^c$ and $[\sigma^M]$. Then, conditional on $[\mathcal D_t]=[d_t]$, $\mathcal D_j^c$ and $[\hat\pi^M]=[\sigma^M]$, we have:
\begin{align*}
    \{V_{n+j}^{Y_{n+j}}(\pi): \pi \in \hat{R}_{n+j}^M(Y_{n+j})\}&= \{V_{n+j}(\pi \circ \hat{\pi};[d_j]): \pi \in \hat{R}_{n+j}(Y_{n+j})\}\\
    &= \{V_{n+j}(\pi \circ \hat{\pi};[d_j]): \pi \in \Pi_{n+j}^{(M)}\cup \{\pi_0\}, \ j\in \hat S^\pi(Y_{n+j})\}\\
    &= \{V_{n+j}(\pi \circ \hat{\pi};[d_j]):  \pi \in \Pi_{n+j}^{(M)}\cup \{\pi_0\},\ j \in S^{\pi \circ \hat\pi}\}\\
    &= \{V_{n+j}(\sigma;[d_j]):\sigma \in [\sigma^M],\ j \in S^{\sigma}\}\\
    &= \{V_{n+j}(\pi;[d_j]): \pi \in R_{n+j}^M\}.
\end{align*}

With this claim, following the same arguments in Proof of (b) in Theorem 1, we have that conditional on $[\mathcal D_t]=[d_t]$, $\mathcal D_j^c$ and $[\hat\pi^M]=[\sigma^M]$
\begin{align}
p_{n+j}(Y_{n+j})=p_{n+j}(\hat\pi;[d_j])\label{eq: off yt to pihat}, \quad \text{where} \quad p_{n+j}(\pi;[d_j])=\frac{\sum_{\pi' \in R_{n+j}^M}\mathds{1}\{V_{n+j}(\pi;[d_j]) \leq V_{n+j}(\pi' ;[d_j])\}}{\vert R_{n+j}^M \vert}.
\end{align}

Returning to the proof of~\eqref{eq: proof most general}, we have that
\begin{align*}
&\hspace{1.7em}\mathbb{P}\big(p_{n+j}(Y_{n+j}) \le \alpha \mid j \in \widehat{S},[\mathcal D_j]=[d_j],\mathcal D_j^c, [\hat{\pi}^M]=[\sigma^M]\big)\\
&=\mathbb{P}\big(p_{n+j}(Y_{n+j}) \le \alpha \mid j \in S^{\hat\pi},[\mathcal D_j]=[d_j],\mathcal D_j^c, [\hat{\pi}^M]=[\sigma^M]\big)\\
&= \frac{\mathbb{P}\big(p_{n+j}(Y_{n+j}) \le \alpha , j \in S^{\hat\pi}\mid [\mathcal D_j]=[d_j],\mathcal D_j^c, [\hat{\pi}^M]=[\sigma^M]\big)}{ \mathbb{P} \left( j \in S^{\hat\pi} \mid [\mathcal D_j]=[d_j],\mathcal D_j^c, [\hat{\pi}^M]=[\sigma^M] \right)} \\
&= \frac{\mathbb{P}\big(p_{n+j}(\hat\pi;[d_j]) \le \alpha , j \in S^{\hat\pi}\mid [\mathcal D_j]=[d_j],\mathcal D_j^c, [\hat{\pi}^M]=[\sigma^M]\big)}{ \mathbb{P} \left( j \in S^{\hat\pi} \mid [\mathcal D_j]=[d_j],\mathcal D_j^c, [\hat{\pi}^M]=[\sigma^M] \right)} \\
&= \frac{\frac{1}{1+M} \sum_{m=0}^M \mathds{1} \left\{p_{n+j}(\sigma^{(m)};[d_j]) \le \alpha \right\}\cdot \mathds 1\{j \in S^{\sigma^{(m)}}  \}}{\frac{1}{1+M} \sum_{m=0}^M \mathds 1\{j \in S^{\sigma^{(m)}}  \} } \\
&= \frac{1}{|R_{n+j}^M|} \sum_{\sigma^{(m)} \in R_{n+j}^M} \mathds{1} \left\{ \frac{ \sum_{\sigma^{(n)} \in R_{n+j}^M} \mathds{1} \left\{ V_{n+j}(\sigma^{(n)};[d_j]) \le V_{n+j}(\sigma^{(m)};[d_j]) \right\} }{ |R_{n+j}^M| } \le \alpha \right\} \\
&\le \alpha.
\end{align*}

Then by tower property, we conclude~\eqref{eq: proof most general}, and thus~\eqref{eq: thm most general dtm}.

\subsection{Proof of Proposition~\ref{prop: covariate}}
\label{proof: prop covariate}

First, recall the prediction set in the general framework: $\widehat{\mathcal{C}}_{\alpha,t}(\Pi_t^{(M)})
=\{y\in\mathcal{Y}: p_t(y;\Pi_t^{(M)}) > \alpha\}.$
In the covariate-dependent setting, the definition of the p-value can be further simplified by noting that the selection rule does not depend on the candidate value $y$. Let $\widehat{R}_t=\{\pi\in\Pi_t^{(M)}:\ S_t(\pi)=1\}\cup \{\pi_0\}, \ B_t=\{\pi\in\Pi_t^{(M)}:\ S_t(\pi)=1,\ \pi(t)\neq t\},$
and write \(V_t^y(\pi)=v(X_{\pi(t)},Y_{\pi(t)})\) and $S_t^y(\pi)=S_t(\pi)$.
Then the \(p\)-value can be expanded in a more explicit form as follows:
\[
\begin{aligned}
p_t(y;\Pi_t^{(M)}) 
&= \frac{\sum_{\pi \in \widehat{R}_t}\mathds{1}\{V_t^y(\pi_0) \leq V_t^y(\pi)\}}{\vert \widehat{R}_t \vert}\\
&= \frac{1 + \sum_{\pi \in \Pi_t^{(M)}} \mathds{1}\{V_t^y(\pi_0) \leq V_t^y(\pi)\} S_t(\pi)}
        {1 + \sum_{\pi \in \Pi_t^{(M)}} S_t(\pi)} \\
&= \frac{1 + \sum_{\pi \in \Pi_t^{(M)}} \mathds{1}\{v(X_t, y) \leq v(X_{\pi(t)}, Y_{\pi(t)})\} S_t(\pi)}
        {1 + \sum_{\pi \in \Pi_t^{(M)}} S_t(\pi)} \\
&= \frac{1 + \sum_{\pi \in \Pi_t^{(M)}} \mathds{1}\{ \pi(t) = t \} S_t(\pi)}
        {1 + \sum_{\pi \in \Pi_t^{(M)}} S_t(\pi)}
   + \frac{\sum_{\pi \in \Pi_t^{(M)},\, \pi(t)\neq t}
        \mathds{1}\{ v(X_t, y) \leq v(X_{\pi(t)}, Y_{\pi(t)})\} S_t(\pi)}
        {1 + \sum_{\pi \in \Pi_t^{(M)}} S_t(\pi)} \\
&= \frac{ |\widehat{R}_t| - |B_t|}{ |\widehat{R}_t|} 
   + \frac{\sum_{\pi \in B_t} \mathds{1}\{ v(X_t, y) \leq v(X_{\pi(t)}, Y_{\pi(t)})\}}
          { |\widehat{R}_t|}.
\end{aligned}
\]

The prediction set $\widehat{\mathcal{C}}_{\alpha,t}(\Pi_t^{(M)})$ consists of all $y$ such that $p_t(y;\Pi_t^{(M)}) \geq \alpha$. Expanding the inequality, we obtain

\[
\begin{aligned}
p_t(y;\Pi_t^{(M)})  > \alpha &\Leftrightarrow \frac{\vert \widehat{R}_t \vert-\vert B_t \vert}{\vert \widehat{R}_t \vert} + \frac{\sum_{\pi \in B_t}\ind\{ v(X_t,y) \leq v(X_{\pi(t)},Y_{\pi(t)})\}}{\vert \widehat{R}_t \vert} > \alpha \\
& \Leftrightarrow \sum_{\pi \in B_t}\ind\{ v(X_t,y) \leq v(X_{\pi(t)},Y_{\pi(t)})\} > (\alpha-1)\cdot \vert \widehat{R}_t \vert+\vert B_t \vert\\
& \Leftrightarrow \sum_{\pi \in B_t}\ind\{ v(X_t,y) \leq v(X_{\pi(t)},Y_{\pi(t)})\} \geq \lfloor (\alpha-1)\cdot \vert \widehat{R}_t \vert+\vert B_t \vert \rfloor +1.\\
\end{aligned}
\]

Let $K := \lfloor (\alpha-1)\cdot \vert \widehat{R}_t \vert+\vert B_t \vert \rfloor +1$. Then the above condition requires that $v(X_t, y)$ is no larger than the $(|B_t|-K+1)$-th smallest value among $\{ v(X_{\pi(t)}, Y_{\pi(t)}) : \pi \in B_t \}$. This ordering can be further expressed in terms of a quantile. Formally, the quantile level corresponding to the order statistics is defined as
\[
\begin{aligned}
    \beta &= \frac{|B_t| - K + 1}{|B_t|}\\
    &=\frac{|B_t| - (\lfloor (\alpha-1)\cdot \vert \widehat{R}_t \vert+\vert B_t \vert \rfloor +1) + 1}{|B_t|}\\
    &=\frac{\lceil (1-\alpha) \cdot |\widehat{R}_t| \rceil}{|B_t|}.
\end{aligned}
\]

Therefore, the prediction set can be written as
\[
\widehat{\mathcal{C}}_{\alpha, t}(\Pi_t^{(M)})
   = \left\{ y \in \mathcal{Y} :
        v(X_t, y)
        \le \mathrm{Quantile}\left(\beta;\ 
        \{ v(X_{\pi(t)}, Y_{\pi(t)}) \}_{\pi \in B_t} \cup \{+\infty\}\right)
     \right\}.
\]
This completes the proof.

\textbf{Proof of (b)} Recall that in this setting, the randomized p-value is defined as
\[
\begin{aligned}
&p_t^\text{rand}(y; \Pi_t^{(M)}) \\
&=\frac{U_t \cdot \sum_{\pi \in \hat{R}_t(y; \Pi_t^{(M)})}\mathds{1}\{V_t^y(\pi_0)= V_t^y(\pi)\} +\sum_{\pi \in \hat{R}_t(y; \Pi_t^{(M)})}\mathds{1}\{V_t^y(\pi_0) < V_t^y(\pi)\}  }{|\hat{R}_t(y; \Pi_t^{(M)})|}\\
&= \frac{U_t \cdot \left(1+\sum_{\pi \in \Pi_t^{(M)}}\mathds{1}\{V_t^y(\pi_0)= V_t^y(\pi)\} S_t(\pi)\right)+\sum_{\pi \in \Pi_t^{(M)}}\mathds{1}\{V_t^y(\pi_0) < V_t^y(\pi)\} S_t(\pi)}{1+\sum_{\pi \in \Pi_t^{(M)}}S_t(\pi)} \\
&= \frac{U_t \cdot \left(1+\sum_{\pi \in \Pi_t^{(M)}}\mathds{1}\{v(X_t,y)= v(X_{\pi(t)},Y_{\pi(t)})\} S_t(\pi)\right)+\sum_{\pi \in \Pi_t^{(M)}}\mathds{1}\{v(X_t,y) < v(X_{\pi(t)},Y_{\pi(t)})\} S_t(\pi)}{1+\sum_{\pi \in \Pi_t^{(M)}}S_t(\pi)}.
\end{aligned}
\]
Here, $U_t \sim \mathrm{Unif}[0,1]$. Since that it is possible for $v(X_t,y) = v(X_{\pi(t)},Y_{\pi(t)})$ to occur even when $\pi(t)\neq t$, for notational convenience, we further define
\[
E_R(v(X_t,y)) = \left\{\pi \in \widehat{R}_t : v(X_t,y) = v(X_{\pi(t)},Y_{\pi(t)})\right\}.
\]
where $\widehat{R}_t = \{\pi\in\Pi_t^{(M)}: S_t(\pi)=1\}\cup\{\pi_0\}$, and similarly $B_t = \{\pi\in\Pi_t^{(M)}: S_t(\pi)=1,\ \pi(t)\neq t\}$ as before. The p-value can thus be rewritten as
\[
p_t^\text{rand}(y; \Pi_t^{(M)}) = \frac{U_t \cdot (|E_R(v(X_t,y))|)}{|\widehat{R}_t|} + \frac{\sum_{\pi \in B_t}\mathds{1}\{v(X_t,y) < v(X_{\pi(t)},Y_{\pi(t)})\}}{|\widehat{R}_t|}.
\]

The prediction set $\widehat{\mathcal{C}}^\text{rand}_{\alpha,t}(\Pi_t^{(M)})$ consists of all $y$ such that $p_t^\text{rand}(y; \Pi_t^{(M)}) > \alpha$. Expanding this inequality yields
\[
U_t \cdot (|E_R(v(X_t,y))|) + \sum_{\pi \in B_t}\mathds{1}\{v(X_t,y) < v(X_{\pi(t)},Y_{\pi(t)})\} > \alpha\cdot |\widehat{R}_t|.
\]

Let $v_{(1)} \leq v_{(2)} \leq \cdots \leq v_{(|B_t|)}$ denote the sorted values of $\{v(X_{\pi(t)}, Y_{\pi(t)})\}_{\pi \in B_t}$. Following the proof of (a), we can find that the quantile threshold used in the prediction set of deterministic p-value corresponds to an upper quantile among these ordered values in this setting. Here, we still define 
$K:= \lfloor (\alpha-1)\cdot \vert \widehat{R}_t \vert+\vert B_t \vert \rfloor +1$ and the upper quantile can be denoted by $v_{(|B_t|-K+1)}$. For convenience, we denote $r^* = \min\left\{i: v_{(i)} = v_{(|B_t|-K+1)} \right\}=\min\{ i : v_{(i)} = v_{(\lceil (1-\alpha)\cdot |\widehat{R}_t| \rceil)} \},$
which is the minimum index among all scores equal to the upper quantile. We also set
\[
e_R^* = \#\{\pi \in \widehat{R}_t: v(X_{\pi(t)},Y_{\pi(t)}) = v_{(r^*)}\}, \quad e_B^* = \#\{\pi \in B_t: v(X_{\pi(t)},Y_{\pi(t)}) = v_{(r^*)}\}.
\]

Based on the definition of $v_{(r^*)}$, we can find from the inequality above that if $v(X_t, y) > v_{(r^*)}$, then $p_t(y) < \alpha$ for all $U_t$; if $v(X_t, y) < v_{(r^*)}$, then $p_t(y) > \alpha$ for all $U_t$; and if $v(X_t, y) = v_{(r^*)}$, whether $p_t(y) > \alpha$ depends on the value of $U_t$. In other words, the threshold for constructing the prediction set is given by the upper quantile with a certain probability, and by the lower quantile otherwise. Next, we can compute the exact probability p at which the upper quantile is selected by directly solving the inequality for $U_t$.

Plugging $v(X_t, y) = v_{(r^*)}$ into the threshold yields
\[
p_t^\text{rand}(y; \Pi_t^{(M)}) > \alpha 
\iff 
U_t > \frac{\alpha\cdot |\widehat{R}_t| - (|B_t| - r^* + 1 - e_B^*)}{e_R^*} \equiv p,
\]
where $p \in [0,1]$ is the probability with which the upper quantile threshold is taken. Thus, the randomization mechanism determines whether the threshold is the lower or upper quantile. Specifically,
\[
q =
\begin{cases}
\mathrm{Quantile}\left(\frac{r^*-1}{|B_t|};\ \{v(X_{\pi(t)}, Y_{\pi(t)})\}_{\pi \in B_t} \cup \{+\infty\}\right), & \text{if } U_t \leq p\\
\mathrm{Quantile}\left(\frac{r^*}{|B_t|};\ \{v(X_{\pi(t)}, Y_{\pi(t)})\}_{\pi \in B_t} \cup \{+\infty\}\right), & \text{if } U_t > p
\end{cases}
\]
or, equivalently,
\[
q \sim
\begin{cases}
\mathrm{Quantile}\left(\frac{r^*-1}{|B_t|};\ \{v(X_{\pi(t)}, Y_{\pi(t)})\}_{\pi \in B_t} \cup \{+\infty\}\right), & \text{with probability } p \\
\mathrm{Quantile}\left(\frac{r^*}{|B_t|};\ \{v(X_{\pi(t)}, Y_{\pi(t)})\}_{\pi \in B_t} \cup \{+\infty\}\right), & \text{with probability } 1-p
\end{cases}
\]

Therefore, the final prediction set is given by
\[
\widehat{\mathcal{C}}^\text{rand}_{\alpha,t}(\Pi_t^{(M)})
   = \left\{ y \in \mathcal{Y} :
        v(X_t, y)
        \le q
     \right\}.
\]
This completes the proof.

\subsection{Proof of Proposition~\ref{prop: pvalue threshold}}
\label{proof: prop pvalue}

We consider the cases $y \leq c_t$ and $y > c_t$ separately.  
The key step is to show that for any $y \le c_t$, the corresponding reference set 
$\widehat{R}_t(y)$ equals $\widehat{R}_t^{1}$, and for any $y > c_t$, 
$\widehat{R}_t(y)$ equals $\widehat{R}_t^{0}$. 
In this way, we can invert each case separately into the form of a prediction 
interval and then merge the two results to obtain the final prediction set.

%Recall that $j$ be the unique integer such that $\pi(j)=t$, so that $\widehat F_{\pi(j)} = \widehat F_{t}$.

\textbf{Case 1: $y\leq c_t$}. In this case, for any $\pi\in \Pi_t^{(M)}$, it holds that 
\begin{align}
    p_{t,\pi}^w(y) &= \frac{w_t + \sum_{i=1}^{t-1} w_i \cdot \ind\{\widehat{F}_{\pi(i)} \geq \widehat{F}_{\pi(t)} , Y_{\pi(i)}\leq c_{\pi(i)}\}}{\sum_{i=1}^t w_i}\\
    &=\frac{w_t+ \sum_{1\le i \le t-1, i\neq \pi^{-1}(t)}
    w_i\cdot\mathds{1}\bigl\{\widehat F_{\pi(i)}\ge\widehat F_{\pi(t)},\;Y_{\pi(i)}\le c_{\pi(i)}\bigr\}+  w_{\pi^{-1}(t)}\cdot \mathds{1}\bigl\{\widehat F_{t}\ge\widehat F_{\pi(t)}, Y_{t}\le c_{t}\bigr\}}{\sum_{i=1}^t w_i}\\
    %&=\frac{w_t+ \sum_{\substack{i=1\\i\neq j}}^{t-1}w_i\cdot\mathds{1}\bigl\{\widehat F_{\pi(i)}\ge\widehat F_{\pi(t)},\;Y_{\pi(i)}\le c_{\pi(i)}\bigr\}+  w_j\cdot \mathds{1}\bigl\{\widehat F_{\pi(j)}\ge\widehat F_{\pi(t)}, Y_{\pi(j)}\le c_{\pi(j)}\bigr\}}{\sum_{i=1}^t w_i}\\
    % &=\frac{w_t+ \sum_{\substack{i=1\\i\neq j}}^{t-1}
    % w_i\cdot\mathds{1}\bigl\{\widehat F_{\pi(i)}\ge\widehat F_{\pi(t)},\;Y_{\pi(i)}\le c_{\pi(i)}\bigr\}+  w_j\cdot \mathds{1}\bigl\{\widehat F_{t}\ge\widehat F_{\pi(t)}, y\le c_t\bigr\}}{\sum_{i=1}^t w_i}\\
    &=\frac{w_t+ \sum_{1\le i \le t-1, i\neq \pi^{-1}(t)}
    w_i\cdot\mathds{1}\bigl\{\widehat F_{\pi(i)}\ge\widehat F_{\pi(t)},\;Y_{\pi(i)}\le c_{\pi(i)}\bigr\}+  w_{\pi^{-1}(t)}\cdot \mathds{1}\bigl\{\widehat F_{t}\ge\widehat F_{\pi(t)}, y\le c_{t}\bigr\}}{\sum_{i=1}^t w_i}\\
    % &=\frac{w_t+ \sum_{\substack{i=1\\i\neq j}}^{t-1}
    % w_i\cdot\mathds{1}\bigl\{\widehat F_{\pi(i)}\ge\widehat F_{\pi(t)},\;Y_{\pi(i)}\le c_{\pi(i)}\bigr\}+  w_j\cdot \mathds{1}\bigl\{\widehat F_{t}\ge\widehat F_{\pi(t)}\bigr\}}{\sum_{i=1}^t w_i}\\
    &=\frac{w_t+ \sum_{1\le i \le t-1, i\neq \pi^{-1}(t)}
    w_i\cdot\mathds{1}\bigl\{\widehat F_{\pi(i)}\ge\widehat F_{\pi(t)},\;Y_{\pi(i)}\le c_{\pi(i)}\bigr\}+  w_{\pi^{-1}(t)}\cdot \mathds{1}\bigl\{\widehat F_{t}\ge\widehat F_{\pi(t)}\bigr\}}{\sum_{i=1}^t w_i}\\
    &=p_{t,\pi}^{w,1}. 
\end{align}

Since existing methods such as \textsc{LOND}~\citep{javanmard2015onlinecontrolfalsediscovery}, \textsc{SAFFRON}~\citep{ramdas2019saffronadaptivealgorithmonline}, and \textsc{ADDIS}~\citep{tian2019addisadaptivediscardingalgorithm} only utilize previous p-values and thresholds in the historical information when computing the adaptive threshold, it follows from the above p-value formulation that 
$\alpha_{t,\pi}(y) = \alpha_{t,\pi}^1$ when $y \le c_t$.

Therefore, for any $y \leq c_t$, we have 
\$
    \widehat{R}_t(y)& =\{\pi \in \Pi_t^{(M)}: S_t^y(\pi)=1\}\cup \{\pi_0\}=\{\pi \in \Pi_t^{(M)}: p_{t,\pi}^w(y)\leq \alpha_{t,\pi}(y)\}\cup \{\pi_0\} \\ 
    &=\{\pi \in \Pi_t^{(M)}: p_{t,\pi}^{w,1}\leq \alpha_{t,\pi}^1\}\cup \{\pi_0\}=\widehat{R}_t^1.
\$

\textbf{Case 2: $y > c_t$ }. In this case, for any $\pi\in \Pi_t^{(M)}$, it holds that 
\begin{align}
    p_{t,\pi}^w(y) &= \frac{w_t + \sum_{i=1}^{t-1} w_i \cdot \ind\{\widehat{F}_{\pi(i)} \geq \widehat{F}_{\pi(t)} , Y_{\pi(i)}\leq c_{\pi(i)}\}}{\sum_{i=1}^t w_i}\\
    % &=\frac{w_t+ \sum_{\substack{i=1\\i\neq j}}^{t-1}
    % w_i\cdot\mathds{1}\bigl\{\widehat F_{\pi(i)}\ge\widehat F_{\pi(t)},\;Y_{\pi(i)}\le c_{\pi(i)}\bigr\}+ w_j\cdot \mathds{1}\bigl\{\widehat F_{\pi(j)}\ge\widehat F_{\pi(t)}, Y_{\pi(j)}\le c_{\pi(j)}\bigr\}}{\sum_{i=1}^t w_i}\\
    &=\frac{w_t+ \sum_{1\le i \le t-1, i\neq \pi^{-1}(t)}
    w_i\cdot\mathds{1}\bigl\{\widehat F_{\pi(i)}\ge\widehat F_{\pi(t)},\;Y_{\pi(i)}\le c_{\pi(i)}\bigr\}+  w_{\pi^{-1}(t)}\cdot \mathds{1}\bigl\{\widehat F_{t}\ge\widehat F_{\pi(t)}, Y_{t}\le c_{t}\bigr\}}{\sum_{i=1}^t w_i}\\
    % &=\frac{w_t+ \sum_{\substack{i=1\\i\neq j}}^{t-1}
    % w_i\cdot\mathds{1}\bigl\{\widehat F_{\pi(i)}\ge\widehat F_{\pi(t)},\;Y_{\pi(i)}\le c_{\pi(i)}\bigr\}+ w_j\cdot \mathds{1}\bigl\{\widehat F_{t}\ge\widehat F_{\pi(t)}, y\le c_t\bigr\}}{\sum_{i=1}^t w_i}\\
    &=\frac{w_t+ \sum_{1\le i \le t-1, i\neq \pi^{-1}(t)}
    w_i\cdot\mathds{1}\bigl\{\widehat F_{\pi(i)}\ge\widehat F_{\pi(t)},\;Y_{\pi(i)}\le c_{\pi(i)}\bigr\}+  w_{\pi^{-1}(t)}\cdot \mathds{1}\bigl\{\widehat F_{t}\ge\widehat F_{\pi(t)}, y\le c_{t}\bigr\}}{\sum_{i=1}^t w_i}\\
    % &=\frac{w_t+ \sum_{\substack{i=1\\i\neq j}}^{t-1}
    % w_i\cdot\mathds{1}\bigl\{\widehat F_{\pi(i)}\ge\widehat F_{\pi(t)},\;Y_{\pi(i)}\le c_{\pi(i)}\bigr\}}{\sum_{i=1}^t w_i}\\
    &=\frac{w_t+ \sum_{1\le i \le t-1, i\neq \pi^{-1}(t)}
    w_i\cdot\mathds{1}\bigl\{\widehat F_{\pi(i)}\ge\widehat F_{\pi(t)},\;Y_{\pi(i)}\le c_{\pi(i)}\bigr\}}{\sum_{i=1}^t w_i}\\
    &=p_{t,\pi}^{w,0}. 
\end{align}

Similarly, we have $\alpha_{t,\pi}(y) = \alpha_{t,\pi}^0$ when $y > c_t$.
Therefore, for any $y > c_t$, we have 
\$
    \widehat{R}_t(y)&=\{\pi \in \Pi_t^{(M)}: S_t^y(\pi)=1\}\cup \{\pi_0\}=\{\pi \in \Pi_t^{(M)}: p_{t,\pi}^w(y)\leq \alpha_{t,\pi}(y)\}\cup \{\pi_0\} \\ 
    &=\{\pi \in \Pi_t^{(M)}: p_{t,\pi}^{w,0}\leq \alpha_{t,\pi}^0\}\cup \{\pi_0\}=\widehat{R}_t^0.
\$

As a result, for all $y \leq c_t$, the reference set $\widehat{R}_t(y)$ is identical and equals $\widehat{R}_t^1$, and for all $y > c_t$, the reference set is also identical and equals $\widehat{R}_t^0$. Following the argument in Proposition~\ref{prop: covariate}, we can thus invert the procedure and obtain two prediction subsets corresponding to $y \leq c_t$ and $y > c_t$ respectively, and merge them to form the final prediction set~\eqref{eq: pvalue prediction set}. This completes the proof.

\subsection{Proof of Proposition~\ref{prop: elond threshold}}
\label{proof: prop evalue}

Similar to Appendix~\ref{proof: prop pvalue}, we separately consider the cases $y \leq c_t$ and 
$y > c_t$. We will show that for any $y \leq c_t$, the reference set satisfies 
$\widehat{R}_t(y) = \widehat{R}_t^1$, and for any $y > c_t$, the reference set 
satisfies $\widehat{R}_t(y) = \widehat{R}_t^0$.

%Recall that $j$ be the unique integer such that $\pi(j)=t$, so that $\widehat F_{\pi(j)} = \widehat F_{t}$.

\textbf{Case 1: $y\leq c_t$}. In this case, for any $\pi\in \Pi_t^{(M)}$, it holds that 
\begin{align}
    &\hspace{1.7em}p_{t,\pi}(y)\\
    &=\frac{1+\sum_{i=-n+1}^{0} \ind\{\widehat F_{\pi(i)} \geq\widehat F_{\pi(t)}, Y_{\pi(i)}\leq c_{\pi(i)}\}}
        {n+1}\\
    &=\frac{
      1
      + \sum_{-n+1\leq i\leq 0,i\neq \pi^{-1}(t)}
        \mathds{1}\bigl\{\widehat F_{\pi(i)}\ge\widehat F_{\pi(t)},\;Y_{\pi(i)}\le c_{\pi(i)}\bigr\}
      + 
        \mathds{1}\bigl\{\widehat F_{t}\ge\widehat F_{\pi(t)},\;Y_{t}\le c_{t}\bigr\} \cdot \ind{\bigl\{\pi^{-1}(t) \leq 0\bigr\}}
    }{
      n+1
    }\\
    % &=\frac{
    %   1
    %   + \sum_{\substack{i=-n+1\\i\neq j}}^{0}
    %     \mathds{1}\bigl\{\widehat F_{\pi(i)}\ge\widehat F_{\pi(t)},\;Y_{\pi(i)}\le c_{\pi(i)}\bigr\}
    %   + 
    %     \mathds{1}\bigl\{\widehat F_{\pi(j)}\ge\widehat F_{\pi(t)},\;Y_{\pi(j)}\le c_{\pi(j)}\bigr\} \cdot \ind{\bigl\{j \leq 0\bigr\}}
    % }{
    %   n+1
    % }\\
    % &=\frac{
    %   1
    %   + \sum_{\substack{i=-n+1\\i\neq j}}^{0}
    %     \mathds{1}\bigl\{\widehat F_{\pi(i)}\ge\widehat F_{\pi(t)},\;Y_{\pi(i)}\le c_{\pi(i)}\bigr\}
    %   + 
    %     \mathds{1}\bigl\{\widehat F_{\pi(j)}\ge\widehat F_{\pi(t)},\;y\le c_{t}\bigr\} \cdot \ind{\bigl\{j \leq 0\bigr\}}
    % }{
    %   n+1
    % }\\
    &=\frac{
      1
      + \sum_{-n+1\leq i\leq 0,i\neq \pi^{-1}(t)}
        \mathds{1}\bigl\{\widehat F_{\pi(i)}\ge\widehat F_{\pi(t)},\;Y_{\pi(i)}\le c_{\pi(i)}\bigr\}
      + 
        \mathds{1}\bigl\{\widehat F_{t}\ge\widehat F_{\pi(t)},\;y\le c_{t}\bigr\} \cdot \ind{\bigl\{\pi^{-1}(t) \leq 0\bigr\}}
    }{
      n+1
    }\\ 
    &=\frac{
      1
      + \sum_{-n+1\leq i\leq 0,i\neq \pi^{-1}(t)}
        \mathds{1}\bigl\{\widehat F_{\pi(i)}\ge\widehat F_{\pi(t)},\;Y_{\pi(i)}\le c_{\pi(i)}\bigr\}
      + 
        \mathds{1}\bigl\{\widehat F_{t}\ge\widehat F_{\pi(t)}\bigr\} \cdot \ind{\bigl\{\pi^{-1}(t) \leq 0\bigr\}}
    }{
      n+1
    }\\
    &=p_{t,\pi}^1.
\end{align}

Following the same arguments, we can express $p_{t,\pi}^+(y)$ and $p_{t,\pi}^-(y)$ as 
$p_{t,\pi}^{1,+}$ and $p_{t,\pi}^{1,-}$ when $y \le c_t$. 
Furthermore, by applying these p-values within the \textsc{LOND} algorithm, 
we can compute the corresponding adaptive thresholds 
$\widehat{\alpha}_{t,\pi}^{1}$, $\widehat{\alpha}_{t,\pi}^{1,-}$, 
and $\widehat{\alpha}_{t,\pi}^{1,+}$.

Therefore, for any $y\le c_t$, the e-value and the reference set can be defined by
    \[
    E_{t,\pi}(y)
    =\frac{\ind\!\left\{p_{t,\pi}(y)\le \widehat{\alpha}_t^{\text{LOND},+}(y)\right\}}{\widehat{\alpha}_t^{\text{LOND},-}(y)}=\frac{\ind\!\left\{p_{t,\pi}^1\le \widehat{\alpha}_t^{1,+}\right\}}
        {\widehat{\alpha}_t^{1,-}}=E_{t,\pi}^{1},
    \]
    \[
     \widehat{R}_t(y)=\bigl\{\pi\in\Pi_t^{(M)} : E_{t,\pi}(y) \geq \frac{1}{\widehat{\alpha}_{t,\pi}(y)} \bigr\}\cup\{\pi_0\}=\bigl\{\pi\in\Pi_t^{(M)} : E_{t,\pi}^{1} \geq \frac{1}{\widehat{\alpha}_{t,\pi}^{1}} \bigr\}\cup\{\pi_0\}=\widehat{R}_t^{1}.
    \]

\textbf{Case 2: $y> c_t$}. In this case, for any $\pi\in \Pi_t^{(M)}$, it holds that 
\begin{align}
    &\hspace{1.7em}p_{t,\pi}(y)\\
&=\frac{1+\sum_{i=-n+1}^{0} \ind\{\widehat F_{\pi(i)} \geq\widehat F_{\pi(t)}, Y_{\pi(i)}\leq c_{\pi(i)}\}}
        {n+1}\\
    % &=\frac{
    %   1
    %   + \sum_{\substack{i=-n+1\\i\neq j}}^{0}
    %     \mathds{1}\bigl\{\widehat F_{\pi(i)}\ge\widehat F_{\pi(t)},\;Y_{\pi(i)}\le c_{\pi(i)}\bigr\}
    %   + 
    %     \mathds{1}\bigl\{\widehat F_{\pi(j)}\ge\widehat F_{\pi(t)},\;Y_{\pi(j)}\le c_{\pi(j)}\bigr\} \cdot \ind{\bigl\{j \leq 0\bigr\}}
    % }{
    %   n+1
    % }\\
    &=\frac{
      1
      + \sum_{-n+1\leq i\leq 0,i\neq \pi^{-1}(t)}
        \mathds{1}\bigl\{\widehat F_{\pi(i)}\ge\widehat F_{\pi(t)},\;Y_{\pi(i)}\le c_{\pi(i)}\bigr\}
      + 
        \mathds{1}\bigl\{\widehat F_{t}\ge\widehat F_{\pi(t)},\;Y_{t}\le c_{t}\bigr\} \cdot \ind{\bigl\{\pi^{-1}(t) \leq 0\bigr\}}
    }{
      n+1
    }\\ 
    &=\frac{
      1
      + \sum_{-n+1\leq i\leq 0,i\neq \pi^{-1}(t)}
        \mathds{1}\bigl\{\widehat F_{\pi(i)}\ge\widehat F_{\pi(t)},\;Y_{\pi(i)}\le c_{\pi(i)}\bigr\}
      + 
        \mathds{1}\bigl\{\widehat F_{t}\ge\widehat F_{\pi(t)},\;y\le c_{t}\bigr\} \cdot \ind{\bigl\{\pi^{-1}(t) \leq 0\bigr\}}
    }{
      n+1
    }\\
    % &=\frac{
    %   1
    %   + \sum_{\substack{i=-n+1\\i\neq j}}^{0}
    %     \mathds{1}\bigl\{\widehat F_{\pi(i)}\ge\widehat F_{\pi(t)},\;Y_{\pi(i)}\le c_{\pi(i)}\bigr\}
    % }{
    %   n+1
    % }\\
    &=\frac{
      1
      + \sum_{-n+1\leq i\leq 0,i\neq \pi^{-1}(t)}
        \mathds{1}\bigl\{\widehat F_{\pi(i)}\ge\widehat F_{\pi(t)},\;Y_{\pi(i)}\le c_{\pi(i)}\bigr\}
    }{
      n+1
    }\\
    &=p_{t,\pi}^0.
\end{align}

Following the same arguments, we can express $p_{t,\pi}^+(y)$ and $p_{t,\pi}^-(y)$ as 
$p_{t,\pi}^{0,+}$ and $p_{t,\pi}^{0,-}$ when $y > c_t$. 
Furthermore, by applying these p-values within the \textsc{LOND} algorithm, we can compute the corresponding adaptive thresholds 
$\widehat{\alpha}_{t,\pi}^{0}$, $\widehat{\alpha}_{t,\pi}^{0,-}$, 
and $\widehat{\alpha}_{t,\pi}^{0,+}$.

Therefore, for any $y> c_t$, the e-value and the reference set can be defined by
    \[
    E_{t,\pi}(y)
    =\frac{\ind\!\left\{p_{t,\pi}(y)\le \widehat{\alpha}_t^{\text{LOND},+}(y)\right\}}{\widehat{\alpha}_t^{\text{LOND},-}(y)}=\frac{\ind\!\left\{p_{t,\pi}^0\le \widehat{\alpha}_t^{0,+}\right\}}
        {\widehat{\alpha}_t^{0,-}}=E_{t,\pi}^{0},
    \]
    \[
     \widehat{R}_t(y)=\bigl\{\pi\in\Pi_t^{(M)} : E_{t,\pi}(y) \geq \frac{1}{\widehat{\alpha}_{t,\pi}(y)} \bigr\}\cup\{\pi_0\}=\bigl\{\pi\in\Pi_t^{(M)} : E_{t,\pi}^{0} \geq \frac{1}{\widehat{\alpha}_{t,\pi}^{0}} \bigr\}\cup\{\pi_0\}=\widehat{R}_t^{0}.
    \]

As a result, for all $y \leq c_t$, the reference set $\widehat{R}_t(y)$ is identical and equals $\widehat{R}_t^1$, and for all $y > c_t$, the reference set is also identical and equals $\widehat{R}_t^0$. Following the argument in Proposition~\ref{prop: covariate}, we can thus invert the procedure and obtain two prediction subsets corresponding to $y \leq c_t$ and $y > c_t$ respectively, and merge them to form the final prediction set~\eqref{eq: evalue prediction set}. This completes the proof.

\subsection{Proof of Proposition~\ref{prop: earlier outcome partition}}
\label{proof: prop earlier outcomes}

Recall that our selection rule based on earlier outcomes is defined as $\widehat{S}_t=\ind\{\widehat{\mu}(X_t) \geq \text{Weighted quantile}(1-\alpha, \{Y_i\}_{i=1}^{t-1})\}$. For convenience, we rewrite this selection rule to a p-value form:
\[
p_t=\frac{\sum_{i=1}^{t-1}w_i\ind\{Y_{i}\geq\widehat{\mu}(X_t)\}}{\sum_{i=1}^{t-1}w_i}, \quad \widehat{S}_t=\ind\{p_t \le \alpha\}.
\]
Suppose that $y \in I_j=(\mu_{(j-1)}, \mu_{(j)})$. Let $l = \pi^{-1}(t)$ be the unique integer such that $\pi(l)=t$. Then for each permutation $\pi \in \Pi_t^{(M)}$, it holds that
\begin{align}
    p_{\pi}(y)&=\frac{\sum_{i=1}^{t-1}w_i\ind\{Y_{\pi(i)}\geq\widehat{\mu}(X_{\pi(t)})\}}{\sum_{i=1}^{t-1}w_i}\\
    &=\frac{\sum_{\substack{i=1\\i\neq l}}^{t-1}w_i\ind\{Y_{\pi(i)}\geq\widehat{\mu}({X_{\pi(t)}})\}+w_l \cdot \ind\{y\geq\widehat{\mu}(X_{\pi(t)})\}}{\sum_{i=1}^{t-1}w_i}.
\end{align}

Since $\{\widehat{\mu}_{(i)}\}_{i=1}^{t-1}$ is the sorted sequence of predictions, $\widehat{\mu}({X_{\pi(t)}})$ cannot belong to the open interval $(\mu_{(j-1)}, \mu_{(j)})$. Therefore, we will consider two separate cases: $\widehat{\mu}({X_{\pi(t)}})\le\widehat{\mu}_{(j-1)}$ and $\widehat{\mu}({X_{\pi(t)}})\ge\widehat{\mu}_{(j)}$.

\textbf{Case 1: $\widehat{\mu}({X_{\pi(t)}})\le\widehat{\mu}_{(j-1)}$}. Since we have $y \in (\mu_{(j-1)}, \mu_{(j)})$, $y$ is always greater than $\widehat{\mu}({X_{\pi(t)}})$ in this case. Then for any permutation $\pi \in \Pi_t^{(M)}$, it holds that
\[
p_{\pi}(y)=\frac{\sum_{\substack{i=1\\i\neq l}}^{t-1}w_i\ind\{Y_{\pi(i)}\geq\widehat{\mu}({X_{\pi(t)}})\}+w_l \cdot \ind\{y\geq\widehat{\mu}(X_{\pi(t)})\}}{\sum_{i=1}^{t-1}w_i}=\frac{\sum_{\substack{i=1\\i\neq l}}^{t-1}w_i\ind\{Y_{\pi(i)}\geq\widehat{\mu}({X_{\pi(t)}})\}+w_l }{\sum_{i=1}^{t-1}w_i}=p_{\pi}^1
\]

\textbf{Case 2: $\widehat{\mu}({X_{\pi(t)}})\ge \widehat{\mu}_{(j)}$}. Since we have $y \in (\mu_{(j-1)}, \mu_{(j)})$, $y$ is always smaller than $\widehat{\mu}({X_{\pi(t)}})$ in this case. Then for any permutation $\pi \in \Pi_t^{(M)}$, it holds that
\[
p_{\pi}(y)=\frac{\sum_{\substack{i=1\\i\neq l}}^{t-1}w_i\ind\{Y_{\pi(i)}\geq\widehat{\mu}({X_{\pi(t)}})\}+w_l \cdot \ind\{y\geq\widehat{\mu}(X_{\pi(t)})\}}{\sum_{i=1}^{t-1}w_i}=\frac{\sum_{\substack{i=1\\i\neq l}}^{t-1}w_i\ind\{Y_{\pi(i)}\geq\widehat{\mu}({X_{\pi(t)}})\} }{\sum_{i=1}^{t-1}w_i}=p_{\pi}^0
\]

Therefore, for any $y \in (\mu_{(j-1)}, \mu_{(j)})$, we have \begin{align*}
    \widehat{R}_t(y)&=\{\pi \in \Pi_t^{(M)}: S_t^y(\pi)=1\}\cup\{\pi_0\}\\
    &=\{\pi \in \Pi_t^{(M)}: p_{\pi}(y)\leq \alpha\}\cup\{\pi_0\}\\
    &= \left\{\, \pi \in \Pi_t^{(M)} : \widehat{\mu}({X_{\pi(t)}}) \le\widehat{\mu}_{(j-1)},p_\pi^1 \le \alpha \right\}\cup\left\{\, \pi \in \Pi_t^{(M)} : \widehat{\mu}({X_{\pi(t)}}) \ge\widehat{\mu}_{(j)},p_\pi^0 \le \alpha \right\}\cup\{\pi_0\}\\
    &=\widehat{R}_t^{(j)}
\end{align*}

Similarly, for all the open interval $I_j=(\mu_{(j-1)}, \mu_{(j)}),j=1, \dots, t$, we can construct the reference sets $\widehat{R}_t^{(j)}, , j=1, \dots, t$ respectively. Following the argument in Proposition~\ref{prop: covariate}, we can thus invert the procedure and obtain prediction subsets corresponding to $y \in I_j,j=1, \dots, t$. Then, for each boundary point $y_k = \widehat{\mu}_{(k)}$ $(k=1,\dots,t-1)$, we directly impute $y_k$ and include it in the prediction set if $p_t^{\text{rand}}(y_k;\Pi_t^{(M)}) > \alpha$. Finally, after considering all the partitions of $\mathcal Y$, we can merge them all into a final prediction set as~\eqref{eq: final set for earlier outcomes}. This completes the proof.
\end{document}